\documentclass[12pt]{article}

\usepackage{amsmath,amssymb}	
\usepackage{bm}
\usepackage{setspace}	
\usepackage{abstract}	
\usepackage{natbib}		
\usepackage{datetime}
\usepackage{appendix}
\usepackage{amsthm}
\usepackage{graphicx}	
\usepackage{verbatim}
\usepackage{rotating}	
\usepackage{tikz}		
\usepackage{fancyhdr}	
\usepackage[english]{babel}	
\usepackage{enumitem}	
\usepackage{hyperref}
\usepackage{soul}
\usepackage{mathabx}
\newtheorem{theorem}{Theorem}[section]
\newtheorem{corollary}{Corollary}[section]
\newtheorem{lemma}{Lemma}[section]

\newtheorem{assumption}{Assumption}[section]
\newtheorem{remark}{Remark}[section]

\hypersetup{colorlinks=true,urlcolor=blue,linkcolor=blue,citecolor=blue}

\begin{document}

\setcounter{page}{0}

\thispagestyle{empty}

\renewcommand{\thefootnote}{\fnsymbol{footnote}}

\begin{center}
	{\LARGE Structural Periodic Vector Autoregressions 
}
\end{center}

\bigskip

	{\Large \qquad Daniel Dzikowski\footnote{TU Dortmund University,  Department of Statistics, D-44221 Dortmund, Germany; dzikowski@statistik.tu-dortmund.de; corresponding author}
	\qquad  \qquad \qquad  Carsten Jentsch\footnote{TU Dortmund University,  Department of Statistics, D-44221 Dortmund, Germany; jentsch@statistik.tu-dortmund.de}}

	{\large \quad TU Dortmund University
	\qquad  \quad \;  \; \; TU Dortmund University}

\medskip

\begin{center}
	{\today}
\end{center}

\begin{abstract} 
While seasonality inherent to raw macroeconomic data is commonly removed by seasonal adjustment techniques before it is used for structural inference, this may distort valuable information in the data. As an alternative method to commonly used structural vector autoregressions (SVARs) for seasonally adjusted data, we propose to model potential periodicity in seasonally unadjusted (raw) data directly by structural periodic vector autoregressions (SPVARs). This approach does not only allow for periodically time-varying intercepts, but also for periodic autoregressive parameters and innovations variances. As this larger flexibility leads to an increased number of parameters, we propose linearly constrained estimation techniques.  
Moreover, based on SPVARs, we provide two novel identification schemes and propose a general framework for impulse response analyses that allows for direct consideration of seasonal patterns. We provide asymptotic theory for SPVAR estimators 
and impulse responses under flexible linear restrictions and introduce a test for seasonality in impulse responses. For the construction of confidence intervals, we discuss several residual-based (seasonal) bootstrap methods and prove their bootstrap consistency under different assumptions. A real data application shows that useful information about the periodic structure in the data may be lost when relying on common seasonal adjustment methods.
\end{abstract}


\medskip



\newpage

\renewcommand{\thefootnote}{\arabic{footnote}}

\setcounter{footnote}{0}


\section{Introduction}

Macroeconomic time series often exhibit seasonal structure defined as recurrent intra-year movement. The reasons for seasonal fluctuations in economic time series are manifold. Classical literature as, for instance, \cite{hylleberg1992modelling},   \cite{hylleberg1993seasonality},  \cite{canova1995seasonal} and \cite{franses1996recent} argue that seasonality of time series may be caused by calendar and weather effects or by seasonally varying behavior of economic agents, among others. To illustrate, we consider a seasonally unadjusted monthly data set consisting of industrial production (IP), CPI inflation (INF) and the federal funds rate (FFR). 
In Figure \ref{sd_acf_data}, we present estimated spectral densities (SDs) and autocovariance functions (ACFs) for all three (univariate) time series. We find that IP and INF exhibit strong periodic annual patterns, while FFR shows no periodic structure at all. The standard approach to deal with seasonality is to remove the seasonal structure of macroeconomic time series by using seasonal adjustment methods \emph{prior} to a further analysis. The most popular approaches are the classical X-11 seasonal adjustment program established by \cite{shishkin1967x},  the TRAMO/SEATS program by \cite{gomez1996programs} and model-based upgrades X-12-ARIMA \citep{findley1998new},  
X-13-ARIMA-SEATS
\citep{monsell2007x} and Demetra$+$ \cite{eurostat2009guidelines}. However, by removing seasonality in raw macroeconomic data \emph{before} it is used for structural inference, this may distort valuable information in the data. Hence, in the literature, seasonal adjustment methods are often controversial due to their high complexity, lack of linkage to economic theory and the resulting lack of transparency; see, for instance, \cite{gersovitz1978seasonality},  \cite{bell1984issues}, \cite{osborn1989performance}, \cite{franses1996recent} and \cite{mcelroy2022review}. \cite{ghysels1996seasonal} and \cite{ghysels2001econometric} show that seasonal adjustment methods may distort the structure of the data leading to misspecified standard errors and inference. Recently, \cite{doppelt2021should} argue that the results of SVAR analyses are highly dependent on the used seasonal adjustment method and that structural analysis of seasonally unadjusted data should be preferred. For instance, when relying on the common technique of seasonal demeaning, this may not capture the entire periodic structure in the data. This becomes visible in Figure \ref{sd_acf_demean}, where the estimated SDs and ACFs of \emph{seasonally demeaned} IP, INF and FFR are displayed. Apparently, while seasonal demeaning removes almost all periodicity of INF, it is clearly not able to adequately remove the entire periodic structure of IP. That is, even after seasonal demeaning, IP still exhibits pronounced periodic structure. 
While this leftover periodicity is not surprising as only a simple seasonal demeaning is used to get Figure \ref{sd_acf_demean}, even after applying more sophisticated seasonal adjustment methods, the seasonally adjusted data may still display seasonal structure; see \cite{del2004consequences}, \cite{bell2012unit} and \cite{findley2017detecting}.  As an example, \cite{rudebusch2015puzzle} and \cite{lunsford2017lingering} found residual or leftover seasonality in seasonally adjusted gross domestic product estimates by the Bureau of Economic Analysis (BEA).  Even after BEA introduced a new estimation strategy to remove residual seasonality of macro variables in 2018, \cite{consolvo2019residual} and \cite{lunsford2025residual} detected remaining residual seasonality in several adjusted macro variables after BEAs improvements.

 \begin{figure}[t]
\centering
\includegraphics[scale=0.8]{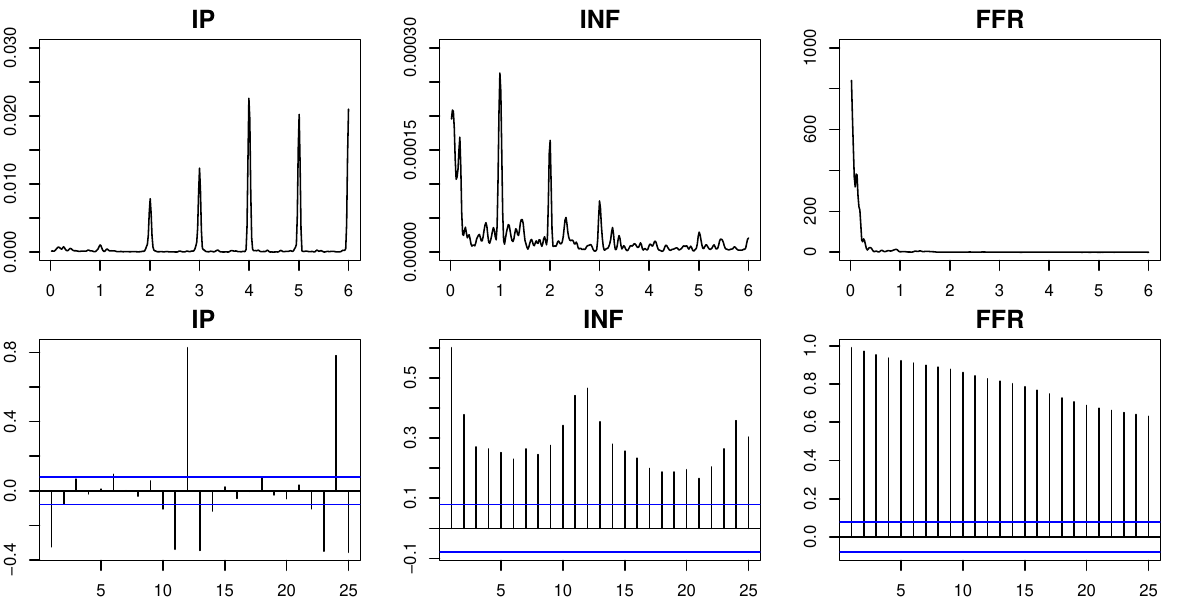}
\caption{Estimated SDs (top panels) and ACFs (bottom panels) of monthly seasonally  unadjusted industrial production (IP), inflation (INF) and federal funds rate (FFR) from 1968-2019. The x-axis of SD plots denotes the normalized frequencies times $S=12$.
}
\label{sd_acf_data}
\end{figure}

In this paper, to address the issue of often highly complex and not transparent seasonal adjustment methods that lack a reasonable economic interpretation, we propose an alternative to the commonly used vector autoregressions (VARs) fitted \emph{after} seasonal adjustment. More precisely, we propose the direct modeling of seasonally unadjusted macroeconomic data based on periodic autoregressions. This approach enables a general framework for structural impulse response analyses that allows to take seasonal patterns directly into account. Periodic autoregressions are useful to capture not only the seasonal structure present in the mean of time series data, but also in their covariances in the sense that both the mean and the (auto)covariances become periodic functions of time. 
Processes with periodically varying covariances, often called periodically correlated stochastic processes, have a long history. First introduced by \cite{gladyshev1961periodically}, pioneering work on (univariate) periodic autoregressive (PAR) processes was done by \cite{jones1967time}, \cite{pagano1978periodic} and \cite{troutman1979some}, who derived inference techniques on periodic autoregressive parameter estimates, predictions using PAR models and the connection between PAR and the related stationary multivariate autoregressive process. PARs for modeling seasonally unadjusted macroeconomic data have been used by \cite{osborn1988seasonality}, \cite{osborn1989performance}, \cite{novales1997forecasting}, \cite{franses2005forecasting} and \cite{ghysels2006forecasting}, who mainly focused on forecasting studies. However, although PAR models have received some attention in economics and certain methodological developments have been made \citep[see e.g.][]{franses2004periodic}, they have not received much popularity for structural analysis. This is mainly because, for analyzing seasonally unadjusted macroeconomic data, a vector-valued extension of the PAR model class is required that naturally comes with a lot of parameters.

\begin{figure}[t]
\centering
\includegraphics[scale=0.8]{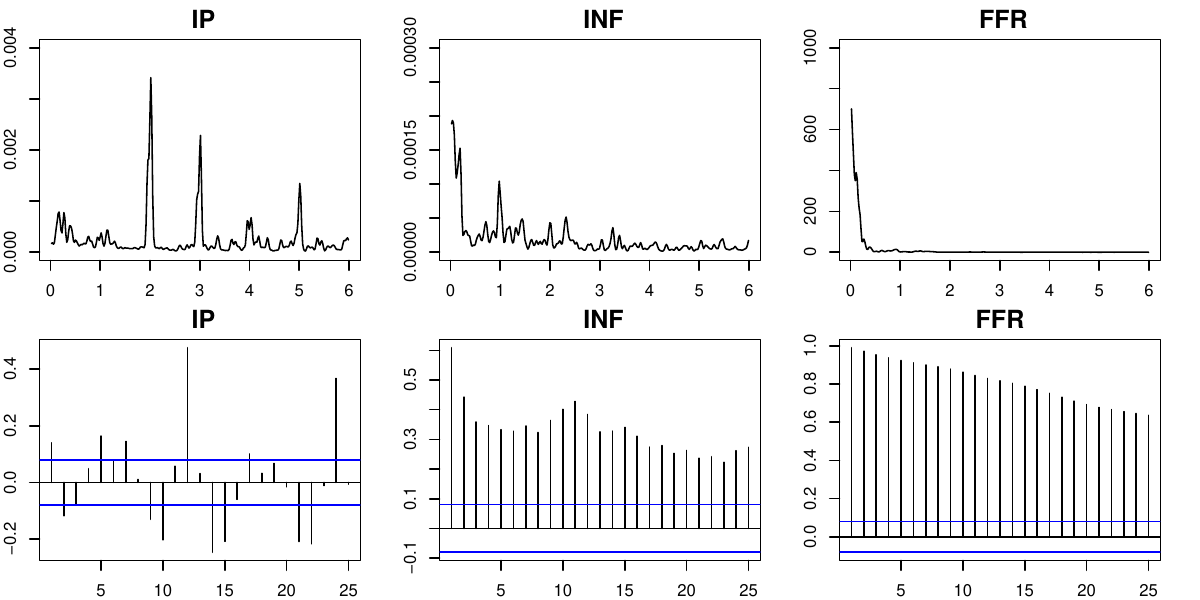}
\caption{Estimated SDs (top panels) and ACFs (bottom panels) of monthly seasonally demeaned industrial production (IP), inflation (INF) and federal funds rate (FFR) from 1968-2019. The x-axis of SD plots denotes the normalized frequencies times $S=12$.
}
\label{sd_acf_demean}
\end{figure}

Asymptotic theory for such periodic vector autoregressive (PVAR) models have been developed by \cite{ula1990periodic,ula1993forecasting}, \cite{franses2004periodic} and \cite{aknouche2007causality}, who derived periodic stationarity conditions and parameter estimation methods. Testing for stationarity against periodic stationary alternatives is considered in \cite{jentsch2012periodic}. Further, P(V)AR models are often considered in stochastic volatility modeling of financial data, see e.g.~\cite{bollerslev1996periodic}, \cite{aknouche2009quasi} and \cite{aknouche2017periodic}. Thanks to their large number of parameters, PVAR models are extremely flexible. However, unconstrained PVAR estimation does often not work well in applications even with moderate sample sizes as typically encountered in macroeconometrics due to pronounced overfitting issues.  To address these, \cite{ursu2009modelling} and \cite{boubacar2023estimating} derived least squares estimators under linear parameter restrictions. However, they only allow to impose linear restrictions on parameters \emph{within} a given season to control the number of parameters to be estimated. In this paper, we extend their approach by explicitly allowing also for practically relevant linear constraints \emph{across} seasons. The importance of constraints across seasons in a PVAR setup is also illustrated in Figures \ref{sd_acf_data} and \ref{sd_acf_demean}. Both figures show that for INF and FFR the autoregressive coefficients may be restricted to be the same across seasons, given evidence that there is hardly any seasonality left after allowing for seasonal mean shifts. This renders the model more parsimonious, facilitating the estimation of the PVAR
model. 

This paper is organized as follows. Section \ref{Section_PVAR} introduces PVAR processes and their properties. 
In Section \ref{Section_Estimation}, we derive (contraint) least-squares estimation for the PVAR parameters that extend the approaches of \cite{ursu2009modelling} and \cite{boubacar2023estimating} by allowing also for practically relevant constraints \emph{across seasons}.
We provide asymptotic theory for the (constrained) estimators of the PVAR parameters under (possibly weak) periodic white noise assumptions. Section \ref{Section_SPVAR} introduces the structural PVAR model, which enables a direct structural analysis of periodic time series without prior seasonal adjustment. We discuss two identification schemes for SPVAR models and derive asymptotic theory for corresponding (structural) impulse response analysis. Section \ref{Section_Bootstrapping} proposes suitable residual-based seasonal bootstrap methods for the construction of confidence intervals and provides bootstrap consistency results. Further, a test for seasonality in impulse responses is introduced. 
A real data application on seasonally unadjusted IP, INF and FFR is provided in Section \ref{Section_Data}, while Section \ref{Section_Conclusion} summarizes our findings. All proofs and additional figures are deferred to the appendix.

\section{Periodic Vector Autoregressive Process}\label{Section_PVAR}

In this section, we discuss periodic vector autoregressive (PVAR) processes and their properties. Let $\{ y_t \}_{t\in \mathbb{Z}} = \{ y_{Sn+s} \}_{n\in \mathbb{Z}, s = 1,\dots, S}$ be an $m$-dimensional PVAR process \citep{ursu2009modelling} described by the model equation 
\begin{align}\label{PVAR_model}
y_{Sn+s} =  \nu(s) +  A_{1}(s) y_{Sn+s-1} + \dots +  A_{p(s)}(s) y_{Sn+s-p(s)} + \epsilon_{Sn+s},	\quad	n \in \mathbb{Z},
\end{align}
where $S$ is the number of seasons within one cycle and $y_{Sn+s}$ denotes the observation in season $s$ and cycle $n$. The model order $p(s)$, the $m$-dimensional seasonal intercept  $\nu(s)$, and the ($m \times m$) seasonal autoregressive coefficient matrices $A_{i}(s)$, $i=1,\dots, p(s)$ are allowed to vary across the seasons $s=1,\dots,S$. The innovation process $\{ \epsilon_{Sn+s} \}_{n\in \mathbb{Z}, s = 1,\dots, S}$ is assumed to be an $m$-dimensional \emph{periodic white noise process} with period $S$ and non-singular covariance matrices $E(\epsilon_{Sn+s} \epsilon^{\prime}_{Sn+s})=\Sigma_{\epsilon}(s)$, $s=1\dots, S$ meaning that 
\begin{align}\label{periodic_white_noise}
\epsilon_{Sn+s} = H_0(s) w_{Sn+s},  \quad  n \in \mathbb{Z},
\end{align} 
where  $\{w_{t} \}_{t\in \mathbb{Z}}$ is a white noise process with $E( w_{t} w_{t}^{\prime} ) = \boldsymbol{I}_{m}$ for all $t\in\mathbb{Z}$ 
and the non-singular matrix $H_0(s)$ satisfies $ \Sigma_{\epsilon}(s) = H_0(s)H_0(s)^{\prime}$. 
Here, $\boldsymbol{I}_{m}$ denotes the $m$-dimensional identity matrix and $ \Sigma_{\epsilon}(s)$ is allowed to vary across the seasons $s=1,\dots,S$.  Since the model order $p(s)$ of a PVAR process can vary across seasons as well, we introduce the notation PVAR($\boldsymbol{p}$) with $\boldsymbol{p} = (p(1), \dots, p(S))$ and PVAR($p$), if $p(s) = p$ for all $s=1,\dots,S$. For $S=1$, the PVAR($\boldsymbol{p}$)=PVAR($p$) process collapses to a common VAR($p$) process. However, the PVAR($\boldsymbol{p}$) in \eqref{PVAR_model} can also be represented with constant order $p$, where $p =\max\{p(1), \dots, p(S)\}$ by imposing zero restrictions on $A_{j}(s)$ for $p(s) < j \leq p$ \citep{lund2000recursive}. 

Throughout the paper, we assume that $\{ y_{Sn+s} \}_{n\in \mathbb{Z}, s = 1,\dots, S}$ is  periodically stationary and a causal infinite-order moving-average representation of \eqref{PVAR_model} exists. \cite{franses2004periodic} use the higher-dimensional VAR representation of PVAR($\boldsymbol{p}$) processes to obtain conditions for periodic stationarity and for causality of PVAR processes. These conditions can be found in Appendix C. The moving-average representation of a periodically stationary (and causal) process $\{y_{Sn+s}\}_{n\in \mathbb{Z}, s = 1,\dots, S}$ is given by
\begin{align}\label{PMA_representation}
y_{Sn+s} = \mu(s) +  \sum\limits_{k=0}^{\infty} \Phi_{k}(s)  \epsilon_{Sn+s-k},	\quad	n\in\mathbb{Z}, 
\end{align}
where $\Phi_{k}(s)$ are $(m \times m)$-dimensional matrices consisting of the moving average coefficients and $\mu(s) = E(y_{Sn+s})$, $s=1,\dots, S$.  Due to  periodic stationarity, the moving average coefficient matrices $\Phi_{k}(s)$ are absolutely summable for all $s=1,\dots,S$. Following \cite{ursu2009modelling}, they can be obtained recursively by the relations
\begin{align*}
\Phi_{0}(s) = \boldsymbol{I}_m   \quad \text{and}  \quad \Phi_{k}(s)  = \sum\limits_{j=1}^{k} A_{j}(s) \Phi_{k-j}(s-j) ,	\quad 	k\in\mathbb{N}.
\end{align*}
The moving average coefficient matrices are periodic in the sense that $\Phi_{k}(Sn+s) =  \Phi_{k}(s), \; n \in \mathbb{Z},\; k \in \mathbb{N}_{0}.$
Reduced-form impulse responses are defined as responses of the system $k$ periods after a one-time unit shock at time $Sn+s$ \citep{kilian2017structural}, while moving average coefficient matrices can be interpreted as the effect of a one-time unit shock at time $Sn + s -k$ on $y_{Sn+s}$.  For a standard VAR, the reduced-form impulse responses are given by its moving average coefficient matrices due to the time-invariant structure of VARs. However, in PVARs, this equality does not hold since the shocks' effect 
on the system of variables depends on the season the shock occurs. Hence, in order to get PVAR reduced-form impulse responses, denoted as $\Phi_{k}^{IR}(s)$, we need to shift the seasonal index of periodic moving average coefficient matrices $k$-steps further, that is  
\[  \Phi_{k}^{IR}(s) =  \Phi_{k}(s+k)=  \sum\limits_{j=1}^{k} A_{j}(s+k) \Phi_{k-j}(s+k-j) , \; \; \;k \in \mathbb{N}_{0} . 
\]
For $i,j = 1,\dots, m$, the element in the $i$-th row and $j$-th column of $\Phi_{k}^{IR}(s)$ quantifies the effect of a (one-time) unit shock in the $j$-th component of $\epsilon_{Sn+s}$ on the $i$-th component of $y_{Sn+k}$, while all other components of $\epsilon_{Sn+s}$ are not shocked at all. Please note that the reduced-form errors are typically contemporaneously correlated with each other and that this contemporaneous correlation is not reflected in the reduced-form impulse responses. Section \ref{Section_SPVAR} provides an overview of how this contemporaneous structure can be taken into account.  

\section{Estimation of PVARs and Asymptotic Inference}\label{Section_Estimation}
In this section, we provide asymptotic theory for a least squares estimation approach for the PVAR model parameters, which allows for flexible linear restrictions. In this regard, we extend the setup of \cite{ursu2009modelling} and cover more general forms of linear restrictions, which also allow for constraints across seasons\footnote{In comparison to \cite{ursu2009modelling}, we find a representation of the PVAR($\boldsymbol{p}$) process that allows for imposing general linear restrictions on the parameter vector. \cite{ursu2009modelling} consider each season of the PVAR($\boldsymbol{p}$) process individually and are therefore only able to restrict the seasonal parameters separately such that no constraints across seasons can be imposed. For example, in our approach, the PVAR can be constrained in such a way that all autoregressive parameters are constant across seasons, and thus a VAR is estimated.}. 
 In addition, it should be noted that \emph{unconstrained} PVAR estimation does not work well in most applications due to pronounced overfitting issues. However, in view of Figures \ref{sd_acf_data} and \ref{sd_acf_demean}, it may be often not necessary to fit a full PVAR model, i.e., without any constraints across seasons, to explain the dependence structure of the data.  Consequently, general linear restrictions play a crucial role in our setup. To address these issues, we derived a PVAR setup which is also able to collapse to a non-periodic VAR under suitable linear restrictions.

\subsection{Estimation of PVAR Models}\label{Subsection_Estimation}
In order to estimate the PVAR parameters suitably under general linear restrictions, a multivariate least squares approach is used. We assume that we observe $N$ complete cycles consisting of $S$ seasons, that is, we have $y_1,  \dots ,y_{SN}$. Additionally, for notational convenience, suppose that pre-sample values $y_{s^{\prime}},\dots,y_0$, where $s^{\prime} = \min \{1-p(1), \dots, S-p(S)\}$ are available.
Then, the PVAR($\boldsymbol{p}$) process in \eqref{PVAR_model} can be rewritten as
\begin{align}\label{Linear_Model}
 Z = B X + E, 
\end{align} 
where $Z,  E,  B$ and $X$ are defined as
\begin{align*}
Z &= (y_1,  \dots ,y_{SN}),  \; \; \; \; \; \;
E = (\epsilon_1 , \dots, \epsilon_{SN}), \\
B &= \bigl(\nu(1), A_{1}(1), \dots, A_{p(1)}(1), \nu(2), A_{1}(2), \dots, A_{p(2)}(2),\dots,\nu(S), A_{1}(S), \dots, A_{p(S)}(S)\bigr), \\
X &= \begin{pmatrix}
X_0(1) & & 0 & X_1(1) & & 0 & & X_{N-1}(1) &  & 0  \\   & \ddots & & &  \ddots & & \dots & &  \ddots & \\
0 &  & X_0(S) & 0 &  & X_1(S) & & 0 & & X_{N-1}(S)
\end{pmatrix},
\end{align*}
which are of dimension $(m \times SN)$, $(m \times SN)$, $(m \times \sum_{s=1}^S (mp(s)+1))$ and $(\sum_{s=1}^S (mp(s)+1) \times SN)$, respectively. The periodic autoregressive parameters are included in $B$, while the entries $X_{n}(s)$ of $X$ are $(mp(s)+1)$-dimensional random vectors given by $X_{n}(s) = (1,y_{Sn+s-1}^{\prime}, \dots, y_{Sn+s-p(s)}^{\prime})^{\prime}$.
Vectorizing the model equation \eqref{Linear_Model} yields
\begin{align}\label{Vec_Linear_Model}
 z= vec\{Z\} &=  vec\{B X\} + vec\{E\}  =\{X^{\prime} \otimes I_m\}  \beta + e, 
\end{align}
where $z= (y_1^{\prime} , \dots, y_{SN}^{\prime})^{\prime}$, $\beta = vec\{B\}$ and $e =  vec\{E\} = (\epsilon_1^{\prime} , \dots, \epsilon_{SN}^{\prime})^{\prime}$, respectively. Note that $\beta$ is given by $\beta = (\beta^{\prime}(1), \dots, \beta^{\prime}(S))^{\prime}$ with $\beta(s) = vec\{(\nu(s),A_{1}(s), \dots, A_{p(s)}(s)) \}$ and that the covariance matrix $\Sigma_E $ of $e$ is block-diagonal with $\Sigma_E = I_N \otimes \Sigma_{\epsilon} $, where
$\Sigma_{\epsilon}  =  \text{diag}\bigl(
\Sigma_{\epsilon}(1),\Sigma_{\epsilon}(2) , \dots,  \Sigma_{\epsilon}(S) \bigr) \in \mathbb{R}^{Sm \times Sm}$ is also block-diagonal. 

\cite{ursu2009modelling} and \cite{boubacar2023estimating} use linear restrictions that can be imposed on parameters \emph{within} one season and do not allow for restrictions \emph{across} seasons.  However, restrictions across seasons can play a crucial role, especially in macroeconomic applications, since they allow parameters to be constrained as constant across the seasons. In general, any set of linear constraints imposed on $\beta$ can be represented by
\begin{align}\label{restriction}
   \beta = R\gamma + r, 
\end{align}
where $R$ is a known $(m \sum_{s=1}^S (mp(s)+1) \times M)$-dimensional matrix with rank $M$, $r$ is a $m \sum_{s=1}^S (mp(s)+1)$-dimensional known vector and $\gamma$ is an $M$-dimensional unconstrained parameter vector of interest.  Substituting $ \beta = R\gamma + r$ into  \eqref{Vec_Linear_Model} yields
\begin{align*} 
\bold{z}_r = \{X^{\prime} \otimes \boldsymbol{I}_m\}  R\gamma + e,
\end{align*}
where 
$\bold{z}_r =z -\{X^{\prime} \otimes \boldsymbol{I}_m\} r$. Hence, the least squares estimator $\widehat{\gamma}_{LS}$ of $\gamma$ is defined as the minimizer of the residual sum of squares $RSS(\gamma)$, where
\begin{align*}
RSS(\gamma) =  \sum\limits_{n=0}^{N-1} \sum\limits_{s=1}^{S}\ \epsilon_{Sn+s}^{\prime} \epsilon_{Sn+s}  = e^{\prime} e = [\bold{z}_r-  \{X^{\prime} \otimes \boldsymbol{I}_m)\}  R\gamma]^{\prime}  [\bold{z}_r-  \{X^{\prime} \otimes  \boldsymbol{I}_m)\}  R\gamma].
\end{align*} 
Hence, according to \cite{lutkepohl2005new}, the corresponding least squares estimator is given by 
$\widehat{\gamma} = \left[R^{\prime} \{X X^{\prime} \otimes \boldsymbol{I}_m \} R \right]^{-1} R^{\prime} \{ X \otimes \boldsymbol{I}_m \} \bold{z}_r$
and the restricted least squares estimator $\widehat{\beta}_{res}$ for $\beta$ is obtained by 
replacing $\gamma$ by $\widehat{\gamma}$ in $\beta = R \gamma + r$,  that is
\begin{align}\label{Res_PVAR_Estimator}
\widehat{\beta}_{res} = R \widehat{\gamma} + r=R \left[R^{\prime} \{X X^{\prime} \otimes \boldsymbol{I}_m \} R \right]^{-1} R^{\prime} \{ X \otimes \boldsymbol{I}_m \} \bold{z}_r + r.
\end{align}
The unrestricted least-squares estimator $\widehat{\beta}_{LS}$ is obtained by setting $R = \boldsymbol{I}_{m \sum_{s=1}^S (mp(s)+1)}$ and $r=\boldsymbol{0}$ in \eqref{Res_PVAR_Estimator} leading to
$\widehat{\beta}_{LS} = \left[ \{X X^{\prime} \otimes  \boldsymbol{I}_m \} \right]^{-1} \{ X \otimes  \boldsymbol{I}_m \} z$.
On the other hand, if $p(s) = p$ for all $s=1,\dots,S$ and $R$ and $r$ are defined as $R =\boldsymbol{1}_S\otimes \boldsymbol{I}_{m  (mp+1)} $ and $r = \boldsymbol{0}$, where $\boldsymbol{1}_l$ denotes the $l$-dimensional vector of ones, the seasonal PVAR($p$) estimators $\beta(s)$ are restricted to be the same across seasons.
Once the periodic autoregressive parameters are estimated for all $s=1,\dots,S$, a natural candidate to estimate the periodic covariance matrices $\Sigma_{\epsilon}(s)$ is 
\[ 
\widehat{\Sigma}_{\epsilon}(s) = \frac{1}{N- k(s)} \sum\limits_{n=0}^{N-1} \widehat{\epsilon}_{Sn+s}  \widehat{\epsilon}_{Sn+s}^{\prime},
\] 
where $k(s)= \lfloor M(s)/m \rfloor$, with $M(s)$ denoting the number of freely varying parameters in season $s$. 
In the following, let $\sigma(s) = vech\{ \Sigma_{\epsilon}(s)\}$ and denote by $\widehat{\sigma}(s) = vech\{\widehat{\Sigma}_{\epsilon} (s)\}$ its empirical counterpart, where the $vech$ operator stacks the entries on the lower-triangular part of a square matrix columnwise below each other.  Further, let $\sigma = (\sigma(1)^{\prime}, \dots,  \sigma(S)^{\prime})^{\prime}$ and $\widehat{\sigma} = (\widehat{\sigma}(1)^{\prime}, \dots,  \widehat{\sigma}(S)^{\prime})^{\prime}$, which are both of dimension $S\widetilde{m}$ with $\widetilde{m} = m(m+1)/2$. 

\subsection{Asymptotic Theory for PVAR Estimators}\label{subsection_Asymptotics}

In order to derive the joint asymptotic properties of  $\widehat{\beta}_{res}$ under general linear restrictions and periodic covariance matrix estimators $\widehat{\sigma}$, some notation is introduced. At first, for $s=1\dots, S$ and $k\in\mathbb{N}$, we define $mp(s) \times m$ matrices $C_k(s)=(\Phi_{k-1}^\prime(s-1), \ldots,\Phi_{k-p(s)}^\prime(s-p(s))^{\prime}$,
where $\Phi_{k}(s)$ denote the moving average coefficient matrices in \eqref{PMA_representation}. In addition,  let $L_{mS}$ be the duplication matrix that satisfies 
$
 (vech\{ B_1 \}^{\prime}, \dots, vech\{ B_S \}^{\prime} )^{\prime}  = L_{mS}  (vec\{ B_1 \}^{\prime}, \dots, vec\{ B_S \}^{\prime})^{\prime}  
$
for symmetric ($m\times m$) matrices $B_1, \dots,B_S$.  Further, for $a,b,c \in \mathbb{Z}$ and $s_1,s_2=1,\ldots,S$, we define the ($m^2 \times m)$ and ($m^2 \times m^2$) dimensional matrices 
\begin{align}\label{cumulants1}
&\kappa_{a,b}(s_1,s_2) \;=  E(vec \{ \epsilon_{Sn+s_1} \epsilon_{Sn+s_1-a}^{\prime} \}  \epsilon_{S(n-b)+s_2}^{\prime}), \\ &\tau_{a,b,c}(s_1, s_2) = \begin{cases}   E(vec \{ \epsilon_{Sn+s_1} \epsilon_{Sn+s_1-a}^{\prime} \}  vec \{ \epsilon_{S(n-b)+s_2} \epsilon_{S(n-b)+s_2-c}^{\prime} \}) &         \\ \quad\quad - vec \{\Sigma_{\epsilon}(s_1) \} vec \{\Sigma_{\epsilon}(s_2) \}^{\prime},  & a=c=0 \\ E(vec \{ \epsilon_{Sn+s_1} \epsilon_{Sn+s_1-a}^{\prime} \}  vec \{ \epsilon_{S(n-b)+s_2} \epsilon_{S(n-b)+s_2-c}^{\prime} \}), & \text{ else}  \label{cumulants2} \end{cases}
\end{align}
and the $(Sm^2 \times Sm^2)$ matrices $\tau_{a,b,c}=(\tau_{a,b,c}(s_1,s_2))_{s_1, s_2 = 1, \dots, S}$.
The symbols $\overset{d}{\to}$ and $\overset{p}{\to}$ denote convergence in distribution and in probability, respectively, and $ \mathcal{N}(\boldsymbol{\mu}, \boldsymbol{\Sigma})$ denotes a normal distribution with mean vector $\boldsymbol{\mu}$ and covariance matrix $\boldsymbol{\Sigma}$.


Note that some of the matrices in \eqref{cumulants1} and \eqref{cumulants2} simplify under an independence assumption imposed on the error process. In contrast to this case, also called strong PVAR, we aim to cover so-called weak PVARs as well that do explicitly not make any independence assumption, but allow for an uncorrelated, but possibly dependent periodic white noise process. To enable the derivation of asymptotic theory for weak PVARs, we impose the following assumptions on the process $\{y_{Sn+s}\}_{n \in \mathbb{Z}, s = 1,\dots,S}$.

\begin{assumption}[PVAR with Periodic Weak White Noise]\ \label{Mixing_Assumptions} 
\begin{itemize} 
\item[(i)] Let the true DGP \eqref{PVAR_model} be periodically stationary. 
\item[(ii)] Let $\{\epsilon_{Sn+s} \}_{n \in \mathbb{Z}, s = 1,\dots,S}$ be a periodic white noise process with period $S$ and 
$E(\epsilon_{Sn+s} \epsilon^{\prime}_{Sn+s})=\Sigma_{\epsilon}(s), s=1\dots, S$ with
$\epsilon_{Sn+s} = H_0(s) w_{Sn+s}$,
$n \in \mathbb{Z}$,
where  $\{w_{t} \}_{t \in \mathbb{Z}}$ is a strictly stationary white noise process with $E(w_{t} w_{t}^{\prime} ) = \boldsymbol{I}_{m}$, and non-singular matrix $H_0(s)$ satisfying $ \Sigma_{\epsilon}(s) = H_0(s)H_0(s)^{\prime}$.
\item[($iii)$] The process $\{w_{t} \}_{t \in \mathbb{Z}}$ is $\alpha$-mixing with $\sum_{h=1}^{\infty} h^{\rho-2}(\alpha(h))^{\delta/(2\rho-2+\delta)} < \infty$ for $\rho=4$ and some $\delta > 0$, 
where $\alpha(h) = \sup_{A \in \mathcal{F}_{-\infty}^0 , B \in \mathcal{F}_{h}^\infty} |P(A\cap B)- P(A)P(B)|$, $h\in\mathbb{N}$ denote the $\alpha$-mixing coefficients of $\{w_{t} \}_{t \in \mathbb{Z}}$ with $\mathcal{F}_{-\infty}^0 = \sigma(\dots, w_{-2}, w_{-1},w_0)$ and $\mathcal{F}_{h}^\infty = \sigma(w_{h}, w_{h+1},\dots)$. Further, we assume $E|w_{t}|^{\rho+2\delta}_{\rho+2\delta}<\infty$, where $|.|_{p}$ denotes the entry-wise $p$-norm.
\end{itemize}
\end{assumption}
Instead of the common i.i.d.~assumption for the white noise process $\{w_{t}\}_{t\in\mathbb{Z}}$, we impose a less restrictive $\alpha$-mixing assumption.
As demonstrated in our real-data application in Figure \ref{WWN1} in Section \ref{Section_Data} below, where the (P)VAR residuals are likely to be non-linearly dependent, the derivation of asymptotic theory under this more general assumption of uncorrelated, but possibly dependent innovations, is important in practice. 
Following \cite{jentsch2022asymptotically}, we impose a typical $\alpha$-mixing assumption, which implies absolute summability of (joint) cumulants up to order $4$. 
The resulting periodic weak white noise process $\{\epsilon_{Sn+s} \}_{n\in\mathbb{Z}, s=1,\dots,S}$ cover a large class of uncorrelated, but possibly dependent periodically (strictly) stationary processes and allows, e.g., for periodic conditional heteroscedasticity, see \cite{francq2011asymptotic},  \cite{bibi2016periodic}.
While \cite{boubacar2023estimating} already showed asymptotic normality in the weak PVAR case for least squares estimators with linear restrictions \emph{within} a given season, we generalize their result and enable linearly constrained least squares estimators that also allow for restrictions \emph{across} seasons. Additionally, we provide asymptotic results for covariance estimators. Under Assumption \ref{Mixing_Assumptions}, we obtain a joint CLT of $\widehat{\beta}_{res}$ and $\widehat{\sigma}$.

\begin{theorem}[Joint CLT with Periodic Weak White Noise]\label{theo_joint_clt}\
Under Assumption \ref{Mixing_Assumptions},  we have
\[ 
\sqrt{N} \begin{pmatrix}
\widehat{\beta}_{res} - \beta  \\  \widehat{\sigma} - \sigma \end{pmatrix}  \overset{d}{\to} \mathcal{N} (\boldsymbol{0}, V) , \; \; \; V = \begin{pmatrix}
V^{(1,1)} & V^{(1,2)}  \\ V^{(2,1)} & V^{(2,2) } 
\end{pmatrix},
\]
where the submatrices $V^{(1,1)} , V^{(2,1)}=V^{(1,2)\prime}$ and $V^{(2,2)} $ are of dimension $(m\sum_{s=1}^S(mp(s) +1) \times m\sum_{s=1}^S(mp(s) +1)),  (S\widetilde{m} \times m\sum_{s=1}^S(mp(s) +1)), (m\sum_{s=1}^S(mp(s) +1) \times S\widetilde{m})$ and $(S\widetilde{m} \times S\widetilde{m})$ and are provided in Appendix A.
\end{theorem}
Due to the seasonal structure of the estimators $\widehat{\beta}_{res}$ and $\widehat{\sigma}$, the submatrices of $V$ are composed of further season-wise submatrices leading to a tedious and lengthy notation. In particular, the submatrices of $V$ depend on infinite sums of $\tau_{a,b,c}(s_1,s_2)$ and $\kappa_{a,b}(s_1,s_2)$, $s_1,s_2 = 1, \dots, S$ and $a,b,c \in \mathbb{Z}$ given in \eqref{cumulants1} and \eqref{cumulants2}, which are cumbersome to estimate. The submatrix $V^{(2,1)} = V^{(1,2)\prime}$ is generally not the zero-matrix, meaning that the estimators $\widehat{\beta}_{res}$ and $\widehat{\sigma}$ become asymptotically dependent. We recommend bootstrap methods to be discussed in Section \ref{Section_Bootstrapping} to approximate the limiting variances when conducting inference. If we impose stronger assumptions such as e.g.~martingale difference sequences or i.i.d.~assumptions, the submatrices of $V$ simplify a lot. Nevertheless, as demonstrated in \cite{bruggemann2016inference} and \cite{jentsch2019dynamic}, this will not allow for valid wild bootstrap inference. The simplifications in the i.i.d.~case are discussed in Remark \ref{strongPVAR_Remark} below. 
For the special case of $S=1$, where the weak (restricted) PVAR collapses to a weak (restricted) VAR, we obtain the same results as in \cite{francq2007multivariate} and \cite{bruggemann2016inference}.

\begin{remark}[Relationship to \cite{boubacar2023estimating}]\label{Main_Ursu_Remark}
Under Assumption \ref{Mixing_Assumptions} and if no restrictions across seasons are imposed, that is, we can write the linear restrictions as $\beta(s)=R(s) \gamma(s)+ r(s)$, $s=1,\ldots,S$, 
where $R(s)$ is of dimension $(m^2p(s)+1) \times {M(s)}$ 
and $r(s)$ of dimension $m^2p(s)+1$ for $s = 1,\dots, S$, our asymptotic results cover the asymptotic results stated in Theorem 3.1 of \cite{boubacar2023estimating} as a special case. However, according to our derivations, we do not obtain block diagonality of $V^{(1,1)}$, which is in contrast to the result claimed in \cite{boubacar2023estimating}. $\beta(s_1)$ and $\beta(s_2)$ are typically asymptotically dependent, as $\tau_{a,b,c}(s_1,s_2) \neq 0$ for $s_1 \neq s_2$ and $a,c \in \mathbb{N}$, $b \in \mathbb{Z}$.
\end{remark}



\begin{remark}[PVAR with Periodic Strong White Noise]\label{strongPVAR_Remark}
For the special case of a strong PVAR estimated under general linear constraints, which is also covered by Theorem \ref{theo_joint_clt}, we
impose a periodic strong white noise \citep{ursu2009modelling} assumption on $\{y_{Sn+s}\}_{n \in \mathbb{Z}, s = 1,\dots,S}$. This is achieved by replacing $(ii)$ and $(iii)$ in Assumption \ref{Mixing_Assumptions} by $(ii)^\prime$ and $(iii)^\prime$, where
\begin{itemize} 
\item[$(ii)^{\prime}$] Let $\{\epsilon_{Sn+s}\}_{n \in \mathbb{Z}, s = 1,\dots,S}$ be a periodic white noise process with period $S$ and 
$E(\epsilon_{Sn+s} \epsilon^{\prime}_{Sn+s})=\Sigma_{\epsilon}(s), s=1\dots, S$ with
$\epsilon_{Sn+s} = H_0(s) w_{Sn+s}$,
$n \in \mathbb{Z}$,
where  $\{w_{t} \}_{t \in \mathbb{Z}}$
is an i.i.d.~white noise with $E( w_{t}w_{t}^{\prime} ) = \boldsymbol{I}_{m}$, and the non-singular matrix $H_0(s)$ satisfies $ \Sigma_{\epsilon}(s) = H_0(s)H_0(s)^{\prime}$. 
\item[$(iii)^{\prime}$] Let $E|w_{t}|^4_4<\infty$, where $|.|_{p}$ denotes the entry-wise $p$-norm. 
\end{itemize}
In this case, we obtain a joint CLT of the PVAR estimators $\widehat{\beta}_{res}$ and $\widehat{\sigma}$ with simplified limiting variances. Details can be found in Corollary B.1 
in Appendix B.
\end{remark}

\section{Structural PVAR Analysis}\label{Section_SPVAR}

In this section, we discuss structural analyses in PVAR setups and propose two novel identification methods for structural-form PVAR models (SPVARs). Further, we derive asymptotic properties of the periodic (structural) impulse responses. 

\subsection{SPVAR Identification Methods}

As in the VAR setup, PVAR models do not capture contemporaneous effects between economic variables. This structure remains in the error term as (possibly) periodic contemporaneous correlation. However, in contrast to VAR setups, this periodic contemporaneous correlation in the PVAR error term consists not only of the periodic contemporaneous effects \emph{between} the variables, but also of the periodic \emph{heteroscedastic} effects of the variables. 
This raises the question of whether the periodic heteroscedasticity in the error component should be explained by the SPVAR model or whether it should remain in the structural shocks in the form that structural shocks occur more/less pronounced in some seasons than in others. To compare the periodic dynamics within the data, we generally propose to standardize the structural shocks across seasons, so that the periodic heteroscedastic structures of the variables are also captured by the SPVAR model. There is also the question of whether and to what extent the periodic contemporaneous correlations between the variables are explainable. 
However, as discussed in the following, a complete explanation of these correlations requires a very large number of restrictions, which may be too restrictive in practice.


In view of commonly used identification approaches used for structural VARs, a natural way to identify the PVAR system structurally is to simply generalize the SVAR identification by imposing restrictions for each period. That is, 
we impose the following assumptions on the variance-covariance matrices of the periodic white noise error term 
\begin{align}\label{SPVAR_Scheme1}
\Sigma_{\epsilon}(s) =H_{0}(s) H_{0}(s)^{\prime} \quad \text{for all } s=1,\dots, S, 
\end{align} 
where $H_{0}(s) \in \mathbb{R}^{m \times m}$ is invertible due to the non-singularity of $\Sigma_{\epsilon}(s)$. For each season $s = 1,\dots, S$, the system of equations consists of $m(m+1)/2$ independent equations and $m^2$ unknowns. Thus, \eqref{SPVAR_Scheme1} is uniquely solvable, if at least $m(m-1)/2$ restrictions are imposed on $H_{0}(s)$ for \emph{each} season $s=1,\dots, S$.\footnote{A trivial way to identify structural shocks is defining $H_{0}(s)$ as the lower-triangular Cholesky factor of $\Sigma_{\epsilon}(s)$, $s=1,\dots, S$.  However, this recursive short-run identification method does not necessarily lead to economically interpretable structural shocks.} Transforming the reduced-form PVAR model \eqref{PVAR_model} by left-multiplying with the inverse SPVAR impact matrix $H^{-1}_{0}(s)$
yields
\begin{align}\label{SPVAR_Model1}
H_{0}^{-1}(s) y_{Sn+s} &=  H_{0}^{-1}(s) A_{1}(s) y_{Sn+s-1} + \dots +  H_{0}^{-1}(s) A_{p(s)}(s) y_{Sn+s-p(s)} + w_{Sn+s},
\end{align} 
where $w_{Sn+s} = H_{0}^{-1}(s)\epsilon_{Sn+s}$ denotes the structural error satisfying $\Sigma_w(s) =Var(w_{Sn+s}) = Var(H_{0}^{-1}(s) \epsilon_{Sn+s}) =  \boldsymbol{I}_m$ for all $s=1,\dots,S$.
This means that, in line with Assumption \ref{Mixing_Assumptions} (ii), the structural errors become a strictly stationary white noise process with the identity matrix as variance-covariance matrix. Further, in this setup, the SPVAR model in \eqref{SPVAR_Model1} is indeed able to explain all the periodic contemporaneous effects and the periodic heteroscedasticity of the variables, since the endogenous variables are driven entirely by strictly stationary white noise shocks.

However, in comparison to the SPVAR setup, the usual SVAR identification techniques such as, e.g., short-run, long-run, proxy, sign or inequality restrictions (see \cite{kilian2017structural} for an overview) are always based on $\widebar{\Sigma}_{\epsilon} = 1/S \sum_{s=1}^S \Sigma_{\epsilon}(s)$, which can be interpreted as the pooled (average) version of the variance-covariance matrices $\Sigma_{\epsilon}(s)$ in our periodic setup. Note that, to achieve full identification of \eqref{SPVAR_Scheme1}, identifying restrictions have to be imposed on \emph{each} season $s=1,\dots ,S$. Hence, structural identification techniques on the PVAR based on \eqref{SPVAR_Scheme1} can be more restrictive than the same identification schemes used in an SVAR setup. This is also made clear by the fact that the implicit assumption $Var(w_{Sn+s}) =\Sigma_w(s) =  \boldsymbol{I}_m$ for all $s$ in \eqref{SPVAR_Model1} implies the corresponding assumption in SVARs $ \widebar{\Sigma}_w = 1/S \sum_{s=1}^S \Sigma_{w}(s) = \boldsymbol{I}_m$, while the reverse direction is generally not true. Accordingly, it is questionable whether these harsher restrictions are appropriate in practical issues. However, it must be made clear that they have to be imposed to capture the entire periodic contemporaneous structure of the data by the SPVAR model.

The general idea behind our proposed structural identification methods for PVARs is based on the assumption that periodic contemporaneous structures cancel out \emph{on average} over all seasons. That is, instead of imposing $Var(w_{Sn+s}) =\Sigma_w(s) =  \boldsymbol{I}_m$ for all periods $s=1,\ldots,S$, we restrict (only) the pooled (average) variance-covariance matrix of the structural shocks to a unit matrix. That is, we impose that
\begin{align}\label{SVAR_Ident}
\widebar{\Sigma}_{\epsilon} = H_0 H_0^{\prime}
\end{align} holds. By construction $H_0 \in \mathbb{R}^{m \times m}$ is invertible and can be interpreted as the non-seasonal impact matrix of the structural shocks. Note that \eqref{SVAR_Ident} is the natural counterpart to the SVAR identification problem.

For structural analyses in PVAR setups, we propose two different identification techniques called \emph{full identification} and \emph{approximate in-average identification}, which are discussed in detail in the following. The full identification approach for SPVARs addresses both sources of periodic contemporaneous correlation, that is, the periodic contemporaneous effect as well as the heteroscedasticity of the variables.
However, many restrictions have to be made to achieve full identification. The (approximate) in-average identification approach allows the structural shocks to have possibly fully occupied seasonal variance-covariance matrices, while their pooled variance-covariance matrix is (approximately) a unit matrix. Approximate in-average identification has the advantage over full identification that identification can be achieved with much fewer restrictions. However, it cannot explain periodic contemporaneous effects of the variables and these remain in the structural error term.

\subsubsection{Full Identification of SPVARs}

The most general SPVAR identification problem based on \eqref{SVAR_Ident} that is able to explain the complete periodic contemporaneous correlation of the variables can be formulated as follows
\begin{align}\label{SPVAR_Scheme4}
&\Sigma_{\epsilon}(s) = H_0  W(s) W(s)^{\prime} H_0^{\prime}, \; \; \text{ for all } s = 1, \dots, S  \\ s.t.  \; \;  &\frac{1}{S} \sum_{s=1}^S W(s) W(s)^{\prime} = \boldsymbol{I}_m,  \nonumber
\end{align} 
where $H_0 \in \mathbb{R}^{m \times m}$ and $W(s) \in \mathbb{R}^{m \times m}, s = 1,\dots, S$ are invertible. In \eqref{SPVAR_Scheme4}, we naturally impose $m(m+1)/2$ restrictions by $1/S \sum_{s=1}^S W(s) W(s)^{\prime} = \boldsymbol{I}_m$ to achieve \eqref{SVAR_Ident}. Therefore, $W(s)$ can be interpreted as a seasonal matrix containing all the periodic heteroscedastic and contemporaneous effects around the mean level, while $H_0$ can be interpreted as the non-seasonal impact matrix. As we have $m^2+Sm^2$ parameters from $H_0$ and $W(s), s = 1,\dots, S$, additional $m(m-1)/2 + Sm(m-1)/2$ restrictions on $H_0$ and $W(s), s = 1,\dots,S$ are needed, in order to solve \eqref{SPVAR_Scheme4} uniquely. Using $H_0(s) = H_0 W(s)$ as SPVAR impact matrix yields
\begin{align}\label{SPVAR_Model4}
H_{0}^{-1}(s) y_{Sn+s} &=  H_{0}^{-1}(s) A_{1}(s) y_{Sn+s-1} + \dots +  H_{0}^{-1}(s) A_{p(s)}(s) y_{Sn+s-p(s)} + w_{Sn+s},
\end{align} 
where $\Sigma_w(s) = \boldsymbol{I}_m$ for all $s=1,\dots,S$. This means that both the periodic heteroscedasticity and periodic contemporaneous effects are fully captured by the SPVAR model. However, the issue of this approach is that due to the large number of restrictions that are needed to fully identify \eqref{SPVAR_Scheme4}, there is no way around restricting the periodic contemporaneous structure between structural shocks and variables, i.e., restricting $W(s),s=1,\dots,S$. Since the periodic contemporaneous structure between structural shocks and variables is unclear, it can be questionable to constrain this structure in any way. Trivial identification strategies such as setting $W(s)$ as the Cholesky factor of $W(s)W(s)^{\prime}$ or restricting $W(s)$ to be a symmetric matrix for all seasons $s=1,\dots,S$ combined with imposing $m(m-1)/2$ restrictions on $H_0$ solve the identification problem, but rely on highly questionable assumptions on the periodic contemporaneous structure of the data. For example, a Cholesky identification would assume that there are no periodic contemporaneous effects between some structural shocks and variables, while e.g. the assumption of a symmetric matrix $W(s)$ would mean that the periodic contemporaneous effects between the $j$-th shock and the $i$-th variable and between the $i$-th shock and the $j$-th variable are equal for $i \neq j$, $i,j = 1,\dots, m $.
\\

\begin{remark}[Extended SVAR identification for SPVARs]\label{SPVAR_Ident_SVARgeneralization}  Note that the generalization of SVAR identification discussed in \eqref{SPVAR_Scheme1} is a special case of the identification problem \eqref{SPVAR_Scheme4}. In order to get full identification of $H_0(s)$ of equation \eqref{SPVAR_Scheme1}, one idea would be to impose identification techniques as, e.g., short-run, long-run or sign restrictions on $H_0(s)$ for all seasons $s = 1,\dots, S$. However, these restrictions on $H_0(s)$ implicitly constrain the periodic contemporaneous structure of the structural shocks on the variables and can therefore be questionable in practice. For example, imposing a short-run restriction on the $ij$-th element of $H_0(s)$ for all seasons $s= 1,\dots, S$ implies a short-run restriction on the $ij$-th element of $H_0 = (h_{ij})_{i,j = 1,\dots, m}$ and additional restrictions on $W(s) = (w_{ij}(s))_{i,j = 1,\dots, m}$ that the periodic contemporaneous effects of all other structural shocks $w_l, l \neq j$ cancel out, i.e.
\begin{align*}
    \sum_{l \neq j} h_{il}w_{lj}(s) = 0, \quad \quad \text{for all } s = 1,\dots,S.
\end{align*}
\end{remark}

\subsubsection{Approximate In-Average Identification of SPVARs}\label{Section_Approx_identification}

Since restrictions on the periodic contemporaneous dependence structure between structural shocks and variables can be highly questionable, we propose an alternative identification method for SPVARs that is based on \eqref{SVAR_Ident}. The idea behind this identification method is to identify seasonal structural shocks whose variances are \emph{normalized} for all seasons, but whose contemporaneous effects may vary seasonally, but cancel out \emph{on average} across all seasons. This alternative identification method is not as flexible as \eqref{SPVAR_Scheme4}, but can be identified without restricting the periodic contemporaneous dependence structure. The system of equations can be formulated by 
\begin{align}\label{SPVAR_Scheme3}
&\Sigma_{\epsilon}(s) = H_0 \;\Lambda(s) \; H_0^{\prime}, \; \;  \text{ for all } s = 1, \dots, S  \\ s.t.  \; \;  &\frac{1}{S} \sum_{s=1}^S \Lambda(s) = \boldsymbol{I}_m,  \nonumber
\end{align} 
where $\Lambda(s) \in \mathbb{R}^{m \times m}$ is a symmetric seasonal matrix. 
In total, we have $Sm(m+1)/2$ independent equations and $m^2+Sm(m+1)/2$ parameters from $H_0$ and $\Lambda(s), s = 1,\dots, S$. As $m(m+1)/2$ restrictions are naturally imposed by $1/S \sum_{s=1}^S \Lambda(s) = \boldsymbol{I}_m$ to achieve \eqref{SVAR_Ident}, additional $m(m-1)/2$ restrictions 
are needed (as in SVAR identification), in order to solve \eqref{SPVAR_Scheme3} uniquely\footnote{From a mathematical perspective, restrictions can also be imposed on the seasonal covariance matrix $\Lambda(s)$, but this is not really practicable as such restrictions are highly questionable. Accordingly, in most applications, we restrict the non-seasonal impact matrix $H_0$ exclusively to obtain identification.}. Using $\widetilde{H}_{0}(s) = H_0 \; diag[\Lambda(s)]^{1/2}$ as SPVAR transformations matrix, where $diag[\Lambda(s)]$ is a diagonal matrix containing the diagonal elements of $\Lambda(s)$,
yields
\begin{align}\label{SPVAR_Model3}
\widetilde{H}_{0}^{-1}(s) y_{Sn+s} &=  \widetilde{H}_{0}^{-1}(s) A_{1}(s) y_{Sn+s-1} + \dots +  \widetilde{H}_{0}^{-1}(s) A_{p(s)}(s) y_{Sn+s-p(s)} + \widetilde{w}_{Sn+s}.
\end{align} 
The covariance matrix of the structural shocks is then given by $\Sigma_{\widetilde{w}}(s) = \Lambda_0(s)$ for all $s=1,\dots,S$ with $\Lambda_0(s) = diag[\Lambda(s)]^{-1/2} \; \Lambda(s) \; diag[\Lambda(s)]^{-1/2}$. We introduce the notation of $\widetilde{H}_0(s)$ and $\widetilde{w}_{Sn+s}$, because the structural shocks $\widetilde{w}_{Sn+s}$ do not necessarily correspond to those $w_{Sn+s}$ from Assumption \ref{Mixing_Assumptions}, due to $\Sigma_{\widetilde{w}}(s) = \Lambda_0(s) \neq \boldsymbol{I}_m$. The construction of $\widetilde{H}_0(s)$ ensures that the variances of $\widetilde{w}_{Sn+s}$ are normalized to one for all seasons. Accordingly  $\Lambda_0(s)$ can be described as the seasonal correlation matrix of the structural shocks meaning that periodic heteroscedasticity of the PVAR shocks is completely captured by the SPVAR model \eqref{SPVAR_Model3}. The main advantage of \eqref{SPVAR_Scheme3} compared to \eqref{SPVAR_Scheme4} is that significantly fewer restrictions need to be imposed on $\widetilde{H}_{0}(s)$ in order to identify the SPVAR system.  However, using this approach, periodic contemporaneous effects are explicitly not captured by the SPVAR in \eqref{SPVAR_Model3} and in general, due to the variance normalization of $\widetilde{w}_{Sn+s}$, it is not true that the periodic correlations across the seasons cancel each other out completely. That is, because $1/S \sum_{s=1}^S \Lambda(s) = \boldsymbol{I}_m$ does not necessarily imply $\widebar{\Sigma}_{\widetilde{w}} = 1/S \sum_{s=1}^S \Lambda_0(s) = \boldsymbol{I}_m$. That is, with greater periodic heteroscedasticity and periodic contemporaneous effects, the off-diagonal elements of the pooled variance matrix $\bar{\Sigma}_{\widetilde{w}}$ of the structural shocks in \eqref{SPVAR_Model3} deviate more from zero.
Hence, if the reduced-form PVAR shocks either show weak periodic heteroscedasticity or weak periodic (contemporaneous) correlation, 
the identification method from \eqref{SPVAR_Scheme3} is very suitable for structural identification since $\widebar{\Sigma}_{\widetilde{w}}  \approx \boldsymbol{I}_m$ holds in this case. This means that the contemporaneous structure in the structural shocks is typically negligibly small such that the structural shocks can be considered orthogonal. 

\begin{remark}[Standardization of PVAR Errors]\label{SPVAR_Correlation_Identification} If the reduced-form PVAR shocks are highly periodic heteroscedastic, it can be very useful to first standardize the reduced-form PVAR shocks and then to identify the structural shocks from the standardized PVAR shocks afterwards. This is achieved by multiplying the left-hand side in equation \eqref{SPVAR_Scheme3} by the matrix  $diag[\Sigma_{\epsilon}(s)]^{-1/2}$ from the left and right, respectively, and replacing the SPVAR impact matrix in \eqref{SPVAR_Scheme3} by $\widetilde{H}_{0}(s) = diag[\Sigma_{\epsilon}(s)]^{1/2}\; H_0 \; diag[\Lambda(s)]$. 
\end{remark}

\begin{remark}[Exact In-Average Identification]\label{SPVAR_Seasonal_Variances} 

If only the static, non-seasonal contemporaneous effects of the system of variables are of interest, then the SPVAR transformation matrix from \eqref{SPVAR_Scheme3} can be replaced by $\widetilde{H}_0(s) = H_0$ for $s =1,\dots,S$. This leads to structural shocks $\widetilde{w}_{Sn+s}$ with $\Sigma_{\widetilde{w}}(s) = \Lambda(s)$ for $s=1,\dots,S$ and $\widebar{\Sigma}_{\widetilde{w}} = \boldsymbol{I}_m$ implying that the structural shocks are fully orthogonal in-average. However, this identification approach does neither capture periodic heteroscedastic nor contemporaneous effects of the variables. Note that in-average identification is equivalent to identifying \eqref{SVAR_Ident} directly.
\end{remark}
\begin{remark}[Practical Identification Techniques]\label{SPVAR_Identification_Methods}
To obtain complete identification of scheme \eqref{SPVAR_Scheme3}, identification techniques that only affect the impact matrix (such as short-run or static sign restrictions) can be applied directly to \eqref{SVAR_Ident} to first identify the non-seasonal impact matrix $H_0$ before computing $\widetilde{H}_0(s)$ for $s=1,\dots,S$. By construction, $\widetilde{H}_0(s),s=1,\dots,S$ is fully identified, if $H_0$ is fully identified. Long-run restrictions can be imposed in the sense that the long-run effects apply on average over the seasons meaning that we impose restrictions on the pooled (average) periodic long run impulse response $\widebar{\Theta}^{LR}$ given by
\[ \widebar{\Theta}^{LR} = \frac{1}{S} \sum_{s=1}^S \Biggl[\Bigl(\sum_{k=0}^{\infty}\Phi^{IR}_{k}(s) \Bigr) \widetilde{H}_{0}(s) \Biggl].\]
\end{remark}

\subsection{Structural Impulse Responses}

Now, we suppose that the structural PVAR impact matrix $H_0(s),s=1,\dots,S$ is fully identified. Structural impulse responses at lag $k$ are defined as a reaction of $y_{Sn+s+k}$ in response to a one-time impulse in $w_{Sn+s}$.  Hence,  following \cite{kilian2017structural}, structural PVAR impulse responses for each season $s=1,\dots, S$ are defined as
\[  \Theta^{SIR}_{k}(s) = \frac{\partial y_{Sn+s+k}}{\partial w_{Sn+s}} , \quad k\in\mathbb{N}_0.\]
For $i,j = 1,\dots,m$, the element in the $i$-th row and $j$-th column of $\Theta^{SIR}_{k}(s)$ denotes the structural impulse response of the $i$-th component of $y_{Sn+s+k}$ to a unit shock in the $j$-th component of $w_{Sn+s}$ at season $s$.  The structural impulse responses can also be deduced by transforming the moving average representation of $y_{Sn+s}$ as follows
\begin{align*}
y_{Sn + s} =  \mu(s) +  \sum_{k=0}^{\infty} \Phi_{k}(s)\epsilon_{Sn+s-k} =   \mu(s) + \sum_{k=0}^{\infty} \Theta_{k}(s)w_{Sn+s-k},
 \end{align*}
where $\Theta_{k}(s) = \Phi_{k}(s)H_{0}(s-k)$ are the structural moving average coefficient matrices indicating the effect of $w_{Sn+s-k}$ on $y_{Sn + s}$. Note that $H_{0}(s) = H_{0}(Sn+s)$ for all $n \in \mathbb{Z}$ and $s=1,\dots,S$. 
The structural impulse responses $\Theta^{SIR}_{k}(s)$ can be expressed analytically as
\[ 
\Theta^{SIR}_{k}(s)  = \Theta_{k}(s+k) = \Phi_{k}(s+k)H_{0}(s) =  \Phi^{IR}_{k}(s) H_{0}(s).  
\]
In writing down the asymptotic properties of the impulse responses $\Phi_{k}^{IR}(s)$ and the structural impulse responses $\Theta_{k}^{SIR}(s)$, $s =1,\dots, S$, we use the higher-dimensional VAR representation of the PVAR which is given in Appendix C. It can be easily shown that the reduced-form impulse responses of the higher-dimensional VAR form, denoted by $\Pi_{h}^{IR}$, $h \in \mathbb{N}_0$, in Appendix C, consist of the PVAR impulse responses $\Phi_{k}^{IR}(s)$ for all $k,s$ satisfying $hS < k+s \leq  (h+1)S$. The structural impulse responses of the higher-dimensional VAR form, denoted by $\Psi_{h}^{SIR}$, $h \in \mathbb{N}_0$, consist of $\Theta_{k}^{SIR}(s)$ for all $k,s$ that satisfy $hS < k+s \leq  (h+1)S$. Precisely, we have
\[
\Psi_{h}^{SIR} = \Pi_h^{IR}\boldsymbol{H}_0=    \begin{pmatrix}
 \Theta_{Sh}^{SIR}(1) &  \Theta_{Sh-1}^{SIR}(2) & \dots &  \Theta_{Sh-S+1}^{SIR}(S) \\  \Theta_{Sh+1}^{SIR}(1) &  \Theta_{Sh}^{SIR}(2) & \ddots & \vdots \\ \vdots & \ddots & \ddots &  \Theta_{Sh-1}^{SIR}(S) \\   \Theta_{Sh+S-1}^{SIR}(1) & \dots &  \Theta_{Sh+1}^{SIR}(S-1) &  \Theta_{Sh}^{SIR}(S) \end{pmatrix}, \; \; \; h\in\mathbb{N}_0,
\]
where $\boldsymbol{H}_0 = \text{diag}\bigl( H_0(1),H_0(2),  \dots, H_0(S) \bigr) \in \mathbb{R}^{Sm \times Sm}$. Let $\widehat{\Pi}_{h}^{IR}$ and $\widehat{\Psi}_{h}^{SIR}$ be the empirical counterparts of $\Pi_{h}^{IR}$ and $\Psi_{h}^{SIR}$ for $h\in\mathbb{N}_0$, respectively.  Note that $\Pi_{h}^{IR}$, $h \in \mathbb{N}_0$ can be represented as a continuously differentiable function of the PVAR coefficients $\beta$, while the structural impulse responses $\Psi_{h}^{SIR}$ can be represented as a continuously differentiable function depending on $\beta$ and $\sigma$.
Consequently, from Theorem \ref{theo_joint_clt} and following \cite{lutkepohl2005new}, Proposition 3.6, we immediately get the following result. 

\begin{theorem}[CLT of Impulse Responses]\label{CLT_IR_SIR}\
Under Assumption \ref{Mixing_Assumptions},  we have
\begin{align*}
&\sqrt{N} \bigl( vec \{ \widehat{\Pi}_{h}^{IR} - \Pi_{h}^{IR} \}  \bigr) \overset{d}{\to}  \mathcal{N}\bigl(\boldsymbol{0}, \Sigma_h^{\Pi}  \bigr),  \\
 &\sqrt{N} \bigl( vec \{ \widehat{\Psi}_{h}^{SIR} - \Psi_{h}^{SIR} \}  \bigr) \overset{d}{\to}  \mathcal{N}\bigl(\boldsymbol{0}, \Sigma_h^{\Psi}  \bigr), \; \; \;  h\in\mathbb{N}_0,
\end{align*}
where
\begin{align*}
 \Sigma_h^{\Pi} = G^{\Pi}_h V^{(1,1)} G_h^{\Pi \prime}  \quad   \text{and}  \quad
 \Sigma_h^{\Psi} = F_h V^{(1,1)} F_h^{\prime} + D_h V^{(2,2)} D_h^{\prime} + F_h V^{(2,1)\prime} D_h^{\prime} + D_h V^{(2,1)} F_h^{\prime}
\end{align*}
with $G_{h}  = \frac{\partial vec \{\Pi_{h}^{IR} \} }{\partial \beta^{\prime}}$, $F_{h} = \frac{\partial vec \{ \Psi_{h}^{SIR} \} }{\partial \beta^{\prime}}$ and $D_{h} = \frac{\partial vec \{ \Psi_{h}^{SIR} \} }{\partial \sigma^{\prime}}$. 
\end{theorem}
The closed-form solutions of the derivatives $G_{h}, F_{h}$ and $D_{h}$ are stated in Appendix C. 
In the strong PVAR case, the same statements as in Theorem \ref{CLT_IR_SIR} hold with the matrices $V^{(1,1)}, V^{(2,1)}$ and $V^{(2,2)}$ replaced by $V^{(1,1)}_{iid}, V^{(2,1)}_{iid}$ and $V^{(2,2)}_{iid}$, respectively. The i.i.d. limiting variance matrices can be found in Appendix B. 

\section{Bootstrap Inference in PVARs}\label{Section_Bootstrapping}

Analytical expressions of the limiting variances of the PVAR estimators and the (structural) impulse responses are very complex and cumbersome to estimate. Therefore, we propose residual-based bootstrap methods 
to approximate the (limiting) distributions of the SPVAR estimators 
and prove their bootstrap consistency. We also construct a bootstrap-based test for seasonality in SPVAR impulse responses.

\subsection{Bootstrap Schemes}\label{Subsection_Bootstrapschemes}

In this section, we discuss different residual-based bootstrap methods that are suitable for conducting inference in different PVAR setups. In Section \ref{Subsubsection_Bootstrapscheme_weak}, we describe residual-based bootstrap schemes that are tailored for the assumption of a periodic weak white noise, while Section \ref{Subsubsection_Bootstrapscheme_strong} contains simplified versions that are sufficient for periodic strong white noise.

\subsubsection{Residual-based seasonal block bootstrapping}\label{Subsubsection_Bootstrapscheme_weak}

We describe residual-based seasonal block bootstrap methods for PVAR processes that work under Assumption \ref{Mixing_Assumptions}, where the PVAR error process  $\{\epsilon_{Sn+s} \}_{n\in\mathbb{Z}, s = 1,\dots,S}$ is assumed to be periodic, possibly \emph{weak} white noise. In this context, block bootstrapping is required to mimic the potential non-linear serial dependence, while a seasonal version is required because $\{\epsilon_{Sn+s} \}_{n\in\mathbb{Z}, s = 1,\dots,S}$ is allowed to be periodically correlated. Based on sample values $y_1,\dots, y_{SN}$ and pre-sample values $y_{s^{\prime}},\dots,y_0$, suppose that estimates of the PVAR($\boldsymbol{p})$ parameters $\beta$ and $\sigma$ are given under general linear constraints $\beta= R\gamma + r$ as stated in Section \ref{Subsection_Estimation}. 

The algorithm is as follows. We initialize the algorithm by choosing a positive integer block size $b < SN$ and define $l = \lceil SN/b \rceil$, where $\lceil\cdot\rceil$ denotes the ceiling function. Next, we define $(m \times b)$ blocks $\widehat{\mathcal{E}}_i = (\widehat{\epsilon}_i, \dots,  \widehat{\epsilon}_{i+b-1})$ for $i= 1,\dots, SN-b+1$. Then, we proceed as follows:
\begin{itemize} 
\item[$(1)$] For $t=1,b+1,2b+1, \dots, (l-1)b+1$, define the block of pseudo residuals by $\widehat{\mathcal{E}}^*_t  = \widehat{\mathcal{E}}_{k_t}$,
where $k_1, k_{b+1}, k_{2b+1}, \dots$ are independent, discrete uniform random variables on the set
\[
\Big\{t\in\{1,\ldots,SN\}:t+kS\text{ for }k\in\{-R_{1,SN},-R_{1,SN}+1,\ldots, R_{2,SN}-1, R_{2,SN}\}\Big\},
\]
where $R_{1,SN} = \lfloor\frac{t-1}{S}\rfloor$ and $R_{2,SN} = \lfloor\frac{SN-b-t}{S}\rfloor$ and $\lfloor\cdot\rfloor$ denotes the floor function.
\item[$(2)$]
Collect the blocks $(\widehat{\mathcal{E}}^*_{1}, \dots, \widehat{\mathcal{E}}^*_{(l-1)b+1})$ and remove the last $lb-SN$ observations, leaving a sample $\widehat{\epsilon}_1^{*}, \dots, \widehat{\epsilon}_{SN}^{*}$ of size $SN$.
\item[$(3)$] Set $(y^*_{s^{\prime}},\dots,y^*_0) =(y_{s^{\prime}},\dots,y_0)$ and use estimates of $\beta$ and $\sigma$ and $\widehat{\epsilon}_1^{*}, \dots, \widehat{\epsilon}_{SN}^{*}$ to recursively compute bootstrap observations $y^*_{1},\dots,y^*_{SN}$ according to
\begin{align*}
y^{*}_{Sn+s} =  \widehat{\nu}(s) +  \widehat{A}_{1}(s) y^{*}_{Sn+s-1} + \dots +  \widehat{A}_{p(s)}(s) y^{*}_{Sn+s-p(s)} + \widehat{\epsilon}^{*}_{Sn+s}
\end{align*}
for $s=1,\dots,S$ and $n=0,\dots, N-1$.
\item[$(4)$] Re-estimate $\beta$, $ \sigma$, $\Pi_h^{IR}$ and $\Psi_h^{SIR}$ from the bootstrap (pre)sample $(y^*_{s^{\prime}},\ldots,y_0^*),y_1^*\dots,y^*_{SN}$ to get $\widehat \beta_{res}^*$, $\widehat \sigma^*$, $\widehat \Pi_h^{IR*}$ and $\widehat \Psi_h^{SIR*}$ and use the same linear constraints and identification scheme as in the calculation of the estimates $\widehat \beta_{res}$, $\widehat \sigma$, $\widehat \Pi_h^{IR}$ and $\widehat \Psi_h^{SIR}$.
\end{itemize}
In the bootstrap algorithm, a seasonal block bootstrap proposed by \cite{dudek2014generalized,dudek2016generalized} is applied directly to the PVAR residuals. Alternatively, under Assumption \ref{Mixing_Assumptions}, a non-seasonal resampling variant can also be used. To initialize this bootstrap, we choose a positive integer block size $b < SN$ and define $l = \lceil SN/b \rceil$. Next, we define $(m \times b)$ blocks $\widehat{\mathcal{U}}_i = (\widehat{u}_i, \dots,  \widehat{u}_{i+b-1})$, $i= 1,\dots, SN-b+1$, where $\widehat{u}_{Sn+s} = \widehat{\Sigma}_{\epsilon}^{-1/2}(s)  \widehat{\epsilon}_{Sn+s}$ denote the \emph{standardized} residuals for $s= 1,\dots,S$ and $n = 0,\dots,N-1$. Then, the non-seasonal variant is obtained by replacing steps $(1)-(2)$ by $(1^{\prime})-(2^{\prime})$, where 
 \begin{itemize} 
\item[$(1^{\prime})$]  
For $t=1, \dots,l$, define the block of pseudo standardized residuals by $\widehat{\mathcal{U}}^*_t  = \widehat{\mathcal{U}}_{k_t}$, where $k_1,\dots, k_{l}$ are independent, discrete uniform random variables on $\{ 1,\dots, SN-b+1\}$.
\item[$(2^{\prime})$] Collect the blocks $(\widehat{\mathcal{U}}^*_{1}, \dots, \widehat{\mathcal{U}}^*_{l})$ and remove the last $lb-SN$ observations, leaving a sample $\widehat{u}_1^{*}, \dots, \widehat{u}_{SN}^{*}$ of size $SN$ and generate bootstrap PVAR residuals $\widehat{\epsilon}_1^{*}, \dots, \widehat{\epsilon}_{SN}^{*}$ by \[\widehat{\epsilon}_{Sn+s}^* = \widehat{\Sigma}_{\epsilon}^{1/2}(s)  \widehat{u}_{Sn+s}^*, \; \; s = 1,\dots, S,\; n = 0, \dots, N-1.\]
\end{itemize}
In this alternative bootstrap method, the periodic residuals are first standardized and then a classical moving-block bootstrap approach, as described, e.g., in \cite{bruggemann2016inference} and \cite{jentsch2022asymptotically}, is applied to the adequately standardized residuals. 
As both bootstrap algorithms need a suitably chosen block length $b$ as initialization, we refer to \cite{nordman2009note} and \cite{bertail2024optimal}, who developed optimal choices for the block length of non-seasonal and seasonal block bootstrap methods, respectively. With $b=1$, both block bootstrap methods collapse to the residual-based independent bootstraps to be discussed in Section \ref{Subsubsection_Bootstrapscheme_strong}. 

By repeating the bootstrap schemes $L$ times, where $L$ is large, we can get bootstrap quantiles for all estimators $\widehat \beta_{res}$, $\widehat \sigma$, $\widehat \Pi_h^{IR}$ and $\widehat \Psi_h^{SIR}$, which enables in particular the construction of (standard percentile) confidence intervals for the structural impulse responses.
 

\subsubsection{Residual-based periodic independent bootstrapping}\label{Subsubsection_Bootstrapscheme_strong}

If Assumption \ref{Mixing_Assumptions} with $(ii)$ and $(iii)$ replaced by $(ii)^\prime$ and $(iii)^\prime$ holds, where the PVAR innovation process $\{\epsilon_{Sn+s} \}_{n\in\mathbb{Z}, s = 1,\dots,S}$ is assumed to be periodic strong white noise, simplified residual-based independent bootstrap techniques can be used. In Section \ref{Subsubsection_Bootstrapscheme_weak}, by setting $b=1$ in steps $(1)-(2)$, the seasonal version of the residual-based independent bootstrap is obtained, while the non-seasonal version is obtained by setting $b=1$ in steps $(1^{\prime})-(2^{\prime})$.

The residual-based independent bootstrap methods provide valid confidence intervals only when Assumption \ref{Mixing_Assumptions} with $(ii)$ and $(iii)$ replaced by $(ii)^\prime$ and $(iii)^\prime$ hold. However, as illustrated in Figure  \ref{WWN1}, the periodic strong white noise assumptions may be violated in applications. Hence, we recommend to use residual-based block bootstrap methods for inference of periodically correlated time series. 

\subsection{Bootstrap Theory for PVARs}\label{Subsection_BootTheory_Weak}

We show bootstrap consistency for the residual-based seasonal block bootstrap methods discussed in Sections \ref{Subsubsection_Bootstrapscheme_weak} and \ref{Subsubsection_Bootstrapscheme_strong} under periodic weak and strong white noise assumptions, respectively. Addressing first the weak PVAR case, we need the following assumption.

\begin{assumption}[Periodic Weak White Noise with finite eighth-order Moments]\ \label{Mixing_assumption8}
The weak white noise process $\{w_{t} \}_{t\in\mathbb{Z}}$ as defined in Assumption \ref{Mixing_Assumptions} (ii) is $\alpha$-mixing such that Assumption \ref{Mixing_Assumptions} (iii) holds for $\rho = 8$. 
\end{assumption}
For weak white noise processes, such an assumption is made, e.g., in \cite{gonccalves2007asymptotic} to guarantee summability of cumulants up to order eight 
to prove consistency of wild and pairwise bootstrap methods for AR($\infty$) processes, in \cite{bruggemann2016inference} to prove consistency of a residual-based moving block bootstrap for VAR($p$) models, and in \cite{jentsch2019dynamic,jentsch2022asymptotically} to prove consistency of a residual-based moving-block bootstrap for Proxy SVARs. We now show that the residual-based seasonal block bootstrap can approximate the limiting distributions of $\sqrt{N}\bigl( \bigl(\widehat{\beta}_{res}- \beta\bigr)^{\prime}, \bigl(\widehat{\sigma}- \sigma \bigr)^{\prime}\bigr)^{\prime}$, $\sqrt{N}\bigl( vec \{ \widehat{\Pi}_{h}^{IR} - \Pi_{h}^{IR} \}\bigr),$ and $\sqrt{N} \bigl(vec \{ \widehat{\Psi}_{h}^{SIR} - \Psi_{h}^{SIR} \}\bigr),$ $h\in \mathbb{N}_0$ derived in Theorems \ref{theo_joint_clt} and \ref{CLT_IR_SIR}, respectively.

\begin{theorem}[Residual-Based Seasonal Block Bootstrap Consistency]\label{Mixing_bootstrap_con}\ Suppose Assumptions \ref{Mixing_Assumptions} and \ref{Mixing_assumption8} hold and the PVAR process \eqref{PVAR_model} fulfills the linear restrictions \eqref{restriction}. Further, suppose one of the residual-based seasonal block bootstrap schemes discussed in Section \ref{Subsubsection_Bootstrapscheme_weak} is used to obtain bootstrap estimators $\widehat{\beta}^{*}_{res}, \widehat{\sigma}^{*}, \widehat{\Pi}_{h}^{IR*}$ and $\widehat{\Psi}_{h}^{SIR*}$ for $h\in \mathbb{N}_0.$ If $b \to \infty$ such that $b^3/N \to 0$ as  $N \to \infty$, we get
\begin{align*}
&\sup_{x \in \mathbb{R}^{m_1}}\Big|P^{*}\Bigl(\sqrt{N}\bigl( \bigl(\widehat{\beta}^{*}_{res}- \widehat{\beta}_{res}\bigr)^{\prime}, \bigl(\widehat{\sigma}^{*}- \widehat{\sigma} \bigr)^{\prime}\bigr)^{\prime} \leq x \Bigr) - P\Bigl(\sqrt{N}\bigl( \bigl(\widehat{\beta}_{res}- \beta\bigr)^{\prime}, \bigl(\widehat{\sigma}- \sigma \bigr)^{\prime}\Bigr)^{\prime} \leq x \bigr)\Big|  \overset{p}{\to}  0,
\end{align*}
where $m_1 = m\sum_{s=1}^S(mp(s) +1)+S\widetilde{m}$, and, for $h\in \mathbb{Z}$, we have
\begin{align*}
&\sup_{x \in \mathbb{R}^{(Sm)^2}}\Big|P^{*}\Bigl(\sqrt{N}\bigl( vec \{ \widehat{\Pi}_{h}^{IR*} - \widehat{\Pi}_{h}^{IR} \}\bigr) \leq x \Bigr) - P\Bigl(\sqrt{N}\bigl( vec \{ \widehat{\Pi}_{h}^{IR} - \Pi_{h}^{IR} \}\bigr) \leq x \Bigr)\Big|  \overset{p}{\to}  0,  \\
&\sup_{x \in \mathbb{R}^{(Sm)^2}}\Big|P^{*}\Bigl(\sqrt{N}\bigl( vec \{ \widehat{\Psi}_{h}^{SIR*} - \widehat{\Psi}_{h}^{SIR} \}\bigr) \leq x \Bigr) - P\Bigl(\sqrt{N}\bigl( vec \{ \widehat{\Psi}_{h}^{SIR} - \Psi_{h}^{SIR} \}\bigr) \leq x \Bigr)\Big|  \overset{p}{\to}  0,
\end{align*}  
where $P^{*}$ denotes the probability measure induced by the residual-based seasonal block bootstrap methods (conditional on the data). Note that $x \leq y$ for some $x,y \in \mathbb{R}^K$ stands for $x_i \leq y_i$ for all $i=1,\dots,K$. 
\end{theorem}
The proof of Theorem \ref{Mixing_bootstrap_con} can be found in Appendix D. 
It is worth noting that in the case, where the more restrictive periodic strong white noise  assumption applies, the independent bootstrap, i.e., the block bootstrap with block length $b=1$, is sufficient.

\begin{corollary}[Residual-Based independent Bootstrap Consistency]\label{IID_bootstrap_con}\ Suppose Assumption \ref{Mixing_Assumptions} with $(ii)$ and $(iii)$ replaced by $(ii)^\prime$ and $(iii)^\prime$ and Assumption \ref{Mixing_assumption8} hold\footnote{Note that in the periodic strong white noise case, Assumption \ref{Mixing_assumption8} boils down to the condition $E|w_{t}|^8_8<\infty$, which could be relaxed to $E|w_{t}|^{4+\delta}_{4+\delta}<\infty$ for some $\delta>0$.} and suppose that the PVAR process \eqref{PVAR_model} fulfills the linear restrictions \eqref{restriction}. Further, one of the residual-based independent bootstrap schemes discussed in Section \ref{Subsubsection_Bootstrapscheme_strong} is used to obtain bootstrap estimators $\widehat{\beta}^*_{res}, \widehat{\sigma}^*, \widehat{\Pi}_{h}^{IR*}$ and $\widehat{\Psi}_{h}^{SIR*}$ for $h\in \mathbb{N}_0$. Then the same statements as in Theorem \ref{Mixing_bootstrap_con} hold,  
where $P^*$ denotes the probability measure induced by the residual-based independent bootstrap (conditional on the data).  
\end{corollary}

\subsection{Test for Seasonal Impulse Responses}\label{Section_SIR_Test}

In this section, we derive a bootstrap-based test for (seasonal) differences in the SPVAR and SVAR impulse response functions. This approach allows to test for the null hypothesis that there is no seasonal difference between SPVAR and SVAR impulse response functions, while the alternative finds a difference in at least one season. Formally, for fixed $i,j = 1,\dots, m$, we consider the null hypothesis
\begin{align}\label{H_0_1}
K_0: \Theta_{ij}^{SIR}(s)  = \Theta^{VAR}_{ij}, \;  \; \text{ for all } s = 1,\dots, S,
\end{align}
where $\Theta^{SIR}_{ij}(s)$ is the $(n_{IR}+1)$-dimensional vector consisting of seasonal impulse responses up to lag $n_{IR}$ of the $i$-th variable to a unit shock in the $j$-th structural shock, that is, $\Theta^{SIR}_{ij}(s) = (\Theta^{SIR}_{0,ij}(s), \dots, \Theta^{SIR}_{n_{IR},ij}(s))^{\prime}$. $\Theta^{VAR}_{ij}$ is the non-seasonal counterpart from the SVAR model on seasonally adjusted data. Then, given seasonally adjusted data $y_1^{SA},\ldots,y_{SN}^{SA}$ and seasonally unadjusted data $y_1^{SU},\ldots,y_{SN}^{SU}$, the test statistic $\widehat Q_{N}(i,j)$ is defined by the sum of squared differences in the estimated SPVAR and SVAR impulse response functions, that is,
\begin{align}\label{Teststatistik}
\widehat{Q}_{N}(i,j) = \sum_{s=1}^S | \widehat{\Theta}^{SIR}_{ij}(s) - \widehat{\Theta}^{VAR}_{ij}|_2^2.
\end{align}

We use a bootstrap-based method to mimic the distribution of $\widehat{Q}_N(i,j)$ under the null hypothesis $K_0$ in \eqref{H_0_1}. For this purpose, the bootstrap algorithm has to address the following points. Firstly, $\widehat{\Theta}^{SIR}_{ij}(s)$ and $\widehat{\Theta}^{VAR}_{ij}$ are generally not independent to each other and secondly, $\widehat{\Theta}^{SIR}_{ij}(s)$ is determined from seasonally unadjusted data, while $\widehat{\Theta}^{VAR}_{ij}$ is determined from seasonally adjusted data. Hence, the impulse response functions $\widehat{\Theta}^{SIR}_{ij}(s)$ and $\widehat{\Theta}^{VAR}_{ij}$ are therefore not only determined by two different models, but also originate from different data sets that are obviously highly interdependent. Accordingly, the bootstrap can only mimic the distribution of $\widehat{Q}_N(i,j)$ properly under the null if this dependence between the seasonally adjusted and unadjusted data is adequately captured.

In order to take these dependencies into account, we propose to generate bootstrap pseudo observations of the seasonally adjusted and unadjusted data by the following steps. First, we initialize the algorithm by choosing a VAR model order $p_{var}$ as the reduced-form VAR model and a (possibly) restricted PVAR of order $p = \max\{p(1), \dots, p(S)\}$ as the reduced-form PVAR model. If the PVAR model order varies across the seasons, note that we impose zero restrictions on $A_{j}(s)$ for $p(s) < j \leq p$. Further, we suppose that their structural model variants are identified using the same identification strategies. We also suppose that $p_{var}$ starting values for seasonally adjusted data and $p$ starting values for seasonally unadjusted data are available. Then, the algorithm works as follows:
\begin{itemize} 
\item[$(1)$] Independently, fit an SVAR($p_{var}$) model to the seasonally adjusted data and an SPVAR($p$) model to the seasonally unadjusted data. Calculate the SVAR estimates $\widehat{\nu}$, $\widehat{A}_1, \dots, \widehat{A}_{p_{var}}$, $\widehat{H}_0^{var}$ and the SPVAR estimates $\widehat{\nu}(s),\widehat{A}_1(s), \dots, \widehat{A}_p(s), \widehat{H}_0(s)$ for $s=1,\dots, S$, which allows to compute the structural residuals $(\widehat{w}_1^{VAR}, \dots, \widehat{w}_{SN}^{VAR})$ and $(\widehat{w}_1^{PVAR}, \dots, \widehat{w}_{SN}^{PVAR})$.  
\item[$(2)$] Calculate the (seasonally) pooled estimates of the PVAR coefficient matrices $\widebar{\widehat{A}}_1, \dots, \widebar{\widehat{A}}_p$, where  $\widebar{\widehat{A}}_i = 1/S \sum_{s=1}^S \widehat{A}_i(s),$ $i = 1,\dots, p$ and the (seasonally) pooled periodic impact matrix $\widebar{\widehat{H}}_0 = 1/S \sum_{s=1}^S \widehat{H}_0(s)$. Further, calculate the pooled coefficients $\widehat{A}^{pool}_i$, $i = 1,\dots, max\{p,p_{var}\}$ by pooling the VAR-estimates $\widehat{A}_1, \dots, \widehat{A}_{p_{var}}$ with the pooled PVAR estimates $\widebar{\widehat{A}}_1, \dots, \widebar{\widehat{A}}_p$ to get $\widehat{A}^{pool}_i=1/2(\widehat{A}_i+\widebar{\widehat{A}}_i)$ as well as the pooled impact matrix $\widehat{H}_0^{pool}$ by $\widehat{H}_0^{pool} = 1/2 ( \widebar{\widehat{H}}_0 + \widehat{H}_0^{var})$.
\item[$(3)$] Let $\widehat{\zeta}_{t} = ((\widehat{w}_{t}^{VAR})^\prime, (\widehat{w}_{t}^{PVAR})^\prime)^{\prime}$, $t = 1,\dots,SN$ denote the $2m$-dimensional vectors of residuals and perform a seasonal block bootstrap to obtain bootstrap pseudo observations $\widehat{\zeta}_1^*, \dots, \widehat{\zeta}_{SN}^*$.
\item[$(4)$] Use $(\nu, \widehat{A}^{pool}_1, \dots, \widehat{A}^{pool}_{max\{p,p_{var}\}}, \widehat{H}_0^{pool})$  and the first $m$ components of $\widehat{\zeta}_1^*, \dots, \widehat{\zeta}_{SN}^*$ to construct bootstrap pseudo observations $y^{SA*} = (y_1^{SA*},\dots,y_{SN}^{SA*})$ of the seasonally adjusted data and use $(\nu(1),\dots, \nu(S),  \widehat{A}^{pool}_1, \dots, \widehat{A}^{pool}_{max\{p,p_{var}\}}, \widehat{H}_0^{pool})$ and the last $m$ components of $\widehat{\zeta}_1^*, \dots, \widehat{\zeta}_{SN}^*$ to construct bootstrap pseudo observations $y^{SU*} = (y_1^{SU*},\dots,y_{SN}^{SU*})$ of the seasonally unadjusted data under the null. 
\item[$(5)$] Fit an SVAR($p_{var}$) to $y^{SA*}$ and an SPVAR($p$) to $y^{SU*}$ and calculate the bootstrap versions $\widehat{\Theta}^{VAR*}_{ij}$ and $\widehat{\Theta}^{SIR*}_{ij}(s), s = 1,\dots,S$. Calculate the bootstrap test statistic $\widehat{Q}_N^*(i,j)$ according to \eqref{Teststatistik}, but based on
the bootstrap versions $\widehat{\Theta}^{VAR*}_{ij}$ and $\widehat{\Theta}^{SIR*}_{ij}(s)$.
\item[$(6)$] Repeat $(3) - (5)$ $L$ times, where $L$ is large, and reject $K_0$ in \eqref{H_0_1}, if $\widehat{Q}_N(i,j) > \widehat{q}^*_{ 1-\alpha}$, where $\widehat{q}^*_{ 1-\alpha}$ is the $1-\alpha$ quantile of $\widehat{Q}_N^{*(1)}(i,j), \dots,\widehat{Q}_N^{*(L)}(i,j)$.
\end{itemize}
In Step (1), we model the seasonally adjusted and unadjusted data independently, which means that the dependence structure between the data is preserved in the residuals. In Step (2), we calculate the pooled autoregressive estimates and the pooled impact matrix, which ensure that the generated pseudo observations in Step (4) come from a model that satisfies the null hypothesis $K_0$. In order to mimic the dependence between the SPVAR and SVAR model estimates, we merge the residuals from both models in Step (3) and apply a seasonal block bootstrap that is able to capture this dependence structure as well as the potential seasonal structure in the SPVAR residuals. Step (4) produces the bootstrap pseudo samples of seasonally adjusted and unadjusted data. For both samples, the same pooled autoregressive coefficients and impact matrix are used to ensure that the true seasonal and non-seasonal impulse responses are the same. Note that the pooled autoregressive coefficients are computed by the mean of the VAR estimates and the pooled PVAR estimates.
To calculate the bootstrap test statistic $\widehat{Q}_N^*(i,j)$ under the null, in Step (5), we first fit (separately) the SVAR($p_{var}$) model to $y^{SA*}$ and the SPVAR($p$) model to $y^{SU*}$. Then, we obtain the bootstrap versions of the corresponding impulse responses by identifying both structural models by the same restrictions as used in Step (1). We calculate the bootstrap test statistics $\widehat{Q}_N^{*(l)}(i,j)$, $l=1,\ldots,L$ by inserting the bootstrap versions of the impulse responses and we reject $K_0$ if the test statistic $\widehat{Q}_N(i,j)$ exceeds the critical value $\widehat{q}^*_{1-\alpha}$.

We also derive a bootstrap-based test for the statistical significance of differences in SPVAR and SVAR impulse response functions, both based on seasonally unadjusted data. The null of this test can be formulated by $\widetilde K_0: \Theta_{ij}^{SIR}(s)  = \widebar{\Theta}^{SIR}_{ij}, \; \forall s=1,\dots,S$, where $\widebar{\Theta}^{SIR}_{ij}$ is the impulse response function from the SVAR($p$) with seasonal means. The test algorithm can be found in Appendix E.

\section{Real Data Application}\label{Section_Data}

In this section, we compare the results of applying a structural PVAR to seasonally unadjusted macroeconomic data with those of applying a structural VAR to seasonally adjusted data. The variables included in the model are monthly industrial production (IP) as a measure of real economic activity,  CPI inflation (INF) and federal funds rate (FFR) from Jan 1968 - Dec 2019 taken from the FRED database.  
The ADF test rejects unit roots in INF and the FFR, while it does not reject a unit root in industrial production. Hence, we address this by taking log differences of 
industrial production, whereas INF and FFR remain untouched. Furthermore, HEGY tests \citep{hylleberg1990seasonal} reject the null of seasonal unit roots for different frequencies in the variables. This result is also confirmed by the OCSB test \citep{osborn1988seasontest}, which also rejects the presence of seasonal unit roots. For structural PVAR analysis, IP and INF are considered on a seasonally unadjusted basis, while their seasonally adjusted counterparts are considered for SVAR analysis. As FFR appears to be non-periodic, we use it for both analyses. 
 
Due to the monthly frequency, an obvious choice for the number of seasons is $S =12$, resulting in $N=52$ cycles (years).  The choice  $S =12$ is strongly supported by the estimates of the SDs and ACFs of IP and INF given in Figure \ref{sd_acf_data}. As demonstrated in Figure \ref{sd_acf_data} and \ref{sd_acf_demean} in the introductory section, we find that IP and INF exhibit strong periodic patterns, while FFR has no periodic structure at all. The periodic structure of IP 
cannot be eliminated by seasonal demeaning alone, while the periodic structure of INF can be explained mainly by seasonal variation in the mean. Consequently, seasonal demeaning appears to be not sufficient to remove the entire seasonal component of IP, while in the case of INF the periodic structure is almost eliminated completely.

\subsection{Reduced-Form PVAR Model}\label{subsection_Reduced}

On the seasonally adjusted data basis, we use a VAR model of order $p= 9$ as the benchmark reduced-form model. The model order is chosen based on BIC. 
Based on Figures \ref{sd_acf_data}  and \ref{sd_acf_demean}, we allow the seasonal intercepts of IP and INF to be periodically varying over the seasons and set the intercepts of FFR to be non-periodic within our PVAR setup. For choosing a suitable model order for the reduced-form PVAR, we set the model order of the PVAR for each season to be the same. For the sake of simplicity, we then choose the model order that minimizes the seasonal BIC for all PVARs with seasonal means and non-periodic autoregressive effects\footnote{Since we do not know which autoregressive coefficients are periodic and which are not without knowing the model order, we proceed in this approach by determining the model order from VAR models with seasonal intercept. Alternatively, the bottom-up approach can be run for different PVAR orders, with the selected model determined by the seasonal BIC.}. Based on BIC, the PVAR model order is selected as $p(s)=12$ for $s=1,\dots,S = 12$.

In order to capture the periodic structure, which goes beyond seasonality in means, and to choose the linear restrictions of the reduced-form PVAR model, we use a slightly more sophisticated bottom-up strategy. We define 
$\widehat{\beta}^{(PVAR)}_i = \bigl(vec \{ \bigl(\widehat{A}_{1}(s), \dots, \widehat{A}_{12}(s) \bigr) \}_i, s = 1,2, \dots, S = 12 \bigr)^{\prime} $ as the periodic $S$-dimensional vector of the $i$-th autoregressive parameter for $i = 1, \dots, pm^2 = 108$, where $vec \{ \bigl(\widehat{A}_{1}(s), \dots, \widehat{A}_{12}(s) \bigr) \}_i$ is the $i$-th component of $vec \{ \bigl(\widehat{A}_{1}(s), \dots, \widehat{A}_{12}(s) \bigr) \}$. $\widebar{\widehat{\beta}}^{(PVAR)}_i$ is defined as the mean of $\widehat{\beta}^{(PVAR)}_i$ for $i=1,\dots,pm^2$. Then, the algorithm for the bottom-up approach is as follows: 

\begin{itemize} 
\item[$(1)$] For $l = 1$, fit all possible $pm^2 = 108$ PVAR models in which exactly one autoregressive parameter is allowed to vary periodically, while all other parameters are restricted to be constant over the seasons. For $i = 1,\dots, pm^2$, estimate the Wald-type statistic \begin{align*}
 \widehat{\lambda}_i = (\widehat{\beta}^{(PVAR)}_i - \widebar{\widehat{\beta}}^{(PVAR)}_i \; \boldsymbol{1}_S)^{\prime} \; \widehat{\Sigma}_{\widehat{\beta}^{(PVAR)}_i}^{-1} \; (\widehat{\beta}^{(PVAR)}_i - \widebar{\widehat{\beta}}^{(PVAR)}_i \; \boldsymbol{1}_S)
 \end{align*}
 and determine $i_1 = argmax \{ \widehat \lambda_i\} $. Choose that PVAR model that allows the $i_1$-th parameter to vary seasonally over the seasons, while all other parameters are restricted to be constant.
\item[$(2)$] For $l \in \mathbb{N}$ with $l\geq 2$, take the chosen PVAR model from iteration $l-1$ and fit all possible $pm^2-l+1$ PVAR models in which exactly one additional autoregressive parameter is allowed to vary periodically, while all remaining parameters are restricted to be constant over the seasons. For $i \in \{ 1,\dots, pm^2 \} \setminus \{i_1,\dots, i_{l-1} \}$, calculate the Wald-type statistic $\widehat{\lambda}_i$ again and determine $i_l = argmax \{ \widehat \lambda_i\}$. Choose that PVAR model that allows the parameters $i_1,\dots, i_{l}$ to vary seasonally over the seasons, while all remaining parameters are restricted to be constant. 
\item[$(3)$] Stop the procedure, if in iteration $l$ the quantity $max\{ \widehat{\lambda}_i \} /S $ does not exceed the $1-\alpha$ quantile of the $\chi^2_{1}$ distribution and take the PVAR model from iteration $l-1$.
\end{itemize}

In order to get $\widehat{\Sigma}_{\widehat{\beta}^{(PVAR)}_i}^{-1}$ for each PVAR model in Step (1) above, we propose to use residual-based seasonal block bootstrap procedures from Section \ref{Subsection_Bootstrapschemes}. 
Note that, in the first iteration of the bottom-up strategy, we begin by using a VAR with seasonal intercepts, before gradually allowing the model to incorporate periodically varying autoregressive parameters at the most periodic points. 
We proceed this way until the most periodic component is no longer significant. 
Note that in each step, assuming that the initial PVAR model is the true model, the Wald-type statistic divided by $S$ is asymptotically $\chi^2_{1}$-distributed.

\begin{figure}[t!]
\centering
\includegraphics[scale=0.8]{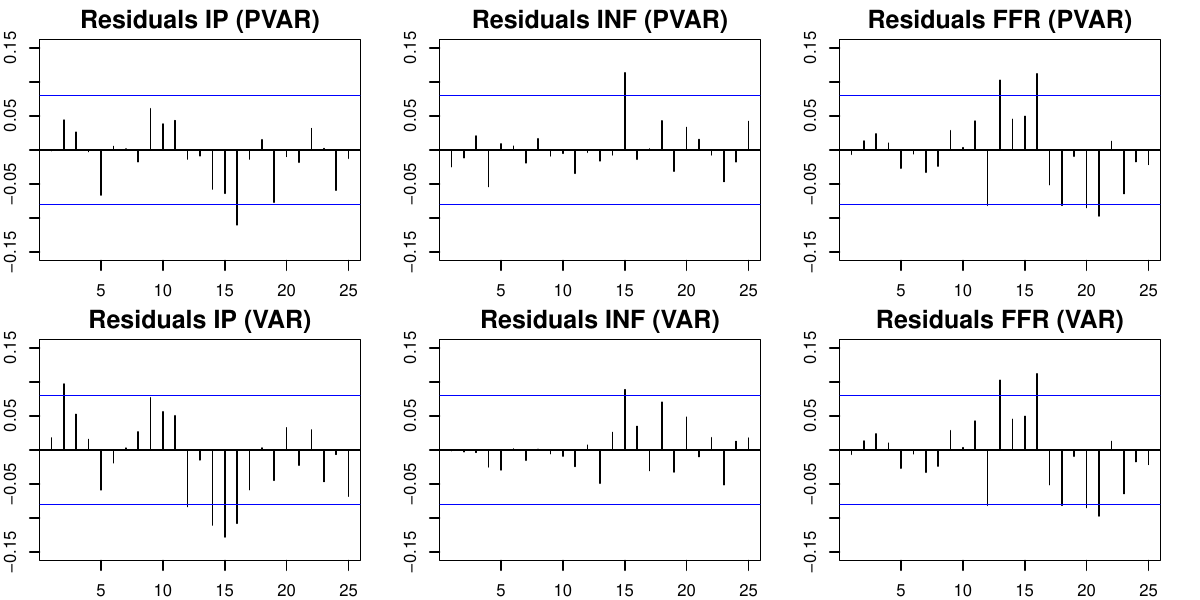}
\caption{ACFs of PVAR (as defined in \eqref{Restricted_PVAR9}) residuals $\widehat{\epsilon}_1, \widehat{\epsilon}_2, \dots, \widehat{\epsilon}_{SN}$ (top panels) and ACFs of residuals obtained by a VAR(12) with seasonal intercept 
(bottom panels). For both analyses, seasonally unadjusted monthly data from 1968-2019 are used.}
\label{ACF_PVAR_VAR}
\end{figure}

In general, we propose $\alpha \in [0.05,0.10]$. Note that too small $\alpha$ can lead to periodically inflexible models, while too large $\alpha$ can lead to models that overfit the periodic structure. In our real data analysis, we use the level $\alpha = 0.075$ and stop after four iterations of the bottom-up strategy. The resulting seasonally varying autoregressive parameters are the one-month delayed effects of IP on INF and on IP, the five-month delayed effects of IP on IP and the twelve-month delayed effects of IP on IP. All the remaining autoregressive coefficients are restricted to be constant. This finding is in line with the findings from Figures \ref{sd_acf_data} and \ref{sd_acf_demean} that IP is the variable in our system that contains most of the periodic structure, which goes beyond seasonality in means, while INF is slightly periodic and FFR is non-periodic.

Consequently, the reduced-form PVAR(12) with period $S=12$, where $y^{ip}, \pi$ and $i$ denote IP, INF and FFR, for $s = 1,\dots, 12$ and $n = 0,\dots ,51$, is as follows
\begin{align}\label{Restricted_PVAR9}
y_{12n+s} =  \nu(s) + \sum\limits_{i \in \mathcal{A}} A_{i}(s) y_{12n+s-i}  + \sum\limits_{i \in \mathcal{A}^c} A_{i}y_{12n+s-i} + \epsilon_{12n+s},
\end{align}
where  $y_{12n+s} = (y^{ip}_{12n+s}, \pi_{12n+s}, i_{12n+s})^{\prime}$, $\nu(s) = (\nu^{y^{ip}}(s), \nu^{\pi}(s), \nu^{i})^{\prime}$, $\mathcal{A} = \{ 1,5,12 \}$, $\mathcal{A}^c = \{1,\dots,12 \} \setminus \mathcal{A}$  and 
\begin{align*}
A_1(s) = \begin{pmatrix}
a^{(1)}_{11}(s) & a^{(1)}_{12} & a^{(1)}_{13} \\ a^{(1)}_{21}(s)  & a^{(1)}_{22} & a^{(1)}_{23} \\ a^{(1)}_{31} & a^{(1)}_{32} & a^{(1)}_{33} \end{pmatrix}, \; \; \; A_i(s) = \begin{pmatrix}
a^{(i)}_{11}(s) & a^{(i)}_{12} & a^{(i)}_{13} \\ a^{(i)}_{21}  & a^{(i)}_{22} & a^{(i)}_{23} \\ a^{(i)}_{31} & a^{(i)}_{32} & a^{(i)}_{33} \end{pmatrix}  \quad   \text{for } i = 5,12. 
\end{align*}

To investigate the validity of the imposed restrictions, we compare the residuals resulting from fitting the restricted PVAR(12) from \eqref{Restricted_PVAR9} with those of a VAR(12) with seasonal intercept. From Figures \ref{ACF_PVAR_VAR} and \ref{SD_PVAR_VAR}, it can be clearly seen that in particular the IP residuals from the PVAR have much less periodic structure than those from the VAR. Thus, the PVAR model \eqref{Restricted_PVAR9} is able to capture the periodicity of the data better than the VAR model with seasonal intercept, especially for IP, but also somewhat for INF. Note that the FFR is modeled in the same way in both reduced-form models, namely not periodically. Further, based on the seasonal BIC, the PVAR is preferred over the VAR with seasonal mean.
 
\subsection{Identification}\label{subsection_Identification}

In our analysis, we use full and approximate in-average identification 
to identify unit variance aggregate demand $(ad)$, aggregate supply $(as)$ and monetary policy $(mp)$ shocks. In both identification methods, a mixture of short- and long-run restrictions following \cite{shapiro1988sources}, \cite{gali1992well} and \cite{rubio2010structural} are used to identify structural $ad, as$ and $mp$ shocks. 

To identify the shocks by the full identification approach, we assume that $ad$ and $mp$ shocks do not have an effect on IP in the long-run for all seasons. Additionally, we impose that $mp$ shocks do not have contemporaneous effects on IP in all seasons. These restrictions result in a total of $Sm(m-1)/2 = 36$ restrictions in identification scheme \eqref{SPVAR_Scheme1}, with two long-run and one short-run restriction applying to all $S$ seasons. As \eqref{SPVAR_Scheme1} is a special case of \eqref{SPVAR_Scheme4}, these restrictions can also be formulated as restrictions on $H_0$ and $W(s),s=1,\dots, S$ in \eqref{SPVAR_Scheme4}. The short-run restriction is imposed as described in Remark \ref{SPVAR_Ident_SVARgeneralization} by setting the corresponding static effect in $H_0$ to zero and by setting the cumulative contemporaneous effects of all other structural shocks to zero for all $S$ seasons. Each long-run restriction is imposed by setting the corresponding element of the pooled (average) long-run impulse response 
\[ \frac{1}{S} \sum_{s=1}^S \Biggl[\Bigl(\sum_{k=0}^{\infty}\Phi^{IR}_{k}(s) \Bigr) H_{0} \Biggl]\]
to zero and by setting the corresponding element of the periodic long-run impulse response 
\[ \Biggl[\Bigl(\sum_{k=0}^{\infty}\Phi^{IR}_{k}(s) \Bigr) H_{0} W(s) \Biggl]\]
to zero for each season $s=1,\dots,S$. Consequently, each short- and long-run restriction results in $S+1$ restrictions on \eqref{SPVAR_Scheme4}, totalling $(S+1)m(m-1)/2 = 39$ restrictions, which are needed to fully identify the system. 

To achieve approximate in-average identification, we only impose the above restrictions on average, rather than on each season individually. Based on \eqref{SPVAR_Scheme3}, we identify the non-seasonal impact matrix $H_0$ and thus also $\widetilde{H}_0(s) = H_0 \; diag[\Lambda(s)]$ for $s = 1,\dots, S$ by imposing the short-run restriction on $H_0$ and the long-run restrictions on the pooled long-run impulse response $\widebar{\Theta}^{LR}$ given in Remark \ref{SPVAR_Identification_Methods}. Consequently, $m(m-1)/2 = 3$ restrictions are imposed in \eqref{SPVAR_Scheme3} which are needed to fully identify the system.

\subsection{Results}\label{Subsection_Results}

In the following, seasonal structural impulse responses using the restricted SPVAR(12) on seasonally unadjusted data and structural impulse responses using an SVAR(9) on seasonally adjusted data are analyzed, in which the aggregate demand $(ad)$, aggregate supply $(as)$ and monetary policy shocks $(mp)$ are identified. To identify the structural $ad$, $as$ and $mp$ shocks, we use a mixture of short- and long-run restrictions discussed in Section \ref{subsection_Identification}. In this Section we only compare SPVAR impulse responses identified by the full identification approach with SVAR impulse responses. Impulse responses identified by approximate in-average identification method can be found in Appendix F. SVAR impulse responses are identified using the same short- and long-run restrictions. 

The black bold line indicates the structural impulse response of the SVAR on seasonally adjusted data which can also be used as the benchmark,  while each of the red lines indicates a structural impulse response generated by the SPVAR for the $S=12$ seasons, i.e., from January to December. The dashed lines give the (pointwise) 68\% confidence intervals of the structural impulse responses of the SPVAR (in red) and the SVAR (in black).
 
\begin{figure}[t!]
\centering
\includegraphics[scale=0.7]{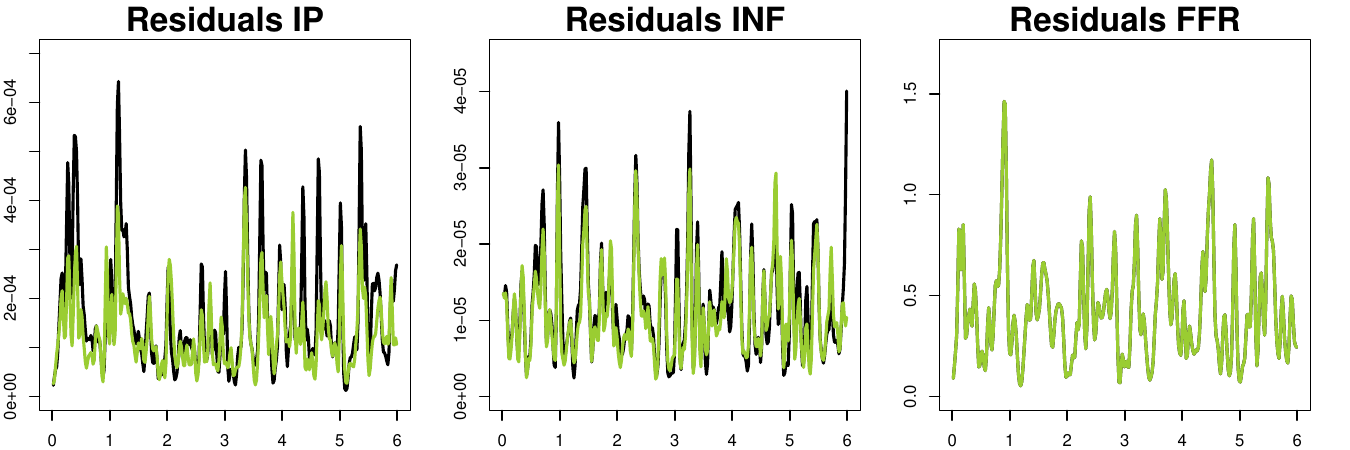}
\caption{Estimated SDs of PVAR (as defined in \eqref{Restricted_PVAR9}) residuals $\widehat{\epsilon}_1, \widehat{\epsilon}_2, \dots, \widehat{\epsilon}_{SN}$ (green) and of residuals obtained by a VAR(12) with seasonal intercept 
 (black). For both analyses, seasonally unadjusted monthly data from 1968-2019 are used. The x-axis of SD plots denotes the normalized frequencies times $S=12$.}
\label{SD_PVAR_VAR}
\end{figure}


For the construction of confidence intervals of the seasonal impulse responses, we use the non-seasonal variant of the residual-based seasonal block bootstrap described in Section \ref{Subsubsection_Bootstrapscheme_weak}, because the bottom row in Figure \ref{WWN1} shows significant autoregressive structure in the squared structural shocks such that the strong periodic white noise assumption appears to not hold. The block length of the residual-based seasonal block bootstrap is set to $b=5$, which is in line with the choice in \cite{jentsch2022asymptotically}, who considered a similar data generating process, and with the results in \cite{bertail2024optimal}.
We use standard percentile intervals for the construction of the confidence intervals. For the SVAR analysis, we use a non-seasonal block bootstrap to construct the confidence intervals. $L = 1000$ bootstrap repetitions are used for each bootstrap method.
 
\begin{figure}[t!]
\centering
\includegraphics[scale=0.8]{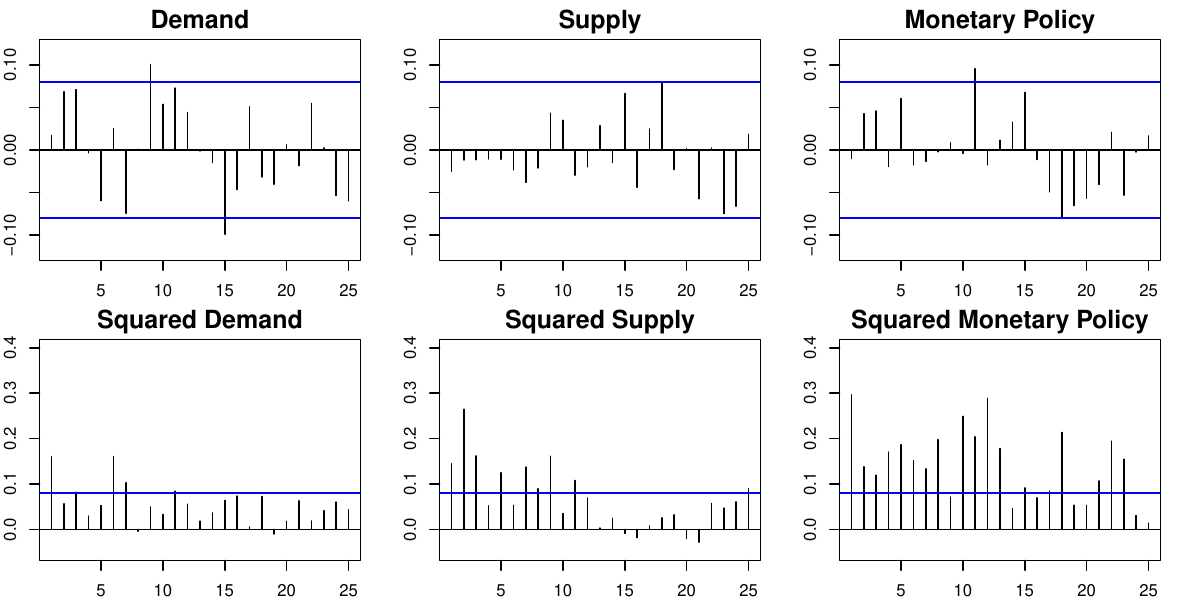}
\caption{ACFs of structural PVAR shocks  $\widehat{w}_1, \dots, \widehat{w}_{SN}$ (top panels)  and ACFs of squared structural PVAR shocks $\widehat{w}^2_1, \dots, \widehat{w}^2_{SN}$ (bottom panels). Structural shocks are identified based on the full identification approach.}
\label{WWN1}
\end{figure}
 
Figures \ref{NS_IR1}, \ref{NS_IR2}, and \ref{NS_IR3} present the seasonal structural impulse responses generated by a monetary policy shock. This shock is normalized as an unexpected increase in FFR of 25 basis points in each season. The figures show the effects on IP, INF and FFR, respectively. The impulse response functions generated by positive demand and negative supply shocks are deferred to the Appendix G.  
It appears that monetary policy shocks do not have significant periodic effects on the macro variables, suggesting that the timing of unexpected rate hikes or cuts by the Fed does not affect the macro variables. However, this may also be because the short- and long-run restrictions of the $mp$ shock to IP naturally limit the periodic flexibility of the impulse response. It is noticeable that the confidence bands for the SPVAR responses, especially for the INF and FFR impulse responses, are narrower compared to the non-seasonal SVAR impulse responses. Accordingly, it appears that we can make more precise statements on the basis of the SPVAR impulse responses in the individual seasons than at the non-seasonal SVAR level.

 \begin{figure}[t!]
\centering
\includegraphics[scale=0.75]{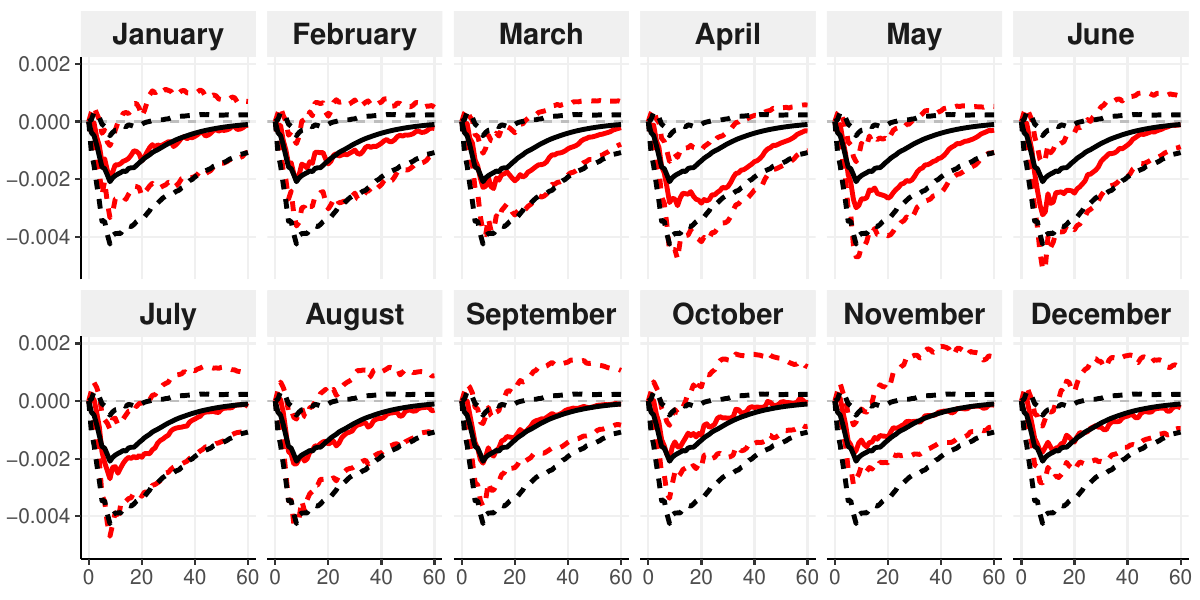}
\caption{Seasonal structural impulse responses of IP after a monetary policy shock. SPVAR impulse responses are in red, while SVAR impulse responses are in black and are constant across the seasons. The corresponding month indicates the time of occurrence of the shock.}
\label{NS_IR1}
\end{figure}

\begin{figure}[t!]
\centering
\includegraphics[scale=0.75]{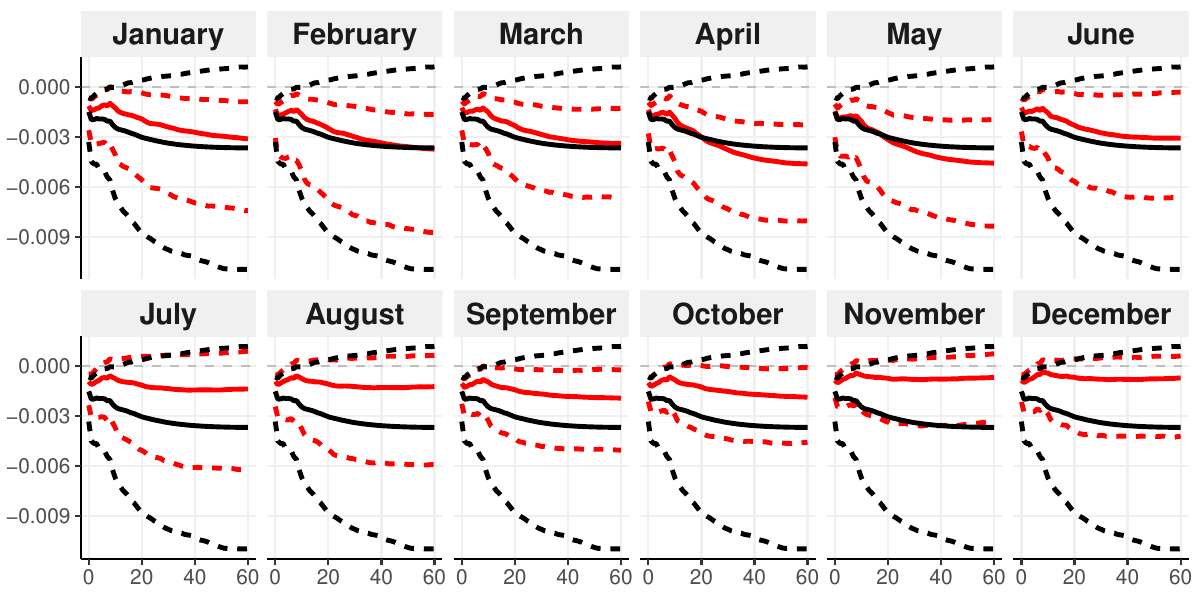}
\caption{Seasonal structural impulse responses of INF after a monetary policy shock. SPVAR impulse responses are in red, while SVAR impulse responses are in black and are constant across the seasons. The corresponding month indicates the time of occurrence of the shock.
}
\label{NS_IR2}
\end{figure}

 \begin{figure}[t!]
\centering
\includegraphics[scale=0.75]{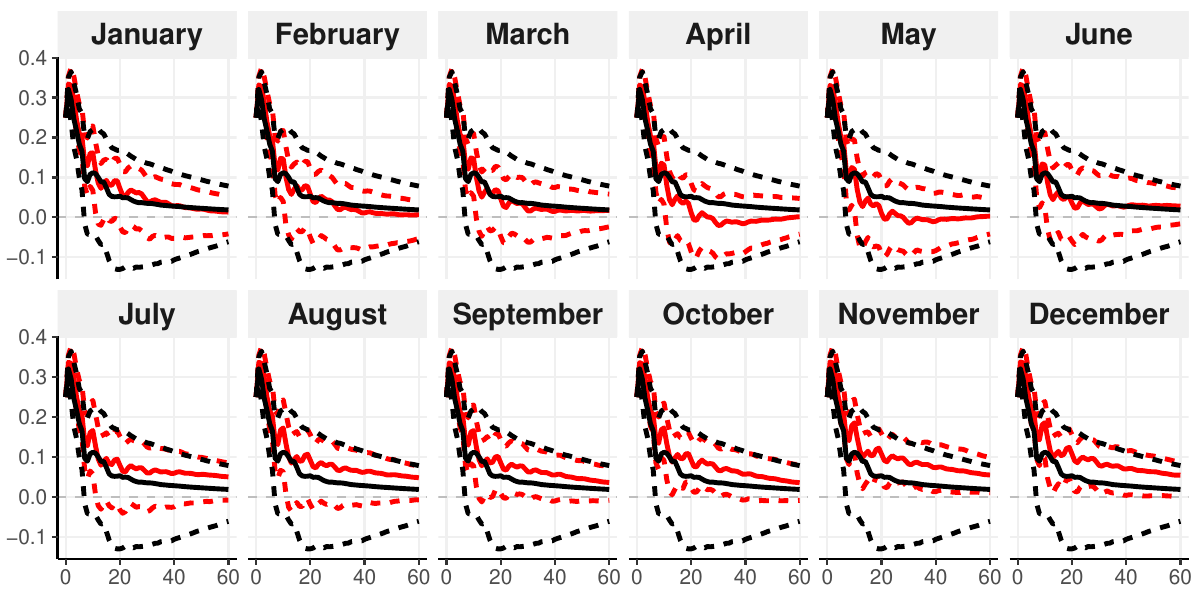}
\caption{Seasonal structural impulse responses of FFR after a monetary policy shock. SPVAR impulse responses are in red, while SVAR impulse responses are in black and are constant across the seasons. The corresponding month indicates the time of occurrence of the shock.}
\label{NS_IR3}
\end{figure}

In terms of periodic patterns, the structural impulse responses of macro variables triggered by positive demand or negative supply shocks exhibit significantly more periodicity, particularly in May and September. After a positive demand shock, we usually observe a strong upward shift in September, while observing an downward shift in May. The opposite pattern is observed after negative supply shocks. 

All in all, we find that the SVAR impulse responses are approximately given by the mean of the seasonal structural impulse responses. Interestingly, we do not see pronounced efficiency losses in the seasonal structural impulse responses that would materialize in terms of considerably increased confidence intervals compared to the SVAR impulse responses.  Instead, we even notice some efficiency gains in certain seasons resulting in more accurate statements; see, e.g., the December panels in Figures G.1 - G.6 in Appendix G.

We also perform the bootstrap-based test for the statistical significance of the seasonal differences in the SPVAR impulse response functions described in Section \ref{Section_SIR_Test}. Using the test statistic $\widehat Q_N(i,j)$, we test the null hypothesis $K_0: \Theta_{ij}^{SIR}(s)  = \Theta^{VAR}_{ij}, \text{ for all } s = 1,\dots, S$, i.e., whether there are significant differences in the SPVAR impulse response function $\Theta_{ij}^{SIR}(s)$ and the SVAR impulse response function $\Theta^{VAR}_{ij}$ from seasonally adjusted data for $i = y^{ip}, \pi, i$ and $ j = ad, as, mp$. 
In order to calculate the critical value under the null, we use the reduced-form PVAR model from Section \ref{subsection_Reduced}, the VAR model order $p_{var} = 9$
, the block length $b = 5$ and $L = 1000$ bootstrap repetitions. Both, SPVAR and SVAR models, are identified using a mixture of short- and long-run restrictions discussed in Section \ref{subsection_Identification}. We use the full identification approach for SPVAR identification.

Table \ref{SIR_Test_Table} shows the $p$-values of the test for seasonality of the impulse responses for different maximum number of lags $n_{IR}$. 
Overall, the test shows signficant periodicities in the impulse responses following demand and supply shocks, while monetary policy shocks do not appear to have any significant periodic effects on the macro variables. 
It should be noted that the contemporaneous effects of a monetary policy shock on the FFR are standardized for all seasons. However, as shown in Table \ref{SIR_Test_Table}, the impulse responses of the variables following monetary policy shocks are not seasonally significant even without standardization. 

\begin{table}[t!]
\centering
\footnotesize
\begin{tabular}{c|ccccccccc}
   \multicolumn{10}{c}{$K_0 : \Theta_{ij}^{SIR}(s)  = \Theta^{VAR}_{ij}$} \\ \hline  $n_{IR}$ & $ ad \to y^{ip}$ & $ ad \to  \pi$ &  $ ad \to  i$ & $as \to y^{ip}$ & $ as\to \pi$ & $ as \to i$ & $mp \to y^{ip}$ & $ mp \to \pi$ &  $mp \to i$ \\ 
  \hline
0 & 0.056 & 0.211 & 0.107 & 0.085 & 0.125 & 0.143 & 1.00 (1.00) & 0.79 (0.73) & 1.00 (0.67)\\ 
  12 & 0.302 & 0.044 & 0.052 & 0.047 & 0.070 & 0.092 & 0.93 (0.83) & 0.74 (0.68) & 0.87 (0.72) \\ 
  24 & 0.408 & 0.034 & 0.033 & 0.023 & 0.078 & 0.087 & 0.91 (0.82 & 0.75 (0.68) & 0.91 (0.77) \\ 
  36 & 0.472 & 0.019 & 0.024 & 0.012 & 0.082 & 0.077 & 0.88 (0.79) & 0.76 (0.69) & 0.92 (0.79) \\ 
  48 & 0.515 & 0.010 & 0.025 & 0.008 & 0.075 & 0.069 & 0.88 (0.88) & 0.77 (0.77) & 0.91 (0.91) \\ 
  60 & 0.538 & 0.007 & 0.021 & 0.007 & 0.064 & 0.063 & 0.89 (0.81) & 0.77 (0.69) & 0.91 (0.79) \\ 
   \hline
\end{tabular}
\caption{$p$-values of test on seasonality of impulse responses. The $p$-values of the impulse responses after a non-normalized monetary policy shock are given in parentheses.}
\label{SIR_Test_Table}
\end{table}

\subsection{Simulation Study}\label{Subsection_Simulation}


In Monte Carlo simulations, we examine the seasonal impulse response intervals constructed by the residual-based seasonal bootstrap methods discussed in Section \ref{Subsubsection_Bootstrapscheme_weak}. For this purpose, we generate seasonal data from the restricted structural PVAR that was obtained in the real data application when modeling the 
dynamics of IP, INF and FFR.
Hence, in the simulation study, we generate a 3-dimensional sample $y_1, \dots, y_{12N}$ by
\begin{align*}
y_{12n+s} =  \nu(s) + \sum\limits_{i \in \mathcal{A}} A_{i}(s) y_{12n+s-i}  + \sum\limits_{i \in \mathcal{A}^c} A_{i}y_{12n+s-i} + H_0(s) w_{12n+s},
\end{align*}
for $s = 1,\dots, 12$ and $n = 0,\dots ,N-1$,
where we use the estimated reduced-form PVAR(12) model described in \eqref{Restricted_PVAR9} 
and the estimated SPVAR impact matrix identified by the full identification approach as the true DGP. 
We generate the structural shocks $\{ w_{t} \}_{t \in \mathbb{Z}}$ using different DGPs, where one variant relies on i.i.d.~errors and three variants make use of differently parametrized GARCH error processes to cover both, the periodic strong and weak white noise cases. Using $v_t = (v_{1,t}, v_{2,t}, v_{3,t})^{\prime} \sim i.i.d. \; \mathcal{N}(\boldsymbol{0}, \boldsymbol{I}_3)$ for all $t \in \mathbb{Z}$, we define $w_{i,t} = \sigma_{i,t} v_{i,t}$ with  $\sigma^2_{i,t} = a_0 + a_1 w_{i,t}^2 + b_1 \sigma^2_{i,t-1}, i = 1,2,3$ and $a_0 = 1 - a_1 - b_1$ implying that the components of the structural shocks follow univariate independent GARCH(1,1) processes with $E(w_{i,t}^2) = 1,i = 1,2,3$. We use the following GARCH(1,1) specifications of $\{w_t\}_{t\in\mathbb{Z}}$ for simulations:
\begin{itemize}
    \item[(G0)] $a_1=0$, $b_1=0$    \   (G1) $a_1=0.05$, $b_1=0.9$  \   (G2)    $a_1=0.3$, $b_1=0.6$    \   (G3) $a_1=0.5$, $b_1=0$,
\end{itemize}
where the case G0 represents the periodic strong white noise case, and (G1)-(G3) represent variants of the periodic weak white noise case. All these cases are also considered in \cite{bruggemann2016inference}.

\begin{table}[t!]
\centering
\scriptsize
\begin{tabular}{l|l|c|ccccccccc}
 $b$ & $ N$ & $s$ & $ ad \to y^{ip}$ & $ ad \to  \pi$ &  $ ad \to  i$ & $as \to y^{ip}$ & $ as\to \pi$ & $ as \to i$ & $mp \to y^{ip}$ & $ mp \to \pi$ &  $mp \to i$ \\ 
  \hline
   & &  1 & 0.902 & 0.798 & 0.846 & 0.780 & 0.903 & 0.907 & 0.925 & 0.925 & 0.859 \\ 
  & &   2 & 0.884 & 0.871 & 0.796 & 0.947 & 0.882 & 0.940 & 0.921 & 0.920 & 0.859 \\ 
  & &   3 & 0.791 & 0.956 & 0.923 & 0.931 & 0.771 & 0.874 & 0.906 & 0.919 & 0.816 \\ 
   & &  4 & 0.857 & 0.720 & 0.722 & 0.831 & 0.935 & 0.867 & 0.910 & 0.934 & 0.844 \\ 
  & &   5 & 0.921 & 0.914 & 0.812 & 0.921 & 0.852 & 0.870 & 0.902 & 0.953 & 0.822 \\ 
 1  & 50 &  6 & 0.894 & 0.960 & 0.928 & 0.939 & 0.830 & 0.943 & 0.926 & 0.900 & 0.829 \\ 
  & &   7 & 0.866 & 0.920 & 0.821 & 0.922 & 0.901 & 0.956 & 0.887 & 0.861 & 0.785 \\ 
  & &   8 & 0.929 & 0.593 & 0.803 & 0.765 & 0.923 & 0.865 & 0.918 & 0.863 & 0.794 \\ 
  & &    9 & 0.948 & 0.745 & 0.963 & 0.645 & 0.877 & 0.802 & 0.906 & 0.900 & 0.794 \\ 
  & &   10 & 0.917 & 0.902 & 0.786 & 0.804 & 0.915 & 0.903 & 0.886 & 0.903 & 0.808 \\ 
  & &   11 & 0.722 & 0.947 & 0.923 & 0.952 & 0.692 & 0.958 & 0.877 & 0.858 & 0.757 \\ 
   & &  12 & 0.704 & 0.918 & 0.826 & 0.948 & 0.907 & 0.929 & 0.885 & 0.865 & 0.813 \\ 
  \hline
  & &  1 & 0.895 & 0.800 & 0.850 & 0.779 & 0.897 & 0.894 & 0.929 & 0.920 & 0.869 \\ 
  & &   2 & 0.886 & 0.870 & 0.787 & 0.947 & 0.884 & 0.944 & 0.918 & 0.920 & 0.869 \\ 
  & &   3 & 0.780 & 0.951 & 0.920 & 0.931 & 0.769 & 0.873 & 0.906 & 0.919 & 0.820 \\ 
  & &   4 & 0.850 & 0.713 & 0.719 & 0.829 & 0.937 & 0.870 & 0.908 & 0.929 & 0.845 \\ 
   & &  5 & 0.925 & 0.914 & 0.808 & 0.923 & 0.851 & 0.869 & 0.907 & 0.953 & 0.830 \\ 
 5 & 50 &   6 & 0.894 & 0.961 & 0.935 & 0.940 & 0.820 & 0.942 & 0.928 & 0.901 & 0.835 \\ 
  & &    7 & 0.867 & 0.914 & 0.813 & 0.919 & 0.896 & 0.954 & 0.888 & 0.860 & 0.785 \\ 
  & &   8 & 0.920 & 0.580 & 0.802 & 0.763 & 0.919 & 0.853 & 0.919 & 0.859 & 0.793 \\ 
  & &   9 & 0.950 & 0.727 & 0.962 & 0.632 & 0.884 & 0.806 & 0.908 & 0.894 & 0.800 \\ 
  & &   10 & 0.913 & 0.894 & 0.781 & 0.805 & 0.911 & 0.903 & 0.891 & 0.911 & 0.811 \\ 
  & &   11 & 0.709 & 0.945 & 0.923 & 0.954 & 0.681 & 0.958 & 0.889 & 0.867 & 0.762 \\ 
  & &   12 & 0.704 & 0.913 & 0.829 & 0.953 & 0.906 & 0.937 & 0.880 & 0.863 & 0.817 \\ \hline
  &  &    1 & 0.933 & 0.855 & 0.899 & 0.821 & 0.909 & 0.882 & 0.920 & 0.892 & 0.859 \\ 
     &  &    2 & 0.900 & 0.857 & 0.814 & 0.932 & 0.900 & 0.901 & 0.920 & 0.891 & 0.862 \\ 
     &  &    3 & 0.843 & 0.927 & 0.905 & 0.927 & 0.836 & 0.905 & 0.911 & 0.906 & 0.836 \\ 
     &  &    4 & 0.915 & 0.792 & 0.817 & 0.869 & 0.910 & 0.890 & 0.903 & 0.903 & 0.846 \\ 
     &  &    5 & 0.920 & 0.907 & 0.848 & 0.908 & 0.899 & 0.893 & 0.902 & 0.932 & 0.847 \\ 
   1 & 100 &    6 & 0.895 & 0.921 & 0.918 & 0.915 & 0.875 & 0.929 & 0.919 & 0.883 & 0.853 \\ 
     &  &    7 & 0.917 & 0.926 & 0.851 & 0.926 & 0.925 & 0.942 & 0.900 & 0.869 & 0.823 \\ 
     &  &    8 & 0.928 & 0.719 & 0.873 & 0.838 & 0.904 & 0.868 & 0.920 & 0.869 & 0.817 \\ 
     &  &    9 & 0.897 & 0.822 & 0.961 & 0.775 & 0.875 & 0.827 & 0.895 & 0.889 & 0.819 \\ 
     &  &   10 & 0.929 & 0.908 & 0.854 & 0.857 & 0.925 & 0.903 & 0.909 & 0.901 & 0.837 \\ 
     &  &   11 & 0.813 & 0.926 & 0.938 & 0.912 & 0.767 & 0.925 & 0.883 & 0.870 & 0.806 \\ 
     &  &   12 & 0.796 & 0.918 & 0.872 & 0.932 & 0.933 & 0.918 & 0.896 & 0.876 & 0.818 \\ 
    \hline
    &  &    1 & 0.932 & 0.855 & 0.896 & 0.822 & 0.914 & 0.882 & 0.915 & 0.905 & 0.860 \\ 
  &  &    2 & 0.898 & 0.867 & 0.821 & 0.931 & 0.898 & 0.902 & 0.920 & 0.896 & 0.864 \\ 
     &  &    3 & 0.843 & 0.918 & 0.908 & 0.935 & 0.831 & 0.900 & 0.901 & 0.910 & 0.840 \\ 
     &  &    4 & 0.915 & 0.791 & 0.820 & 0.875 & 0.914 & 0.888 & 0.904 & 0.903 & 0.847 \\ 
     &  &    5 & 0.916 & 0.905 & 0.844 & 0.918 & 0.897 & 0.892 & 0.909 & 0.940 & 0.835 \\ 
  7 &  100 &    6 & 0.894 & 0.920 & 0.917 & 0.919 & 0.883 & 0.931 & 0.917 & 0.886 & 0.844 \\ 
     &  &    7 & 0.916 & 0.921 & 0.856 & 0.924 & 0.917 & 0.936 & 0.908 & 0.876 & 0.820 \\ 
     &  &    8 & 0.929 & 0.720 & 0.864 & 0.838 & 0.907 & 0.874 & 0.917 & 0.868 & 0.821 \\ 
     &  &    9 & 0.899 & 0.822 & 0.971 & 0.770 & 0.872 & 0.828 & 0.895 & 0.894 & 0.826 \\ 
     &  &   10 & 0.928 & 0.903 & 0.844 & 0.856 & 0.919 & 0.899 & 0.906 & 0.900 & 0.838 \\ 
     &  &   11 & 0.805 & 0.924 & 0.933 & 0.914 & 0.765 & 0.933 & 0.885 & 0.868 & 0.807 \\ 
     &  &   12 & 0.797 & 0.913 & 0.868 & 0.935 & 0.937 & 0.914 & 0.896 & 0.877 & 0.815 \\ 
\end{tabular}
\caption{Simulated coverage of 90\%-confidence intervals of seasonal structural impulse response of IP ($y^{ip}$), INF $(\pi)$ and FFR ($i$) $k=12$ months after a shock to aggregate demand $(ad)$, aggregate supply ($as$) and monetary policy ($mp$). Structural schocks are generated using GARCH specification G0.}
\label{Cov_G0_k12}
\end{table}

\begin{table}[t!]
\centering
\scriptsize
\begin{tabular}{l|l|c|ccccccccc}
$b$ &$ N$ & $s$ & $ ad \to y^{ip}$ & $ ad \to  \pi$ &  $ ad \to  i$ & $as \to y^{ip}$ & $ as\to \pi$ & $ as \to i$ & $mp \to y^{ip}$ & $ mp \to \pi$ &  $mp \to i$ \\ 
  \hline
  & & 1 & 0.869 & 0.771 & 0.808 & 0.740 & 0.854 & 0.882 & 0.923 & 0.919 & 0.857 \\ 
  & &  2 & 0.865 & 0.848 & 0.742 & 0.935 & 0.853 & 0.955 & 0.911 & 0.927 & 0.859 \\ 
  & &  3 & 0.744 & 0.932 & 0.890 & 0.931 & 0.729 & 0.840 & 0.888 & 0.919 & 0.809 \\ 
  & &  4 & 0.856 & 0.682 & 0.675 & 0.825 & 0.916 & 0.870 & 0.880 & 0.928 & 0.835 \\ 
  & &  5 & 0.908 & 0.922 & 0.796 & 0.908 & 0.818 & 0.843 & 0.881 & 0.924 & 0.825 \\ 
 1 & 50 &  6 & 0.874 & 0.948 & 0.901 & 0.916 & 0.774 & 0.925 & 0.911 & 0.904 & 0.833 \\ 
   & &  7 & 0.829 & 0.909 & 0.777 & 0.905 & 0.877 & 0.948 & 0.869 & 0.870 & 0.781 \\ 
   & & 8 & 0.905 & 0.550 & 0.739 & 0.723 & 0.921 & 0.852 & 0.904 & 0.881 & 0.787 \\ 
  & &  9 & 0.954 & 0.670 & 0.949 & 0.609 & 0.893 & 0.781 & 0.886 & 0.898 & 0.793 \\ 
  & &  10 & 0.893 & 0.863 & 0.723 & 0.776 & 0.878 & 0.888 & 0.884 & 0.899 & 0.799 \\ 
  & &  11 & 0.639 & 0.945 & 0.885 & 0.935 & 0.641 & 0.946 & 0.868 & 0.863 & 0.737 \\ 
  & &  12 & 0.677 & 0.883 & 0.775 & 0.944 & 0.867 & 0.931 & 0.879 & 0.864 & 0.796 \\ 
   \hline
 & & 1 & 0.879 & 0.771 & 0.803 & 0.749 & 0.861 & 0.895 & 0.932 & 0.921 & 0.856 \\ 
  & &  2 & 0.871 & 0.854 & 0.746 & 0.944 & 0.851 & 0.954 & 0.916 & 0.933 & 0.858 \\ 
  & &  3 & 0.745 & 0.938 & 0.904 & 0.941 & 0.729 & 0.848 & 0.890 & 0.922 & 0.803 \\ 
  & &  4 & 0.857 & 0.680 & 0.669 & 0.828 & 0.919 & 0.875 & 0.890 & 0.931 & 0.845 \\ 
  & &  5 & 0.908 & 0.927 & 0.796 & 0.909 & 0.821 & 0.839 & 0.890 & 0.937 & 0.828 \\ 
 5 & 50 &  6 & 0.878 & 0.952 & 0.900 & 0.930 & 0.775 & 0.943 & 0.915 & 0.908 & 0.823 \\ 
  & &  7 & 0.828 & 0.917 & 0.779 & 0.915 & 0.878 & 0.957 & 0.875 & 0.870 & 0.776 \\ 
  & &  8 & 0.907 & 0.547 & 0.741 & 0.744 & 0.924 & 0.859 & 0.907 & 0.883 & 0.795 \\ 
  & &  9 & 0.951 & 0.667 & 0.953 & 0.594 & 0.897 & 0.789 & 0.895 & 0.897 & 0.790 \\ 
  & &  10 & 0.897 & 0.857 & 0.719 & 0.785 & 0.890 & 0.900 & 0.882 & 0.904 & 0.804 \\ 
  & &  11 & 0.646 & 0.946 & 0.894 & 0.943 & 0.643 & 0.951 & 0.867 & 0.864 & 0.737 \\ 
  & &  12 & 0.670 & 0.877 & 0.779 & 0.945 & 0.880 & 0.932 & 0.880 & 0.869 & 0.796 \\ 
   \hline
  &  &    1 & 0.897 & 0.816 & 0.846 & 0.809 & 0.878 & 0.872 & 0.916 & 0.900 & 0.835 \\ 
     &  &    2 & 0.860 & 0.844 & 0.793 & 0.908 & 0.873 & 0.899 & 0.926 & 0.911 & 0.850 \\ 
     &  &    3 & 0.779 & 0.886 & 0.865 & 0.912 & 0.779 & 0.879 & 0.894 & 0.913 & 0.834 \\ 
     &  &    4 & 0.890 & 0.755 & 0.754 & 0.846 & 0.905 & 0.873 & 0.893 & 0.911 & 0.829 \\ 
     &  &    5 & 0.897 & 0.899 & 0.810 & 0.889 & 0.845 & 0.866 & 0.881 & 0.923 & 0.820 \\ 
  1 & 100 &    6 & 0.873 & 0.903 & 0.870 & 0.897 & 0.827 & 0.922 & 0.908 & 0.894 & 0.829 \\ 
     &  &    7 & 0.891 & 0.913 & 0.821 & 0.910 & 0.897 & 0.930 & 0.879 & 0.867 & 0.817 \\ 
     &  &    8 & 0.925 & 0.658 & 0.817 & 0.788 & 0.896 & 0.858 & 0.906 & 0.864 & 0.817 \\ 
     &  &    9 & 0.887 & 0.759 & 0.941 & 0.714 & 0.863 & 0.824 & 0.899 & 0.894 & 0.807 \\ 
     &  &   10 & 0.899 & 0.892 & 0.795 & 0.828 & 0.876 & 0.893 & 0.892 & 0.903 & 0.818 \\ 
     &  &   11 & 0.719 & 0.905 & 0.891 & 0.907 & 0.720 & 0.912 & 0.896 & 0.861 & 0.787 \\ 
     &  &   12 & 0.756 & 0.899 & 0.832 & 0.915 & 0.903 & 0.914 & 0.896 & 0.858 & 0.805 \\ 
   \hline
    &  &    1 & 0.904 & 0.817 & 0.850 & 0.816 & 0.884 & 0.885 & 0.923 & 0.899 & 0.839 \\ 
     &  &    2 & 0.869 & 0.848 & 0.795 & 0.923 & 0.878 & 0.905 & 0.932 & 0.919 & 0.856 \\ 
     &  &    3 & 0.806 & 0.895 & 0.878 & 0.927 & 0.790 & 0.883 & 0.901 & 0.920 & 0.834 \\ 
     &  &    4 & 0.890 & 0.759 & 0.757 & 0.853 & 0.918 & 0.883 & 0.901 & 0.919 & 0.834 \\ 
     &  &    5 & 0.906 & 0.896 & 0.823 & 0.898 & 0.860 & 0.872 & 0.887 & 0.932 & 0.826 \\ 
   7 & 100 &    6 & 0.878 & 0.915 & 0.888 & 0.914 & 0.823 & 0.932 & 0.907 & 0.904 & 0.839 \\ 
    &  &    7 & 0.884 & 0.922 & 0.825 & 0.917 & 0.903 & 0.936 & 0.886 & 0.872 & 0.816 \\ 
     &  &    8 & 0.926 & 0.674 & 0.818 & 0.808 & 0.913 & 0.871 & 0.918 & 0.874 & 0.822 \\ 
     &  &    9 & 0.896 & 0.760 & 0.951 & 0.725 & 0.868 & 0.831 & 0.909 & 0.899 & 0.810 \\ 
     &  &   10 & 0.903 & 0.896 & 0.801 & 0.830 & 0.891 & 0.891 & 0.903 & 0.905 & 0.821 \\ 
     &  &   11 & 0.726 & 0.916 & 0.904 & 0.911 & 0.717 & 0.914 & 0.896 & 0.868 & 0.784 \\ 
     &  &   12 & 0.751 & 0.897 & 0.829 & 0.927 & 0.909 & 0.917 & 0.897 & 0.862 & 0.814 \\ 
\end{tabular}
\caption{Simulated coverage of 90\%-confidence intervals of seasonal structural impulse response of industrial production ($y^{ip}$), inflation $(\pi)$ and federal funds rate ($i$) $k=12$ months after a shock to aggregate demand $(ad)$, aggregate supply ($as$) and monetary policy ($mp$). Structural schocks are generated using GARCH specification G3.}
\label{Cov_G3_k12}
\end{table}

We simulate a total of $M=1000$ time series of length $N = 20,50,100$. 
Further, we bootstrap the structural impulse responses by using the non-seasonal variant of the residual-based seasonal block bootstrap for PVAR processes discussed in Section \ref{Subsubsection_Bootstrapscheme_weak} with block sizes $b \in \{1,3\}$ for $N=20$, $b\in\{1,5\}$ for $N=50$ and $b\in\{1,7\}$ for $N=100$ 
to obtain approximations of the confidence intervals. The block sizes are adopted from \cite{jentsch2022asymptotically}, who considered similar block and sample size ratios. Note that the block size has to increase with increasing sample size and that for $b=1$, the residual-based seasonal block bootstrap collapses to a seasonal independent variant. The confidence intervals are constructed by standard percentile intervals. The nominal coverage rate is 90\% and we use $L = 1000$ bootstrap repetitions.

Tables \ref{Cov_G0_k12} and \ref{Cov_G3_k12} 
show the simulated coverage rates for lag $k = 12$ and $N \in \{50,100 \}$ for GARCH specifications G0 and G3.  
The tables for lag $k = 0$, the tables for sample size $N=20$ and the tables for the other GARCH specifications G1 and G2 can be found in Appendix H, respectively. Due to a simplified presentation, we only provide coverage rates for lags $k \in\{0, 12\}$. For other lags, the coverage rate is not particularly different.
Generally, it can be clearly seen that the simulated coverage rates approach the nominal coverage rate of 90\% for increasing $N$. Table \ref{Cov_G0_k12} shows small efficiency losses when using the seasonal block bootstrap compared to its independent variant ($b=1$). This is to be expected, as the i.i.d.~case G0 contains the assumption of strong periodic white noise. When we consider Table \ref{Cov_G3_k12}, the simulated coverage rates approach the nominal coverage rate somewhat more closely when using the seasonal block bootstrap. We also find this result for cases G1 and G2. This is because the seasonal independent bootstrap (for $b=1$) is not able to capture the non-linear dependencies of the GARCH errors, while its block variant is able to do so. All in all, the standard percentile intervals hit the nominal coverage very well. In the case of a rather small data basis consisting of $N = 20$ years, the coverage rates look moderate and deviate quite a bit from the nominal coverage for some entries. For $N = 50, 100$, the coverage rates improve significantly and come very close to nominal coverage.

\section{Conclusion}\label{Section_Conclusion}

When dealing with seasonal data, deseasonalization using standard seasonal adjustment procedures is the standard methodology nowadays. This is the case even though seasonal adjustment procedures may distort the structure of the data and remove useful information.  In this paper, we offer an alternative \emph{direct} approach to analyze seasonally unadjusted macroeconomic data by structural models.  Instead of first seasonally adjusting the data and then fitting a structural VAR model, we propose to fit a structural periodic VAR directly on raw macroeconomic data that is seasonally unadjusted. We provide a PVAR representation that enables linearly restricted estimation of PVAR models under more general linear constraints than \cite{ursu2009modelling}, which allows to impose also constraints \emph{across seasons}. In addition, different identification methods for structural PVAR models are proposed, which, compared to the identification in SVARs, also take the (contemporaneous) seasonality of the data into account.  Further,  we show consistency and asymptotic normality of (constrained) least squares estimators for PVAR coefficients and structural impulse responses under periodic weak white noise assumption. For standard error estimation and confidence interval construction, we introduce different residual-based bootstrap methods for PVARs and prove their bootstrap consistency. We also discuss bootstrap-based statistical testing for seasonality in impulse responses.

Our empirical results illustrate the analysis of seasonally unadjusted industrial production, CPI inflation and federal funds rate using a linearly constrained structural PVAR.  For the structural analysis of the data, we identify three structural shocks,  i.e.,  monetary policy shock, aggregate supply and demand shock based on a mixture of short- and long-run restrictions.
We determine seasonal structural impulse responses based on the linearly restricted structural PVAR and compare the results with those from a standard SVAR analysis on the same US macro variables in seasonally adjusted form. We find that monetary policy shocks do not seem to have significant periodic effects on the macroeconomic variables. However, we observe significant periodic effects caused by aggregate supply and demand shocks. 

Moreover, we do not find pronounced efficiency losses of suitably constrained structural PVARs compared to seasonally adjusted SVARs, but even efficiency gains in some individual seasons. In total, we clearly see that useful insights into the dynamics of the macro variables are lost, if we seasonally adjust the data first. In the simulation study, we show that the simulated coverage rates constructed by the proposed residual-based bootstrap methods for PVARs closely match the nominal coverage rate, even for realistic sample sizes.

\section*{Funding}

Financial support by the Deutsche Forschungsgemeinschaft (DFG, German Research Foundation; Project-ID 520388526; TRR 391: Spatio-temporal Statistics for the Transition of Energy and Transport) and by the
Mercator Research Center Ruhr (MERCUR) with
project number Pe-2019-0044 is gratefully acknowledged.

\bibliographystyle{agsm}
\bibliography{literatur}

@article{lunsford2025residual,
  title={Residual Seasonality in Five Measures of PCE Inflation},
  author={Lunsford, Kurt G},
  journal={Economic Commentary},
  number={2025-03},
  year={2025},
  publisher={Federal Reserve Bank of Cleveland}
}

@article{mcelroy2022review,
  title={A review of seasonal adjustment diagnostics},
  author={McElroy, Tucker and Roy, Anindya},
  journal={International Statistical Review},
  volume={90},
  number={2},
  pages={259--284},
  year={2022},
  publisher={Wiley Online Library}
}

@article{eurostat2009guidelines,
  title={Guidelines on Seasonal Adjustment},
  author={Eurostat, ESS},
  journal={Techn. Ber. Eurostat Methodologies und Working Papers, European Commission},
  year={2009}
}

@article{hylleberg1990seasonal,
  title={Seasonal integration and cointegration},
  author={Hylleberg, Svend and Engle, Robert F and Granger, Clive WJ and Yoo, Byung Sam},
  journal={Journal of Econometrics},
  volume={44},
  number={1-2},
  pages={215--238},
  year={1990},
  publisher={Elsevier}
}

@book{hylleberg1992modelling,
  title={Modelling seasonality},
  author={Hylleberg, Svend},
  year={1992},
  publisher={Oxford University Press}
}

@article{hylleberg1993seasonality,
  title={Seasonality in macroeconomic time series},
  author={Hylleberg, Svend and J{\o}rgensen, Clara and S{\o}rensen, Nils Karl},
  journal={Empirical Economics},
  volume={18},
  pages={321--335},
  year={1993},
  publisher={Springer}
}

@article{canova1995seasonal,
  title={{Are seasonal patterns constant over time? A test for seasonal stability}},
  author={Canova, Fabio and Hansen, Bruce E},
  journal={Journal of Business \& Economic Statistics},
  volume={13},
  number={3},
  pages={237--252},
  year={1995},
  publisher={Taylor \& Francis}
}

@article{shishkin1967x,
  title={{The X-11 variant of the Census II method seasonal adjustment program}},
  author={Shiskin, Julius and Young, AH and Musgrave, JC},
  note={Technical Paper No. 15, US Department of Commerce, Bureau of Economic Analysis},
  year={1967}
}

@article{findley1998new,
  title={{New capabilities and methods of the X-12-ARIMA seasonal-adjustment program}},
  author={Findley, David F and Monsell, Brian C and Bell, William R and Otto, Mark C and Chen, Bor-Chung},
  journal={Journal of Business \& Economic Statistics},
  volume={16},
  number={2},
  pages={127--152},
  year={1998},
  publisher={Taylor \& Francis}
}

@inproceedings{monsell2007x,
  title={{The X-13A-S seasonal adjustment program}},
  author={Monsell, Brian C},
  booktitle={Proceedings of the 2007 Federal Committee On Statistical Methodology Research Conference. URL http://www. fcsm. gov/07papers/Monsell. II-B. pdf},
  pages={515},
  year={2007}
}

@book{gomez1996programs,
  title={Programs TRAMO and SEATS: Instructions for the user (beta version: September 1996)},
  author={G{\'o}mez, Victor and Maravall Herrero, Agust{\'\i}n},
  year={1996},
  publisher={Banco de Espa{\~n}a. Servicio de Estudios}
}

@article{gersovitz1978seasonality,
  title={Seasonality in regression: An application of smoothness priors},
  author={Gersovitz, Mark and MacKinnon, James G},
  journal={Journal of the American Statistical Association},
  volume={73},
  number={362},
  pages={264--273},
  year={1978},
  publisher={Taylor \& Francis}
}

@article{bell1984issues,
  title={Issues involved with the seasonal adjustment of economic time series},
  author={Bell, William R and Hillmer, Steven C},
  journal={Journal of Business \& Economic Statistics},
  volume={2},
  number={4},
  pages={291--320},
  year={1984},
  publisher={Taylor \& Francis}
}

@article{franses1996recent,
  title={Recent advances in modelling seasonality},
  author={Franses, Philip Hans},
  journal={Journal of Economic Surveys},
  volume={10},
  number={3},
  pages={299--345},
  year={1996},
  publisher={Wiley Online Library}
}

@article{osborn1989performance,
  title={{The performance of periodic autoregressive models in forecasting seasonal UK consumption}},
  author={Osborn, Denise R and Smith, Jeremy P},
  journal={Journal of Business \& Economic Statistics},
  volume={7},
  number={1},
  pages={117--127},
  year={1989},
  publisher={Taylor \& Francis}
}

@article{ghysels1996seasonal,
  title={Is seasonal adjustment a linear or nonlinear data-filtering process?},
  author={Ghysels, Eric and Granger, Clive WJ and Siklos, Pierre L},
  journal={Journal of Business \& Economic Statistics},
  volume={14},
  number={3},
  pages={374--386},
  year={1996},
  publisher={Taylor \& Francis}
}

@book{ghysels2001econometric,
  title={The econometric analysis of seasonal time series},
  author={Ghysels, Eric and Osborn, Denise R},
  year={2001},
  publisher={Cambridge University Press}
}

@article{doppelt2021should,
  title={{Should macroeconomists use seasonally adjusted time series? Structural identification and bayesian estimation in seasonal vector autoregressions}},
  author={Doppelt, Ross},
  journal={Manuscript, Michigan State University. https://www2. gwu. edu/\~{} forcpgm/Doppelt\_Seasonality\_2021. pdf},
  year={2021}
}

@article{del2004consequences,
  title={The consequences of seasonal adjustment for periodic autoregressive processes},
  author={Del Barrio Castro, Tom{\'a}s and Osborn, Denise R},
  journal={The Econometrics Journal},
  volume={7},
  number={2},
  pages={307--321},
  year={2004},
  publisher={Oxford University Press Oxford, UK}
}

@article{bell2012unit,
  title={Unit root properties of seasonal adjustment and related filters},
  author={Bell, William R},
  journal={Journal of Official Statistics},
  volume={28},
  number={3},
  pages={441--461},
  year={2012},
  publisher={Statistics Sweden (SCB)}
}

@article{findley2017detecting,
  title={Detecting seasonality in seasonally adjusted monthly time series},
  author={Findley, David F and Lytras, Demetra P and McElroy, Tucker S},
  note={Census Bureau Research Report: RRS2017/03},
  year={2017}
}

@article{rudebusch2015puzzle,
  title={{The puzzle of weak first-quarter GDP growth}},
  author={Rudebusch, Glenn D and Wilson, Daniel and Mahedy, Tim},
  journal={Puzzle},
  volume={2015},
  pages={16},
  year={2015},
  publisher={Citeseer}
}

@article{lunsford2017lingering,
  title={{Lingering residual seasonality in GDP growth}},
  author={Lunsford, Kurt G},
  journal={Economic Commentary},
  number={2017-06},
  year={2017},
  publisher={Federal Reserve Bank of Cleveland}
}

@article{consolvo2019residual,
  title={{Residual seasonality in GDP growth remains after latest BEA improvements}},
  author={Consolvo, Victoria and Lunsford, Kurt G},
  journal={Economic Commentary},
  number={2019-05},
  year={2019},
  publisher={Federal Reserve Bank of Cleveland}
}

@article{gladyshev1961periodically,
  title={Periodically correlated random sequence},
  author={Gladyshev, EG},
  journal={Soviet. Math.},
  volume={2},
  pages={385--388},
  year={1961}
}

@article{jones1967time,
  title={Time series with periodic structure},
  author={Jones, Richard H and Brelsford, William M},
  journal={Biometrika},
  volume={54},
  number={3-4},
  pages={403--408},
  year={1967},
  publisher={Oxford University Press}
}

@article{pagano1978periodic,
  title={On periodic and multiple autoregressions},
  author={Pagano, Marcello},
  journal={The Annals of Statistics},
  pages={1310--1317},
volume={6},
  year={1978},
  publisher={JSTOR}
}

@article{troutman1979some,
  title={Some results in periodic autoregression},
  author={Troutman, Brent M},
  journal={Biometrika},
  volume={66},
  number={2},
  pages={219--228},
  year={1979},
  publisher={Oxford University Press}
}

@article{osborn1988seasonality,
  title={Seasonality and habit persistence in a life cycle model of consumption},
  author={Osborn, Denise R},
  journal={Journal of Applied Econometrics},
volume={3},
  pages={255--266},
  year={1988},
  publisher={JSTOR}
}

@article{osborn1988seasontest,
  title={Seasonality and the order of integration for consumption},
  author={Osborn, Denise R and Chui, Alice PL and Smith, Jenny P and Birchenball, CR},
  journal={Oxford Bulletin of Economics \& Statistics},
  volume={50},
  number={4},
  year={1988}
}

@book{franses2004periodic,
  title={Periodic time series models},
  author={Franses, Philip Hans and Paap, Richard},
  year={2004},
  publisher={Oxford University Press}
}

@article{ghysels2006forecasting,
  title={Forecasting seasonal time series},
  author={Ghysels, Eric and Osborn, Denise R and Rodrigues, Paulo MM},
  journal={Handbook of Economic Forecasting},
  volume={1},
  pages={659--711},
  year={2006},
  publisher={Elsevier}
}

@article{novales1997forecasting,
  title={Forecasting with periodic models A comparison with time invariant coefficient models},
  author={Novales, Alfonso and de Fruto, Rafael Flores},
  journal={International Journal of Forecasting},
  volume={13},
  number={3},
  pages={393--405},
  year={1997},
  publisher={Elsevier}
}

@article{aknouche2007causality,
  title={{Causality conditions and autocovariance calculations in PVAR models}},
  author={Aknouche, Abdelhakim},
  journal={Journal of Statistical Computation and Simulation},
  volume={77},
  number={9},
  pages={769--780},
  year={2007},
  publisher={Taylor \& Francis}
}

@article{bollerslev1996periodic,
  title={Periodic autoregressive conditional heteroscedasticity},
  author={Bollerslev, Tim and Ghysels, Eric},
  journal={Journal of Business \& Economic Statistics},
  volume={14},
  number={2},
  pages={139--151},
  year={1996},
  publisher={Taylor \& Francis}
}

@article{aknouche2009quasi,
  title={{Quasi-maximum likelihood estimation of periodic GARCH and periodic ARMA-GARCH processes}},
  author={Aknouche, Abdelhakim and Bibi, Abdelouahab},
  journal={Journal of Time Series Analysis},
  volume={30},
  number={1},
  pages={19--46},
  year={2009},
  publisher={Wiley Online Library}
}

@article{aknouche2017periodic,
  title={Periodic autoregressive stochastic volatility},
  author={Aknouche, Abdelhakim},
  journal={Statistical Inference for Stochastic Processes},
  volume={20},
  pages={139--177},
  year={2017},
  publisher={Springer}
}

@article{franses2005forecasting,
  title={The forecasting performance of various models for seasonality and nonlinearity for quarterly industrial production},
  author={Franses, Philip Hans and Van Dijk, Dick},
  journal={International Journal of Forecasting},
  volume={21},
  number={1},
  pages={87--102},
  year={2005},
  publisher={Elsevier}
}

@article{ula1990periodic,
  title={Periodic covariance stationarity of multivariate periodic autoregressive moving average processes},
  author={Ula, Taylan A},
  journal={Water Resources Research},
  volume={26},
  number={5},
  pages={855--861},
  year={1990},
  publisher={Wiley Online Library}
}

@article{ula1993forecasting,
  title={Forecasting of multivariate periodic autoregressive moving-average processes},
  author={Ula, Taylan A},
  journal={Journal of Time Series Analysis},
  volume={14},
  number={6},
  pages={645--657},
  year={1993},
  publisher={Wiley Online Library}
}

@book{lutkepohl2005new,
  title={New introduction to multiple time series analysis},
  author={L{\"u}tkepohl, Helmut},
  year={2005},
  publisher={Springer Science \& Business Media}
}

@article{ursu2009modelling,
  title={On modelling and diagnostic checking of vector periodic autoregressive time series models},
  author={Ursu, Eugen and Duchesne, Pierre},
  journal={Journal of Time Series Analysis},
  volume={30},
  number={1},
  pages={70--96},
  year={2009},
  publisher={Wiley Online Library}
}

@article{boubacar2023estimating,
  title={Estimating weak periodic vector autoregressive time series},
  author={Boubacar Ma{\"\i}nassara, Yacouba and Ursu, Eugen},
  journal={TEST},
volume = {32},
  pages={958--997},
  year={2023},
  publisher={Springer}
}

@article{lund2000recursive,
  title={{Recursive prediction and likelihood evaluation for periodic ARMA models}},
  author={Lund, Robert and Basawa, IV},
  journal={Journal of Time Series Analysis},
  volume={21},
  number={1},
  pages={75--93},
  year={2000},
  publisher={Wiley Online Library}
}

@book{kilian2017structural,
  title={Structural vector autoregressive analysis},
  author={Kilian, Lutz and L{\"u}tkepohl, Helmut},
  year={2017},
  publisher={Cambridge University Press}
}

@article{francq2011asymptotic,
  title={{Asymptotic properties of weighted least squares estimation in weak PARMA models}},
  author={Francq, Christian and Roy, Roch and Saidi, Abdessamad},
  journal={Journal of Time Series Analysis},
  volume={32},
  number={6},
  pages={699--723},
  year={2011},
  publisher={Wiley Online Library}
}

@article{bibi2016periodic,
  title={On periodic time-varying bilinear processes: Structure and asymptotic inference},
  author={Bibi, Abdelouahab and Ghezal, Ahmed},
  journal={Statistical Methods \& Applications},
  volume={25},
  pages={395--420},
  year={2016},
  publisher={Springer}
}

@article{bruggemann2016inference,
  title={{Inference in VARs with conditional heteroskedasticity of unknown form}},
  author={Br{\"u}ggemann, Ralf and Jentsch, Carsten and Trenkler, Carsten},
  journal={Journal of Econometrics},
  volume={191},
  number={1},
  pages={69--85},
  year={2016},
  publisher={Elsevier}
}

@article{jentsch2022asymptotically,
  title={{Asymptotically valid bootstrap inference for proxy SVARs}},
  author={Jentsch, Carsten and Lunsford, Kurt G},
  journal={Journal of Business \& Economic Statistics},
  volume={40},
  number={4},
  pages={1876--1891},
  year={2022},
  publisher={Taylor \& Francis}
}

@article{jentsch2012periodic,
  title={A new frequency domain approach of testing for covariance stationarity and for periodic stationarity in multivariate linear processes},
  author={Jentsch, Carsten},
  journal={Journal of Time Series Analysis},
  volume={33},
  number={2},
  pages={177--192},
  year={2012},
  publisher={Wiley}
}

@article{francq2007multivariate,
  title={Multivariate portmanteau test for autoregressive models with uncorrelated but nonindependent errors},
  author={Francq, Christian and Ra{\"\i}ssi, Hamdi},
  journal={Journal of Time Series Analysis},
  volume={28},
  number={3},
  pages={454--470},
  year={2007},
  publisher={Wiley Online Library}
}

@article{dudek2014generalized,
  title={A generalized block bootstrap for seasonal time series},
  author={Dudek, Anna E and Le{\'s}kow, Jacek and Paparoditis, Efstathios and Politis, Dimitris N},
  journal={Journal of Time Series Analysis},
  volume={35},
  number={2},
  pages={89--114},
  year={2014},
  publisher={Wiley Online Library}
}

@article{dudek2016generalized,
  title={Generalized seasonal tapered block bootstrap},
  author={Dudek, Anna E and Paparoditis, Efstathios and Politis, Dimitris N},
  journal={Statistics \& Probability Letters},
  volume={115},
  pages={27--35},
  year={2016},
  publisher={Elsevier}
}

@article{bertail2024optimal,
  title={Optimal choice of bootstrap block length for periodically correlated time series},
  author={Bertail, Patrice and Dudek, Anna E},
  journal={Bernoulli},
  volume={30},
  number={3},
  pages={2521--2545},
  year={2024},
  publisher={Bernoulli Society for Mathematical Statistics and Probability}
}

@article{gali1992well,
  title={{How well does the IS-LM model fit postwar US data?}},
  author={Gali, Jordi},
  journal={The Quarterly Journal of Economics},
  volume={107},
  number={2},
  pages={709--738},
  year={1992},
  publisher={MIT Press}
}

@article{nordman2009note,
  title={A note on the stationary bootstrap’s variance},
  author={Nordman, Daniel J},
  year={2009}, journal={The Annals of Statistics}, volume = {37}, pages = {359--370}
}

@article{gonccalves2007asymptotic,
  title={{Asymptotic and bootstrap inference for AR($\infty$) processes with conditional heteroskedasticity}},
  author={Gon{\c{c}}alves, S{\'\i}lvia and Kilian, Lutz},
  journal={Econometric Reviews},
  volume={26},
  number={6},
  pages={609--641},
  year={2007},
  publisher={Taylor \& Francis}
}

@article{jentsch2019dynamic,
  title={{The dynamic effects of personal and corporate income tax changes in the United States: Comment}},
  author={Jentsch, Carsten and Lunsford, Kurt G},
  journal={American Economic Review},
  volume={109},
  number={7},
  pages={2655--2678},
  year={2019},
  publisher={American Economic Association 2014 Broadway, Suite 305, Nashville, TN 37203}
}

@article{rubio2010structural,
  title={Structural vector autoregressions: Theory of identification and algorithms for inference},
  author={Rubio-Ramirez, Juan F and Waggoner, Daniel F and Zha, Tao},
  journal={The Review of Economic Studies},
  volume={77},
  number={2},
  pages={665--696},
  year={2010},
  publisher={Wiley-Blackwell}
}

@article{shapiro1988sources,
  title={Sources of business cycle fluctuations},
  author={Shapiro, Matthew D and Watson, Mark W},
  journal={NBER Macroeconomics annual},
  volume={3},
  pages={111--148},
  year={1988},
  publisher={MIT Press}
}

\begin{titlepage} 
\renewcommand{\thefootnote}{\fnsymbol{footnote}}
\begin{center}
	{\LARGE Online Supplemental Appendices to \\ \bigskip Structural Periodic Vector Autoregressions 
 }
\end{center}
\bigskip

\begin{center}
	{\Large  Daniel Dzikowski\footnote{TU Dortmund University,  Department of Statistics, D-44221 Dortmund, Germany; dzikowski@statistik.tu-dortmund.de; corresponding author} \\   TU Dortmund University}
 \end{center}
 \begin{center}
	 {\Large  Carsten Jentsch\footnote{TU Dortmund University,  Department of Statistics, D-44221 Dortmund, Germany; jentsch@statistik.tu-dortmund.de} \\  TU Dortmund University}
  \end{center}

\bigskip

\begin{center}
	{\today}
\end{center}

\end{titlepage}
\appendix 
\numberwithin{equation}{section}
\numberwithin{figure}{section}
\numberwithin{table}{section}
\section{Proof of Theorem \ref*{theo_joint_clt}}\label{Joint_CLT_Proof}

In order to prove the joint CLT of the PVAR estimators under Assumption \ref*{Mixing_Assumptions}, some preparatory work has to be done. At first, for $a,b,c \in \mathbb{Z}$, we define $\kappa_{a,b}(s_1,s_2)$, $\tau_{a,b,c}(s_1, s_2)$, $s_1,s_2 = 1,\dots, S$ and $\tau_{a,b,c}$
as in (\ref*{cumulants1}) and (\ref*{cumulants2}). As given in (\ref*{Res_PVAR_Estimator}), the multivariate LS estimator under general linear restrictions $\beta = R \gamma + r$ is given by:
\begin{align*}
\widehat{\gamma} = \left[R^{\prime} \{X X^{\prime} \otimes \boldsymbol{I}_m \} R \right]^{-1} R^{\prime} \{ X \otimes \boldsymbol{I}_m \} \bold{z}.
\end{align*}
Plugging in $\bold{z}  = \{X^{\prime} \otimes \boldsymbol{I}_m\}  R\gamma + e$ gives

\begin{align*}
 \widehat{\gamma} - \gamma  =  \left[R^{\prime} \{X X^{\prime} \otimes \boldsymbol{I}_m \} R \right]^{-1} R^{\prime} \{ X \otimes \boldsymbol{I}_m \} e = \left[R^{\prime} \{X X^{\prime} \otimes \boldsymbol{I}_m \} R \right]^{-1} R^{\prime} \; vec\{EX^{\prime}\},
\end{align*}
where
\[
vec\{EX^{\prime}\} = \sum\limits_{n=0}^{N-1} \begin{pmatrix}
\epsilon_{Sn+1}\ \\
vec\{\epsilon_{Sn+1} y_{Sn}^{\prime} \}
 \\ \vdots \\ vec\{\epsilon_{Sn+1} y_{Sn-p(1)+1}^{\prime} \} \\
 \epsilon_{Sn+2} \\ vec\{\epsilon_{Sn+2} y_{Sn+1}^{\prime} \}
 \\ \vdots \\ vec\{\epsilon_{Sn+2} y_{Sn-p(2)+2}^{\prime} \} \\ \vdots \\\epsilon_{Sn+S-1} \\ vec\{\epsilon_{Sn+S} y_{Sn+S-1}^{\prime} \}
 \\ \vdots \\ vec\{\epsilon_{Sn+S} y_{Sn+S-p(S)}^{\prime} \}
\end{pmatrix}.
\]
Due to Assumption \ref*{Mixing_Assumptions} $(i)$, the moving-average representation $y_{Sn+s} = \mu(s) +  \sum_{k=0}^{\infty} \Phi_{k}(s)  \epsilon_{Sn+s-k}$ of the PVAR($\boldsymbol{p}$) process exists and $\mu(s)=\sum_{k=0}^\infty \Phi_{k}(s)  \nu(s-k)$ can be used to represent $X_n(s)$ as
\begin{align*}
X_n(s) &= \begin{pmatrix} 1 \\ y_{Sn+s-1} \\ \vdots \\ y_{Sn+s-p(s)} \end{pmatrix} =  \begin{pmatrix}   1 \\ \mu(s-1)+  \sum\limits_{k=0}^{\infty}  \Phi_{k}(s-1)  \epsilon_{Sn+s-1-k} \\ \vdots \\  \mu(s-p(s)) + \sum\limits_{k=0}^{\infty}  \Phi_{k}(s-p(s))  \epsilon_{Sn+s-p(s)-k}  \end{pmatrix}\\&= \begin{pmatrix}
1 \\
 \sum\limits_{k=1}^{\infty} \begin{pmatrix}  \Phi_{k-1}(s-1) [ \epsilon_{Sn+s-k} + \nu(s-k)] \\ \vdots \\  \Phi_{k-1}(s-p(s)) [ \epsilon_{Sn+s-p(s)+1-k} + \nu(s-p(s)+1-k)]  \end{pmatrix} \end{pmatrix} \\
&= \begin{pmatrix} 1 \\ \sum\limits_{k=1}^{\infty} \begin{pmatrix}  \Phi_{k-1}(s-1) [ \epsilon_{Sn+s-k} + \nu(s-k)] \\ \vdots \\  \Phi_{k-p(s)}(s-p(s)) [ \epsilon_{Sn+s-k} + \nu(s-k)] \end{pmatrix} \end{pmatrix} \\ 
&= \begin{pmatrix} 1 \\  \sum\limits_{k=1}^{\infty} C_k(s)[ \epsilon_{Sn+s-k}+ \nu(s-k)] \end{pmatrix},
\end{align*}
where
\[ 
C_k(s) = \begin{pmatrix}
\Phi_{k-1}(s-1) \\ \vdots \\ \Phi_{k-p(s)}(s-p(s)) 
\end{pmatrix} \in \mathbb{R}^{mp(s) \times m}. 
\]
Note that $\Phi_{k-p(s)}(s)= \boldsymbol{0}$ for $k < p(s), s=1,\dots,S$. Hence, for $vec\{EX^{\prime}\}$, we get
\begin{align*}
vec\{EX^{\prime}\} =&  \begin{pmatrix}
\boldsymbol{I}_m & \boldsymbol{0}_{m\times m^2} & \dots & \dots &  \boldsymbol{0} \\ \boldsymbol{0}_{m^2p(1) \times m} & \sum_{k=1}^{\infty}  (C_k(1) \otimes \boldsymbol{I}_m)  & \ddots  & & \vdots \\ \vdots & \ddots & \ddots & \ddots & \vdots \\  \vdots &  & \ddots & \boldsymbol{I}_m &  \boldsymbol{0}_{m \times m^2} \\ \boldsymbol{0} & \dots & \dots & \boldsymbol{0}_{m^2p(S) \times m} & \sum_{k=1}^{\infty}  (C_k(S)  \otimes \boldsymbol{I}_m)
\end{pmatrix}  \\ &\times
 \sum\limits_{n=0}^{N-1} \begin{pmatrix} \epsilon_{Sn+1} \\ vec\{\epsilon_{Sn+1} [\epsilon_{Sn+1-k}^{\prime} + \nu(1-k)^{\prime}] \} \\ \epsilon_{Sn+2} \\ vec\{\epsilon_{Sn+2}[ \epsilon_{Sn+2-k}^{\prime} +   \nu(2-k)^{\prime}] \} \\ \vdots \\ \epsilon_{Sn+S}  \\ vec\{\epsilon_{Sn+S} [ \epsilon_{Sn+S-k}^{\prime} +   \nu(S-k)^{\prime} ]   \} \end{pmatrix}.
 \end{align*}
Note that, here and in the following notation, the sums over $k$ in the first matrix above are also applied to the second factor $\sum_{n=0}^{N-1}(\cdots)$.
Further, using $\sqrt{N} \bigl(\widehat{\beta}_{res} - \beta \bigr) = \sqrt{N} R \bigl(\widehat{\gamma} - \gamma  \bigr)$, we get the deviation form of $\beta$, that is
\begin{align}\label{Deviationform}
\sqrt{N} \bigl(\widehat{\beta}_{res} - \beta \bigr) =& R \left[R^{\prime} \left\{\frac{1}{N} X X^{\prime} \otimes \boldsymbol{I}_m \right\} R \right]^{-1} R^{\prime}  \nonumber  \\ &\times \begin{pmatrix}
\boldsymbol{I}_m & \boldsymbol{0}_{m\times m^2} & \dots & \dots &  \boldsymbol{0} \\ \boldsymbol{0}_{m^2p(1) \times m} & \sum_{k=1}^{\infty}  (C_k(1) \otimes \boldsymbol{I}_m)  & \ddots  & & \vdots \\ \vdots & \ddots & \ddots & \ddots & \vdots \\  \vdots &  & \ddots & \boldsymbol{I}_m &  \boldsymbol{0}_{m \times m^2} \\ \boldsymbol{0} & \dots & \dots & \boldsymbol{0}_{m^2p(S) \times m} & \sum_{k=1}^{\infty}  (C_k(S)  \otimes \boldsymbol{I}_m)
\end{pmatrix}  \\ &\times \frac{1}{\sqrt{N}} \sum\limits_{n=0}^{N-1} \begin{pmatrix} \epsilon_{Sn+1} \nonumber \\ vec\{\epsilon_{Sn+1} [\epsilon_{Sn+1-k}^{\prime} + \nu(1-k)^{\prime}] \} \\ \epsilon_{Sn+2} \\ vec\{\epsilon_{Sn+2}[ \epsilon_{Sn+2-k}^{\prime} +   \nu(2-k)^{\prime}] \} \\ \vdots \\ \epsilon_{Sn+S}  \\ vec\{\epsilon_{Sn+S} [ \epsilon_{Sn+S-k}^{\prime} +   \nu(S-k)^{\prime} ]   \} \end{pmatrix}.
\end{align}
Next, we derive asymptotic properties of the estimator of the periodic covariance matrix $\Sigma_{\epsilon}(s)$. For this purpose, we define $\sigma(s) = vech\{ \Sigma_{\epsilon}(s)\}$,  $\widehat{\sigma}(s) = vech\{\widehat{\Sigma}_{\epsilon} (s)\}$ and $\widetilde{\sigma}(s) = vech\{\widetilde{\Sigma}_{\epsilon}(s) \}$, where $\widehat{\Sigma}_{\epsilon}(s) = \frac{1}{N- k(s)}\sum_{n=0}^{N-1} \widehat{\epsilon}_{Sn+s}  \widehat{\epsilon}_{Sn+s}^{\prime}$ and $\widetilde{\Sigma}_{\epsilon}(s) = \frac{1}{N}\sum_{n=0}^{N-1} \epsilon_{Sn+s}  \epsilon_{Sn+s}^{\prime}$. Here, $\widehat{\Sigma}_{\epsilon}(s)$ is the natural estimator of the seasonal covariance matrix $\Sigma_{\epsilon}(s)$, $s= 1,\dots,S$ based on the PVAR residuals, $\widetilde{\Sigma}_{\epsilon}(s)$ is the corresponding (infeasible) estimator based on the non-observable PVAR error terms. The $vech$-operator stacks the entries on the lower-triangular part of a square matrix columnwise below each other.  Further, let $\sigma = (\sigma(1)^{\prime}, \dots,  \sigma(S)^{\prime})^{\prime}$, $\widehat{\sigma} = (\widehat{\sigma}(1)^{\prime}, \dots,  \widehat{\sigma}(S)^{\prime})^{\prime}$ and $\widetilde{\sigma} = (\widetilde{\sigma}(1)^{\prime}, \dots,  \widetilde{\sigma}(S)^{\prime})^{\prime}$. By standard arguments, it can be easily shown that
\[ \sqrt{N}\bigl( \widetilde{\sigma} - \widehat{\sigma} \bigr) = o_P(1) .\]
Hence, for further calculations, $\widetilde{\sigma}$ can be used and its deviation form can be expressed by
\[ \sqrt{N}\bigl( \widetilde{\sigma} - \sigma \bigr) = \frac{1}{\sqrt{N}}\sum\limits_{n=0}^{N-1} L_{mS} \begin{pmatrix} vec\{ \epsilon_{Sn+1}  \epsilon_{Sn+1}^{\prime} \} - vec\{ \Sigma_{\epsilon}(1)  \} \\ \vdots \\ vec\{ \epsilon_{Sn+S}  \epsilon_{Sn+S}^{\prime} \} - vec\{ \Sigma_{\epsilon}(S)  \} \end{pmatrix} .\]
Together with \eqref{Deviationform}, the joint deviation form can be represented as
\begin{align}\label{A_representation}
\sqrt{N} \begin{pmatrix}
\widehat{\beta}_{res} - \beta  \\  \widetilde{\sigma} - \sigma \end{pmatrix} &= \begin{pmatrix}
    \widehat{Q}^{\beta} R^{\beta} \frac{1}{\sqrt{N}} \sum_{n=0}^{N-1} \begin{pmatrix} \epsilon_{Sn+1} \\ vec\{\epsilon_{Sn+1} [\epsilon_{Sn+1-k}^{\prime} + \nu(1-k)^{\prime}] \} \\ \epsilon_{Sn+2} \\ vec\{\epsilon_{Sn+2}[ \epsilon_{Sn+2-k}^{\prime} +   \nu(2-k)^{\prime}] \} \\ \vdots \\ \epsilon_{Sn+S}  \\ vec\{\epsilon_{Sn+S} [ \epsilon_{Sn+S-k}^{\prime} +   \nu(S-k)^{\prime} ]   \} \end{pmatrix}  \\  \\  \frac{1}{\sqrt{N}}\sum_{n=0}^{N-1} L_{mS} \begin{pmatrix} vec\{ \epsilon_{Sn+1}  \epsilon_{Sn+1}^{\prime} \} - vec\{ \Sigma_{\epsilon}(1)  \} \\ \vdots \\ vec\{ \epsilon_{Sn+S}  \epsilon_{Sn+S}^{\prime} \} - vec\{ \Sigma_{\epsilon}(S)  \} \end{pmatrix}  \end{pmatrix}\\ \nonumber \\ &= A_q + (A-A_q)   \nonumber ,
\end{align}
with 
\begin{align*}
    \widehat{Q}^{\beta} &= R \left[R^{\prime} \left\{\frac{1}{N} X X^{\prime} \otimes \boldsymbol{I}_m \right\} R \right]^{-1} R^{\prime}  \\
    R^{\beta} &= \begin{pmatrix}
\boldsymbol{I}_m & \boldsymbol{0}_{m\times m^2} & \dots & \dots &  \boldsymbol{0} \\ \boldsymbol{0}_{m^2p(1) \times m} & \sum_{k=1}^{\infty}  (C_k(1) \otimes \boldsymbol{I}_m)  & \ddots  & & \vdots \\ \vdots & \ddots & \ddots & \ddots & \vdots \\  \vdots &  & \ddots & \boldsymbol{I}_m &  \boldsymbol{0}_{m \times m^2} \\ \boldsymbol{0} & \dots & \dots & \boldsymbol{0}_{m^2p(S) \times m} & \sum_{k=1}^{\infty}  (C_k(S)  \otimes \boldsymbol{I}_m)
\end{pmatrix},
\end{align*}
where $A$ denotes the right-hand side of \eqref{A_representation} and $A_q$ is defined as the same, but with  $\sum_{k=1}^{\infty}$ replaced by $\sum_{k=1}^{q}$ for some $q \in \mathbb{N}$. Note again that, in this short-hand notation, $R^{\beta}$ contains sums over $k$ that are also applied to the factor $\frac{1}{\sqrt{N}}\sum_{n=0}^{N-1}(\cdots)$.
For proving the Central Limit Theorem of the joint deviation form, we make use of Proposition 6.3.9 in Brockwell \& Davies (1991) and it suffices to show
 
\begin{itemize}
 \item[(a)] $A_q \overset{d}{\to}  \mathcal{N}(0, V_q)$ as $N \to \infty$,
 \item[(b)] $V_q \to V$  as $q \to \infty$,
 \item[(c)] $\forall \delta > 0: \; \; \underset{q \to \infty}{\lim} \underset{N \to \infty}{\limsup} P( |A - A_q|_1 > \delta) = 0  $.
\end{itemize}
In order to prove (a), setting $\widetilde{m} = \frac{m(m+1)}{2}$, we bring $A_q$ to the following form
\begin{align*}
A_q = \widehat{Q}_N  R_q \frac{1}{\sqrt{N}} \sum\limits_{n=0}^{N-1} \begin{pmatrix}
W_{n,q}^{(1)} \\ W_{n,q}^{(2)} 
\end{pmatrix},
\end{align*}
where
\begin{align*}
\widehat{Q}_N = \begin{pmatrix}
R \left[R^{\prime} \{\frac{1}{N} X X^{\prime} \otimes \boldsymbol{I}_m \} R \right]^{-1} R^{\prime} \; \;  \;  & \boldsymbol{0}_{m \sum_{s=1}^S (mp(s)+1) \times S \widetilde{m}} \\ \boldsymbol{0}_{ S\widetilde{m} \times m \sum_{s=1}^S (mp(s)+1)} & \boldsymbol{I}_{S\widetilde{m}} 
\end{pmatrix},
\end{align*}
\begin{align*}
R_q = \begin{pmatrix}
\boldsymbol{I}_m & \boldsymbol{0} & \dots   & \boldsymbol{0} & \dots & \dots & \dots &  \dots &  \dots & \boldsymbol{0} \\  \boldsymbol{0} & C_1(1) \otimes \boldsymbol{I}_m  & \dots & C_q(1)  \otimes \boldsymbol{I}_m & \ddots & &  &  &  & \vdots \\  \boldsymbol{0} & \boldsymbol{0}  &  \dots  &  \boldsymbol{0} & \ddots & \ddots & & & & \vdots  \\  \vdots  &  & &  & \ddots & \boldsymbol{I}_m & \boldsymbol{0}  &  \dots & \boldsymbol{0}  & \boldsymbol{0}  \\  \vdots &  &  &  & & \boldsymbol{0}  & C_1(S) \otimes \boldsymbol{I}_m & \dots  & C_q(S)  \otimes \boldsymbol{I}_m& \boldsymbol{0} \\  \boldsymbol{0} & \dots &\dots   &\dots & \dots& \boldsymbol{0}&\boldsymbol{0} &\dots & \boldsymbol{0} & \boldsymbol{I}_{S\widetilde{m}}
\end{pmatrix} 
\end{align*}
and 
\begin{align*}
\begin{pmatrix}
W_{n,q}^{(1)} \\ W_{n,q}^{(2)} 
\end{pmatrix}  = \begin{pmatrix}
  \begin{pmatrix} \epsilon_{Sn+1} \\ vec\{\epsilon_{Sn+1} [\epsilon_{Sn}^{\prime} + \nu(0)^{\prime}] \}
 \\ \vdots \\ vec\{\epsilon_{Sn+1} [\epsilon_{Sn+1-q}^{\prime} + \nu(1-q)^{\prime}]  \}  \\ \epsilon_{Sn+2} \\ vec\{\epsilon_{Sn+2} [\epsilon_{Sn+1}^{\prime} + \nu(1)^{\prime}]  \}
 \\ \vdots \\ vec\{\epsilon_{Sn+2}  [\epsilon_{Sn+2-q}^{\prime} + \nu(2-q)^{\prime}] \} \\ \vdots \\  \epsilon_{Sn+S} \\ vec\{\epsilon_{Sn+S} [\epsilon_{Sn+S-1}^{\prime} + \nu(S-1)^{\prime}] \}
 \\ \vdots \\ vec\{\epsilon_{Sn+S} [\epsilon_{Sn+S-q}^{\prime} + \nu(S-q)^{\prime}] \}
\end{pmatrix}  \\  L_{mS} \begin{pmatrix} vec\{ \epsilon_{Sn+1}  \epsilon_{Sn+1}^{\prime} \} - vec\{ \Sigma_{\epsilon}(1)  \} \\ \vdots \\ vec\{ \epsilon_{Sn+S}  \epsilon_{Sn+S}^{\prime} \} - vec\{ \Sigma_{\epsilon}(S)  \} \end{pmatrix} \end{pmatrix}
\end{align*}
are of dimension $(m\sum_{s=1}^S (mp(s) +1) + S\widetilde{m}) \times (m\sum_{s=1}^S (mp(s) +1) + S\widetilde{m}), (m\sum_{s=1}^S (mp(s) +1) + S\widetilde{m}) \times (S(m+qm^2)+S\widetilde{m})$ and $S(m+qm^2)+S\widetilde{m}$.
Following Brüggemann et al. (2016), 
 $\frac{1}{N} XX^{\prime} \overset{p}{\to} \Gamma$ as $N \to \infty$ leads to $\widehat{Q}_N \overset{p}{\to} Q$ as $N \to \infty$ ,  where 
\[ Q  = \begin{pmatrix}
R \left[R^{\prime} \{ \Gamma \otimes \boldsymbol{I}_m \} R \right]^{-1} R^{\prime} \; \;  \;  & \boldsymbol{0}_{m \sum_{s=1}^S (mp(s)+1) \times S \widetilde{m}} \\ \boldsymbol{0}_{ S\widetilde{m} \times m \sum_{s=1}^S (mp(s)+1)} & \boldsymbol{I}_{S\widetilde{m}} 
\end{pmatrix} .\]
The matrix $\Gamma = E(\frac{1}{N} XX^{\prime})$ is of dimension $\sum_{s=1}^{S} (mp(s)+1) \times \sum_{s=1}^{S} (mp(s)+1)$ and block-diagonal and can be represented as
\[  \Gamma  = \begin{pmatrix}
\Gamma(1) & 0 & \dots & 0 \\ 0 & \Gamma(2)  & \ddots & \vdots \\ \vdots & \ddots & \ddots & 0 \\  0 & \dots & 0 & \Gamma(S)
\end{pmatrix}  , \]
where $\Gamma(s)$, $s=1,\dots,S$, is given by
\[\Gamma(s) = \begin{pmatrix}
1 & \widetilde{\mu}(s)^{\prime}  \\
\widetilde{\mu}(s)  &  \widetilde{\mu}(s)\widetilde{\mu}(s)^{\prime} + \sum_{k=1}^{\infty} C_k(s) \Sigma_{\epsilon}(s-k)C_k^{\prime}(s) 
\end{pmatrix} \in \mathbb{R}^{(mp(s)+1) \times (mp(s) +1)} \]
with $\widetilde{\mu}(s) = (\mu(s-1)^{\prime}, \dots, \mu(s-p(s))^{\prime})^{\prime}$.
We can derive the closed-form solution of $\Gamma$ using that $XX^{\prime}$ can be written as
\begin{align*}
 XX^{\prime} = \begin{pmatrix}
\sum_{n=0}^{N-1} X_n(1)X_n^{\prime}(1) & 0 & \dots & 0 \\ 0 & \sum_{n=0}^{N-1} X_n(2)X_n^{\prime}(2)  & \ddots & \vdots \\ \vdots & \ddots & \ddots & 0 \\  0 & \dots & 0 & \sum_{n=0}^{N-1} X_n(S)X_n^{\prime}(S)
\end{pmatrix} 
\end{align*}
and that the moving average representation $y_{Sn+s} = \mu(s) +  \sum_{k=0}^{\infty} \Phi_{k}(s)  \epsilon_{Sn+s-k}$ of the PVAR($\boldsymbol{p}$) process exists. Consequently, under Assumption \ref*{Mixing_Assumptions}, we get for $q \in \mathbb{N}$ that
\[ A_q \overset{d}{\to}  \mathcal{N}(0, V_q) \; \text{ as } \; N \to \infty,\]
with
\[ V_q= \begin{pmatrix}
V_{q}^{(1,1)} & V_{q}^{(1,2)}  \\ V_{q}^{(2,1)}  & V^{(2,2)} 
\end{pmatrix}  = QR_q \Omega_q R_q^{\prime} Q^{\prime}. \]
Applying the CLT for innovations stated in Lemma \ref{CLT_Innovation} to obtain $\Omega_q$ and using the matrices $Q$ and $R_q$ given above,  the limiting variance matrices $V_{q}^{(1,1)}, V_{q}^{(2,1)}$ and $V^{(2,2)} $ are of dimension $(m\sum_{s=1}^S(mp(s) +1) \times m\sum_{s=1}^S(mp(s) +1)),  (S\widetilde{m} \times m\sum_{s=1}^S(mp(s) +1))$ and $(S\widetilde{m} \times S\widetilde{m})$ and given by
\begin{align}
V_{q}^{(1,1)} &= R \left[R^{\prime} \{ \Gamma \otimes \boldsymbol{I}_m \} R \right]^{-1} R^{\prime} \; \boldsymbol{\Omega}_q^{V^{(1,1)}} \; \bigl(R \left[R^{\prime} \{ \Gamma \otimes \boldsymbol{I}_m \} R \right]^{-1} R^{\prime} \bigr)^{\prime},   \nonumber   \\
V_q^{(2,1)} &= V_q^{(1,2)\prime} =  \boldsymbol{\Omega}_q^{V^{(2,1)}}  \bigl(R \left[R^{\prime} \{ \Gamma \otimes \boldsymbol{I}_m \} R \right]^{-1} R^{\prime} \bigr)^{\prime},    \label{V_q}   \\
V^{(2,2)} &= L_{mS} \Bigl( \sum\limits_{h=-\infty}^{\infty} \tau_{0,h,0} \Bigr) L_{mS}^{\prime}.   \nonumber
\end{align} 
The matrix $\boldsymbol{\Omega}_q^{V^{(1,1)}}$ is of dimension $(m\sum_{s=1}^S(mp(s) +1) \times m\sum_{s=1}^S(mp(s) +1))$ and can be written as
\[ \boldsymbol{\Omega}_q^{V^{(1,1)}} = \begin{pmatrix}
\boldsymbol{\Omega}_q^{V^{(1,1)}}(s_1,s_2) \\
s_1, s_2 = 1, \dots, S
\end{pmatrix},\]
where $\boldsymbol{\Omega}_q^{V^{(1,1)}}(s_1,s_2)$, $s_1,s_2 = 1,\dots, S$ are $(m^2p(s_1)+m) \times (m^2p(s_2) + m)$ dimensional and can be rewritten as \[\boldsymbol{\Omega}_q^{V^{(1,1)}}(s_1,s_2) = \begin{pmatrix} [\boldsymbol{\Omega}_q^{V^{(1,1)}}(s_1,s_2)]^{(i,j)} \\ i,j =1,2 \end{pmatrix}, \]  
where the submatrices
\begin{align*}
[ \boldsymbol{\Omega}_q^{V^{(1,1)}}(s_1,s_2)]^{(1,1)} &=  \Sigma_{\epsilon}(s_1) \boldsymbol{1}(s_1 = s_2), \\  [ \boldsymbol{\Omega}_q^{V^{(1,1)}}(s_1,s_2)]^{(2,1)} &= \sum\limits_{k=1}^{q} (C_k(s_1) \otimes \boldsymbol{I}_m)\Bigl(\sum\limits_{h=-\infty}^{\infty} \widecheck{\kappa}_{k,h}(s_1,s_2)\Bigr),
 \\ 
 [ \boldsymbol{\Omega}_q^{V^{(1,1)}}(s_1,s_2)]^{(1,2)} &=  \sum\limits_{k=1}^{q} \Bigl(\sum\limits_{h=-\infty}^{\infty} \widecheck{\kappa}^{\prime}_{k,h}(s_2,s_1)\Bigr) (C_k(s_2) \otimes \boldsymbol{I}_m)^{\prime}, \\ [ \boldsymbol{\Omega}_q^{V^{(1,1)}}(s_1,s_2)]^{(2,2)} &=\sum_{k,l=1}^{q} \sum_{h=-\infty}^{\infty} (C_k(s_1) \otimes \boldsymbol{I}_m)[ \widecheck{\tau}_{k,h,l}(s_1, s_2) +  \kappa_{k,h}(s_1, s_2)(\nu(s_2-l) \otimes \boldsymbol{I}_m)^{\prime} \\  &\quad + (\nu(s_1-k) \otimes \boldsymbol{I}_m) \kappa_{l,h}(s_1,s_2)^{\prime}]  (C_l(s_2) \otimes \boldsymbol{I}_m)^{\prime}
 \end{align*}
 are of dimension $(m \times m), (m^2p(s_1) \times m), (m \times m^2p(s_2))$ and $ (m^2p(s_1) \times m^2p(s_2))$, respectively. Notice that
 \begin{align*}
 \widecheck{\kappa}_{k,h}(s_1,s_2) &= \kappa_{k,h}(s_1,s_2) + (\nu(s_1-k) \otimes \boldsymbol{I}_m) \Sigma_{\epsilon}(s_1) \boldsymbol{1}(h=0, s_1 = s_2), \\
 \widecheck{\tau}_{k,h,l}(s_1, s_2) &= \tau_{k,h,l}(s_1, s_2) +(\nu(s_1-k) \otimes \boldsymbol{I}_m) \Sigma_{\epsilon}(s_1) (\nu(s_1-l) \otimes \boldsymbol{I}_m)^{\prime} \boldsymbol{1}(h=0, s_1 = s_2).
 \end{align*}
 The $(S\widetilde{m} \times m\sum_{s=1}^S(mp(s) +1))$ dimensional matrix $\boldsymbol{\Omega}_q^{V^{(2,1)}}$ can be decomposed into
 \[ \boldsymbol{\Omega}_q^{V^{(2,1)}} =  L_{mS} \bigl(\bigl[[\boldsymbol{\Omega}^{V^{(2,1)}}(s)]^{(1)}, \; \; [\boldsymbol{\Omega}_q^{V^{(2,1)}}(s)]^{(2)} \bigr],  \; \; s = 1,\dots,S  \bigr),\]
 where, for $s=1,\dots, S$, the submatrices $[\boldsymbol{\Omega}^{V^{(2,1)}}(s)]^{(1)}$ and $[\boldsymbol{\Omega}_q^{V^{(2,1)}}(s)]^{(2)}$ are of dimension $(Sm^2 \times m)$ and $(Sm^2 \times m^2p(s))$ and given by
 \begin{align*}
 [\boldsymbol{\Omega}^{V^{(2,1)}}(s)]^{(1)} &= \sum\limits_{h=-\infty}^{\infty} \begin{pmatrix} \kappa_{0,h}(1,s) \\  \vdots \\ \kappa_{0,h}(S,s) \end{pmatrix},   \\
 [\boldsymbol{\Omega}_q^{V^{(2,1)}}(s)]^{(2)} &=  \sum\limits_{k=1}^{q} \sum\limits_{h=-\infty}^{\infty} \begin{pmatrix} [\tau_{0,h,k}(1,s) + \kappa_{0,h}(1,s)(\nu(s-k) \otimes \boldsymbol{I}_m)^{\prime}](C_k(s) \otimes \boldsymbol{I}_m)^{\prime} \\ \vdots \\ [\tau_{0,h,k}(S,s) + \kappa_{0,h}(S,s)(\nu(s-k) \otimes \boldsymbol{I}_m)^{\prime}]  (C_k(s) \otimes \boldsymbol{I}_m)^{\prime} \end{pmatrix}.  \end{align*}
Due to Assumption \ref*{Mixing_Assumptions}, part b) follows directly since the fourth order cumulants are assumed to be absolutely summable and the sequence of moving average coefficient matrices $\{\Phi_{k}(s) \}_{k \in \mathbb{N}}$ decreases geometrically to zero for all $s =1,\dots, S$, which implies that $\{C_{k}(s) \}_{k \in \mathbb{N}}$, $s =1,\dots,S$ decreases geometrically to zero as well.  Hence, we have 
\[
V_q \to V =  \begin{pmatrix}
V^{(1,1)} & V^{(1,2)} \\ V^{(2,1)} & V^{(2,2)}
\end{pmatrix}  \; \text{ as } \; q \to \infty.
\]  
Note that the sums $ \sum_{k=1}^{q}$ and $ \sum_{k,l=1}^{q}$ in $V_q^{(1,1)}$, $ V_q^{(1,2)}$ and $V_q^{(2,1)}$ in \eqref{V_q} are replaced by $\sum_{k=1}^{\infty}$ and $ \sum_{k,l=1}^{\infty}$, respectively, to obtain the limiting variances $V^{(1,1)}$, $V^{(2,1)}$ and $V^{(1,2)}$. The variance matrix $V^{(2,2)}$ is already provided in \eqref{V_q}.  To prove c), note that odd parts of the first subvector and the complete second subvector of $A-A_q$ in \eqref{A_representation} are zero. Therefore, it suffices to show part c) only for the even parts of the first subvector of $A-A_q$.  For any constant vector $d\in\mathbb{R}^{m^2p(s)}$ and $\delta >0$, we use the Markov inequality to get 
\begin{align*}
& P\Biggl( \Biggl|\sum\limits_{k=q+1}^{\infty} d^{\prime}(C_k(s) \otimes \boldsymbol{I}_m) \frac{1}{\sqrt{N}} \sum\limits_{n=0}^{N-1} vec\{\epsilon_{Sn+s} [\epsilon_{Sn+s-k}^{\prime} + \nu(s-k)^{\prime}] \} \Biggr| > \delta \Biggr) \\
\leq \; \;  &\frac{1}{\delta^2 N} E \Biggl( \Biggl|\sum\limits_{k=q+1}^{\infty} d^{\prime}(C_k(s) \otimes \boldsymbol{I}_m) \sum\limits_{n=0}^{N-1} vec\{\epsilon_{Sn+s} [\epsilon_{Sn+s-k}^{\prime} + \nu(s-k)^{\prime}] \} \Biggr|^2 \Biggr) \\
= \; \;  &\frac{1}{\delta^2} \sum\limits_{k,l=q+1}^{\infty} d^{\prime}(C_k(s) \otimes \boldsymbol{I}_m) \Biggl( \frac{1}{N} \sum\limits_{n_1, n_2=0}^{N-1} E \Bigl( vec\{\epsilon_{Sn+s} [\epsilon_{Sn+s-k}^{\prime} + \nu(s-k)^{\prime}] \} \\   & \times vec\{\epsilon_{Sn+s} [\epsilon_{Sn+s-k}^{\prime} + \nu(s-k)^{\prime}] \}^{\prime} \Bigr) \Biggr) (C_l(s) \otimes \boldsymbol{I}_m)^{\prime} d \\ \underset{N \to \infty}{\to}   & \; \; \frac{1}{\delta^2} \sum\limits_{k,l=q+1}^{\infty} d^{\prime}(C_k(s) \otimes \boldsymbol{I}_m) \Biggl(  \sum\limits_{h=- \infty}^{\infty} \Bigl[
\widecheck{\tau}_{k,h,l}(s, s) +  \kappa_{k,h}(s, s)(\nu(s-l) \otimes \boldsymbol{I}_m)^{\prime} \\  &+ (\nu(s-k) \otimes \boldsymbol{I}_m) \kappa_{l,h}(s,s)^{\prime} \Bigr] \Biggr)  (C_l(s) \otimes \boldsymbol{I}_m)^{\prime} d  \\
\underset{q \to \infty}{\to} & \; \;  0 
\end{align*}
for all $s=1,\dots,S$.  Hence,  we proved
\[ \sqrt{N} \begin{pmatrix}
\widehat{\beta}_{res} - \beta  \\  \widetilde{\sigma} - \sigma \end{pmatrix}  \overset{d}{\to} \mathcal{N} \Biggl(0, \begin{pmatrix}
V^{(1,1)} & V^{(1,2)} \\ V^{(2,1)} & V^{(2,2)}
\end{pmatrix} \Biggr). \]

\begin{lemma}[CLT for Innovations]\label{CLT_Innovation}

Let $q \in \mathbb{N}$ and define $W_{n,q} = (W^{(1)\prime}_{n,q},  W^{(2)\prime}_{n})^{\prime}$ as
\begin{align*}
W^{(1)}_{n,q}  = \begin{pmatrix}
W^{(1,1)}_{n,q} \\ \vdots \\ W^{(1,S)}_{n,q}
\end{pmatrix}, 
\quad
W^{(2)\prime}_{n}  = 
 L_{mS} \begin{pmatrix} vec\{ \epsilon_{Sn+1}  \epsilon_{Sn+1}^{\prime} \} - vec\{ \Sigma_{\epsilon}(1)  \} \\ \vdots \\ vec\{ \epsilon_{Sn+S}  \epsilon_{Sn+S}^{\prime} \} - vec\{ \Sigma_{\epsilon}(S)  \} \end{pmatrix},
\end{align*}
where, for $s=1,\dots, S$,
\begin{align*}
W^{(1,s)}_{n,q} = \begin{pmatrix}
 \epsilon_{Sn+s} \\ vec\{\epsilon_{Sn+s} [\epsilon_{Sn+s-1}^{\prime} + \nu(s-1)^{\prime}] \}
 \\ \vdots \\ vec\{\epsilon_{Sn+s} [\epsilon_{Sn+s-q}^{\prime} + \nu(s-q)^{\prime}]  \} \end{pmatrix}.
\end{align*}
Then, under Assumption \ref*{Mixing_Assumptions}, we get the following CLT
\[ \frac{1}{\sqrt{N}} \sum_{n=0}^{N-1} \begin{pmatrix} W_{n,q}^{(1)} \\  W_{n,q}^{(2)} \end{pmatrix} \overset{d}{\to}  \mathcal{N}\bigl(\boldsymbol{0},   \boldsymbol{\Omega}_q  \bigr), \]
with $(S(m+qm^2) +S\widetilde{m}) \times  (S(m+qm^2) +S\widetilde{m})$-dimensional limiting variance matrix $\boldsymbol{\Omega}_q$.
\end{lemma} 

\begin{proof}
To prove the statement, we write the limiting variance matrix $\boldsymbol{\Omega}_q$ as
\[ \boldsymbol{\Omega}_q = \begin{pmatrix} \boldsymbol{\Omega}_q^{(1)} & \boldsymbol{\Omega}_{q}^{(2,1)\prime} \\ \boldsymbol{\Omega}_{q}^{(2,1)} & \boldsymbol{\Omega}^{(2)} \end{pmatrix}, \] 
where the matrices $\boldsymbol{\Omega}_q^{(1)}, \boldsymbol{\Omega}_q^{(2,1)}$ and $\boldsymbol{\Omega}^{(2)}$  are of dimension $(S(m+qm^2) \times  S(m+qm^2)), (S\widetilde{m} \times  S(m+qm^2))$ and $(S\widetilde{m} \times S\widetilde{m})$.  
The matrix $\boldsymbol{\Omega}_q^{(1)}$ can be represented as 
 \begin{align*}
\boldsymbol{\Omega}_q^{(1)} = \begin{pmatrix}
\boldsymbol{\Omega}_q^{(1)}(s_1,s_2) \\
s_1, s_2 = 1, \dots, S
\end{pmatrix},
\end{align*}
where $\boldsymbol{\Omega}_q^{(1)}(s_1, s_2)$ can be decomposed as
\begin{align*}
 \boldsymbol{\Omega}_q^{(1)}(s_1,s_2) =  \begin{pmatrix}
 \Xi_{1}^{(1,1)}(s_1,s_2) &   \Xi_{1,q}^{(1,2)}(s_1,s_2)   \\  \Xi_{1,q}^{(2,1)}(s_1,s_2)&   \Xi_{1,q}^{(2,2)}(s_1,s_2)
\end{pmatrix} \in \mathbb{R}^{(m+qm^2) \times (m+qm^2)}, \; \; \; s_1,s_2= 1, \dots, S.
\end{align*}
For $s_1,s_2 =1,\dots,S $, the submatrices $\Xi_{1}^{(1,1)}(s_1,s_2) ,  \Xi_{1,q}^{(2,1)}(s_1,s_2), \Xi_{1,q}^{(1,2)}(s_1,s_2) $ and $\Xi_{1,q}^{(2,2)}(s_1,s_2)$ are of dimension $(m \times m), (qm^2 \times m)$, $(m \times qm^2)$ and $(qm^2 \times qm^2)$ and given by
\begin{align*}
  \Xi_{1}^{(1,1)}(s_1,s_2) &= \sum_{h= -\infty}^{\infty} Cov(\epsilon_{Sn+s_1},  \epsilon_{S(n-h)+s_2}) = \Sigma_{\epsilon}(s_1) \boldsymbol{1}(s_1 = s_2) , \\ \Xi_{1,q}^{(2,1)}(s_1, s_2) &= \sum\limits_{h=-\infty}^{\infty} \begin{pmatrix} 
Cov \bigl(  vec\{\epsilon_{Sn+s_1} [\epsilon_{Sn+s_1-1}^{\prime} + \nu(s_1-1)^{\prime}],  \epsilon_{S(n-h)+s_2} \} \bigr)   \\ \vdots \\   Cov \bigl(  vec\{\epsilon_{Sn+s_1} [\epsilon_{Sn+s_1-q}^{\prime} + \nu(s_1-q)^{\prime}] \},  \epsilon_{S(n-h)+s_2} \bigr) 
\end{pmatrix}\\ &= \sum\limits_{h=-\infty}^{\infty} \begin{pmatrix} 
 \kappa_{1,h}(s_1,s_2) + (\nu(s_1-1) \otimes \boldsymbol{I}_m) \Sigma_{\epsilon}(s_1) \boldsymbol{1}(h=0, s_1 = s_2) \\ \vdots \\   \kappa_{q,h}(s_1,s_2) + (\nu(s_1-q) \otimes \boldsymbol{I}_m) \Sigma_{\epsilon}(s_1) \boldsymbol{1}(h=0, s_1 = s_2)  
\end{pmatrix}, \\ \Xi_{1,q}^{(1,2)}(s_1, s_2) &= \sum\limits_{h=-\infty}^{\infty} \begin{pmatrix} 
Cov \bigl(  vec\{\epsilon_{Sn+s_2} [\epsilon_{Sn+s_2-1}^{\prime} + \nu(s_2-1)^{\prime}],  \epsilon_{S(n-h)+s_1} \} \bigr)   \\ \vdots \\   Cov \bigl(  vec\{\epsilon_{Sn+s_2} [\epsilon_{Sn+s_2-q}^{\prime} + \nu(s_2-q)^{\prime}] \},  \epsilon_{S(n-h)+s_1} \bigr) 
\end{pmatrix}^{\prime} \\ &= \sum\limits_{h=-\infty}^{\infty} \begin{pmatrix} 
 \kappa_{1,h}(s_2,s_1) + (\nu(s_1-1) \otimes \boldsymbol{I}_m) \Sigma_{\epsilon}(s_1) \boldsymbol{1}(h=0, s_1 = s_2) \\ \vdots \\   \kappa_{q,h}(s_2,s_1) + (\nu(s_1-q) \otimes \boldsymbol{I}_m) \Sigma_{\epsilon}(s_1) \boldsymbol{1}(h=0, s_1 = s_2)  
\end{pmatrix}^{\prime}, \\
\Xi_{1,q}^{(2,2)}(s_1,s_2) &= \sum\limits_{h=-\infty}^{\infty} \begin{pmatrix} Cov \bigl(  vec\{\epsilon_{Sn+s_1} [\epsilon_{Sn+s_1-k}^{\prime} + \nu(s_1-k)^{\prime}] \} ,  \\  vec \{\epsilon_{S(n-h)+s_2}[\epsilon_{S(n-h)+s_2-l}^{\prime}  + \nu(s_2-l)^{\prime}] \} \bigr)
 \\ k,l = 1, \dots, q
\end{pmatrix} \\ &= \sum\limits_{h=-\infty}^{\infty} \begin{pmatrix}
\tau_{k,h,l}(s_1, s_2) +  \kappa_{k,h}(s_1, s_2)(\nu(s_2-l) \otimes \boldsymbol{I}_m)^{\prime}  + (\nu(s_1-k) \otimes \boldsymbol{I}_m) \kappa_{l,h}(s_1,s_2)^{\prime} \\ + (\nu(s_1-k) \otimes \boldsymbol{I}_m) \Sigma_{\epsilon}(s_1) (\nu(s_1-l) \otimes \boldsymbol{I}_m)^{\prime} \boldsymbol{1}(h=0, s_1 = s_2) \\ k,l = 1, \dots, q
\end{pmatrix}.
\end{align*}
The matrix $\boldsymbol{\Omega}^{(2)}$ is given by 
\begin{align*}
\boldsymbol{\Omega}^{(2)} &=  L_{mS} \left( \sum_{h=-\infty}^{\infty} Cov \left( \begin{pmatrix} vec\{ \epsilon_{Sn+1}  \epsilon_{Sn+1}^{\prime} \}  \\ \vdots \\ vec\{ \epsilon_{Sn+S}  \epsilon_{Sn+S}^{\prime} \}  \end{pmatrix}, \begin{pmatrix} vec\{ \epsilon_{S(n-h)+1}  \epsilon_{S(n-h)+1}^{\prime} \}  \\ \vdots \\ vec\{ \epsilon_{S(n-h)+S}  \epsilon_{S(n-h)+S}^{\prime} \}  \end{pmatrix} \right) \right)  L_{mS}^{\prime}\\ &= L_{mS} \left( \sum\limits_{h=-\infty}^{\infty} \tau_{0,h,0} \right) L_{mS}^{\prime}, \; \; \; \in \mathbb{R}^{S\widetilde{m} \times S\widetilde{m}}.
\end{align*}
The $(S\widetilde{m} \times  S(m+qm^2))$-dimensional matrix $ \boldsymbol{\Omega}_q^{(2,1)}$ can be represented by
\[
\boldsymbol{\Omega}_q^{(2,1)} =  L_{mS} \left(\left[\Xi_{2,1}^{(1)}(s), \; \; \Xi_{2,1,q}^{(2)}(s)\right],  \; \; s = 1,\dots,S  \right),
\]
where, for $s=1,\dots, S$, the submatrices $\Xi_{2,1}^{(1)}(s)$ and $\Xi_{2,1,q}^{(2)}(s)$ are of dimension $(Sm^2 \times m)$ and $(Sm^2 \times qm^2)$ and given by
\begin{align*}
\Xi_{2,1}^{(1)}(s) &=    \sum\limits_{h=-\infty}^{\infty} Cov \left(  \begin{pmatrix} vec\{ \epsilon_{Sn+1}  \epsilon_{Sn+1}^{\prime} \}  \\ \vdots \\ vec\{ \epsilon_{Sn+S}  \epsilon_{Sn+S}^{\prime} \}  \end{pmatrix},   \epsilon_{S(n-h)+s} \right)  =   \sum\limits_{h=-\infty}^{\infty} \begin{pmatrix} \kappa_{0,h}(1,s) \\  \vdots \\ \kappa_{0,h}(S,s) \end{pmatrix}   \\
\Xi_{2,1,q}^{(2)}(s) &=  \sum\limits_{h=-\infty}^{\infty} Cov \left(  \begin{pmatrix} vec\{ \epsilon_{Sn+1}  \epsilon_{Sn+1}^{\prime} \}  \\ \vdots \\ vec\{ \epsilon_{Sn+S}  \epsilon_{Sn+S}^{\prime} \}  \end{pmatrix},  \begin{pmatrix} vec\{\epsilon_{S(n-h)+s} [\epsilon_{S(n-h)+s-1}^{\prime} + \nu(s-1)^{\prime}] \}
 \\ \vdots \\ vec\{\epsilon_{S(n-h)+s} [\epsilon_{S(n-h)+s-q}^{\prime} + \nu(s-q)^{\prime}]  \} \end{pmatrix} \right)  \\
&=  \sum\limits_{h=-\infty}^{\infty} \begin{pmatrix} \tau_{0,h,k}(1,s) + \kappa_{0,h}(1,s)(\nu(s-k) \otimes \boldsymbol{I}_m)^{\prime} &  \\ \vdots &, \; \; \;  k= 1,\dots, q \\ \tau_{0,h,k}(S,s) + \kappa_{0,h}(S,s)(\nu(s-k) \otimes \boldsymbol{I}_m)^{\prime} &  \end{pmatrix}.  
\end{align*}
\end{proof}

\section{Joint CLT with Periodic Strong White Noise}
\label{sec_iid_joint_clt}

\begin{corollary}[Joint CLT with Periodic Strong White Noise]\label{iid_joint_clt}\
Under Assumption \ref*{Mixing_Assumptions} with $(ii)$ and $(iii)$ replaced by $(ii)^{\prime}$ and $(iii)^{\prime}$, we have
\[
\sqrt{N} \begin{pmatrix}
\widehat{\beta}_{res} - \beta  \\  \widehat{\sigma} - \sigma \end{pmatrix}  \overset{d}{\to} \mathcal{N} (\boldsymbol{0},  V_{iid}),   \quad   V_{iid} = \begin{pmatrix}
V^{(1,1)}_{iid} & V^{(1,2)}_{iid}  \\ V^{(2,1)}_{iid} & V^{(2,2) }_{iid}
\end{pmatrix},
\]  
where, using  $\frac{1}{N} X X^{\prime}  \overset{p}{\to} \Gamma$, the submatrices $V^{(1,1)}_{iid},V^{(2,1)}_{iid},V^{(1,2)}_{iid}$ and $V^{(2,2)}_{iid}$ are given by
\begin{align*}
V^{(1,1)}_{iid} &= R \left[R^{\prime} \{\Gamma \otimes \boldsymbol{I}_m \} R \right]^{-1} R^{\prime} \; \boldsymbol{\Omega}^{V^{(1,1)}}_{iid} \bigl( R \left[R^{\prime} \{\Gamma \otimes \boldsymbol{I}_m \} R \right]^{-1} R^{\prime} \bigr)^{\prime}, \\ 
V^{(2,1)}_{iid}  &= V^{(1,2) \prime}_{iid}  =\boldsymbol{\Omega}^{V^{(2,1)}}_{iid}  \bigl(R \left[R^{\prime} \{ \Gamma \otimes \boldsymbol{I}_m \} R \right]^{-1} R^{\prime} \bigr)^{\prime}, \\
V^{(2,2)}_{iid} &= L_{mS} \{ \tau_{0,0,0} \} L_{mS}^{\prime}, 
\end{align*}
where 
\[
\boldsymbol{\Omega}^{V^{(1,1)}}_{iid} = 
 \begin{pmatrix}
[\boldsymbol{\Omega}^{V^{(1,1)}}_{iid}(1)]^{(1,1)} & [\boldsymbol{\Omega}^{V^{(1,1)}}_{iid}(1)]^{(2,1)\prime} &  \boldsymbol{0} & \dots &  \boldsymbol{0} \\   [\boldsymbol{\Omega}^{V^{(1,1)}}_{iid}(1)]^{(2,1)}& [\boldsymbol{\Omega}^{V^{(1,1)}}_{iid}(1)]^{(2,2)} &   &\ddots & \vdots \\  \boldsymbol{0}  & & \ddots &  &  \boldsymbol{0} \\  \vdots & \ddots  &  & [\boldsymbol{\Omega}^{V^{(1,1)}}_{iid}(S)]^{(1,1)} & [\boldsymbol{\Omega}^{V^{(1,1)}}_{iid}(S)]^{(2,1)\prime} \\ \boldsymbol{0} & \dots & \boldsymbol{0} & [\boldsymbol{\Omega}^{V^{(1,1)}}_{iid}(S)]^{(2,1)}  & [\boldsymbol{\Omega}^{V^{(1,1)}}_{iid}(S)]^{(2,2)}
\end{pmatrix}
\] 
with, for $s=1,\dots, S$,
\begin{align*}
 [\boldsymbol{\Omega}^{V^{(1,1)}}_{iid}(s)]^{(1,1)} &= \Sigma_{\epsilon}(s), \\
 [\boldsymbol{\Omega}^{V^{(1,1)}}_{iid}(s)]^{(2,1)}&=\sum\limits_{k=1}^{\infty} (C_k(s)\otimes \boldsymbol{I}_m)[\nu(s-k) \otimes \Sigma_{\epsilon}(s)], \\
[\boldsymbol{\Omega}^{V^{(1,1)}}_{iid}(s)]^{(2,2)} &= \sum\limits_{k=1}^{\infty} (C_k(s)\otimes \boldsymbol{I}_m)[ \tau_{k,0,k}(s,s)+ [\nu(s-k)\nu(s-k)^{\prime}  \otimes \Sigma_{\epsilon}(s)]] (C_k(s)\otimes \boldsymbol{I}_m)^{\prime}.
\end{align*}
The $(S\widetilde{m} \times m\sum_{s=1}^S(mp(s) +1))$ dimensional matrix $\boldsymbol{\Omega}_{iid}^{V^{(2,1)}}$ can be decomposed into
 \[ \boldsymbol{\Omega}^{V^{(2,1)}}_{iid} =  L_{mS} \bigl(\bigl[[\boldsymbol{\Omega}_{iid}^{V^{(2,1)}}(s)]^{(1)}, \; \; [\boldsymbol{\Omega}_{iid}^{V^{(2,1)}}(s)]^{(2)} \bigr],  \; \; s = 1,\dots,S  \bigr)\] 
 where, for $s=1,\dots, S$, the submatrices $[\boldsymbol{\Omega}_{iid}^{V^{(2,1)}}(s)]^{(1)}$ and $[\boldsymbol{\Omega}_{iid}^{V^{(2,1)}}(s)]^{(2)}$ are of dimension $(Sm^2 \times m)$ and $(Sm^2 \times m^2p(s))$ and given by
 \begin{align*}
 [\boldsymbol{\Omega}_{iid}^{V^{(2,1)}}(s)]^{(1)} &= e_s \otimes  \kappa_{0,0}(s,s),   \\
 [\boldsymbol{\Omega}_{iid}^{V^{(2,1)}}(s)]^{(2)} &= e_s \otimes \left[\kappa_{0,0}(s,s) \sum\limits_{k=1}^{\infty}  (\nu(s-k) \otimes \boldsymbol{I}_m)^{\prime}(C_k(s) \otimes \boldsymbol{I}_m)^{\prime}\right],  \end{align*}
where $e_s$ is an $S$-dimensional unit vector filled with 1 at the $s$-th component, while all remaining components are zero. 
\end{corollary}
As a direct consequence of Corollary \ref{iid_joint_clt}, for periodic strong white noise, we generally do not obtain block-diagnoality of $V_{iid}$. That is, $\widehat{\beta}_{res}$ and $\widehat{\sigma}$ are asymptotically dependent if the third moments $\kappa_{0,0}(s,s), s = 1\dots, S$ of the error term are unequal to zero. Hence, under periodic strong white noise, we generally find that $\widehat{\beta}_{res}(s_1)$ and $\widehat{\sigma}(s_2)$ are asymptotically dependent for $s_1 = s_2$, while they are asymptotically independent for $s_1 \neq s_2$ with $s_1, s_2 = 1,\dots, S$.  

In the case that the periodic intercept $\nu(s)$ is zero for all seasons $s = 1,\dots,S$ and no intercept is estimated, the matrices $\boldsymbol{\Omega}^{V^{(2,1)}}_{iid}$, and hence $V^{(2,1)}_{iid}$, become zero matrices leading to block-diagonality of $V_{iid}$ and therefore asymptotic independence of $\widehat{\beta}_{res}$ and $\widehat{\sigma}$. Further, we obtain the same results as Ursu \& Duchesne (2009), 
including asymptotic independence among the seasonal PVAR estimators $\widehat{\beta}_{res}(s)$, $s=1,\dots,S$, if no restrictions \emph{across} seasons are imposed. The block-diagonality of the matrix $\tau_{0,0,0}$ implies also asymptotic independence among the seasonal covariance estimators $ \widehat{\sigma}(s)$, $s =1,\dots, S$.

\section{Proof of Theorem \ref*{CLT_IR_SIR}}
\label{SIR_Proof}

A higher-dimensional time-invariant, i.e., (second-order) stationary, VAR($P$) representation of the PVAR($\boldsymbol{p}$) process in (\ref*{PVAR_model}) is given by the model equation (Franses \& Paap, 2004)
\begin{align}\label{VAR_model}
\boldsymbol{A}_0 Y_n = \nu+  \boldsymbol{A}_1 Y_{n-1} + \dots + \boldsymbol{A}_P  Y_{n-P} + \xi_n,	\quad	n\in\mathbb{Z},
\end{align}
where 
    $Y_n = (y_{Sn+1}^{\prime}, y_{Sn+2}^{\prime}, \dots , y_{Sn+S}^{\prime}  )^{\prime}$, $\xi_n = (\epsilon_{Sn+1}^{\prime}, \epsilon_{Sn+2}^{\prime}, \dots, \epsilon_{Sn+S}^{\prime})^{\prime}$, $\nu = (\nu(1)^{\prime}, \dots, \nu(S)^{\prime} )^{\prime}$
are of dimension $Sm$, respectively. The coefficient matrices $\boldsymbol{A}_0, \boldsymbol{A}_1 \dots, \boldsymbol{A}_P$ contain the PVAR coefficient matrices from (\ref*{PVAR_model}) and are of dimension ($Sm \times Sm$). Precisely, it holds
\begin{align*} 
\boldsymbol{A}_0 &=  \begin{pmatrix}
\boldsymbol{I}_m & \boldsymbol{0} & \dots & \boldsymbol{0} \\ 
\scalebox{0.6}[1]{\( - \)}A_{1}(2) & \boldsymbol{I}_m & \ddots & \vdots \\ \vdots & \ddots & \ddots 
& \boldsymbol{0} \\  \scalebox{0.6}[1]{\( - \)}
A_{S\scalebox{0.6}[0.6]{\( - \)}1}(S) 
& \dots & \scalebox{0.6}[1]{\( - \)}
A_{1}(S) & \boldsymbol{I}_m 
\end{pmatrix} ,
\boldsymbol{A}_{i} = \begin{pmatrix}
A_{Si}(1) & A_{Si\scalebox{0.6}[0.6]{\( - \)}1}(1) & \dots & A_{Si\scalebox{0.6}[0.6]{\( - \)}S\scalebox{0.6}[0.6]{\( + \)}1}(1)\\  A_{Si \scalebox{0.6}[0.6]{\( + \)}1}(2) & A_{Si }(2)  & \ddots  & \vdots \\ \vdots & \ddots & \ddots &  A_{Si \scalebox{0.6}[0.6]{\( - \)}1}(S\scalebox{0.6}[1]{\( - \)}1)\\  A_{Si \scalebox{0.6}[0.6]{\( + \)}S\scalebox{0.6}[0.6]{\( - \)}1}(S) & \dots & A_{Si \scalebox{0.6}[0.6]{\( + \)}1}(S) &  A_{Si}(S) 
\end{pmatrix}
\end{align*}
for $i=1,\ldots,P$. 
The VAR model order $P$ of the process $\{Y_n\}_{n\in\mathbb{Z}}$ is given by $P = \lceil p/S \rceil$, where $\lceil . \rceil$ is the ceiling function.  Since the number of unit roots of a PVAR($\boldsymbol{p}$) process equals the one in its higher-dimensional VAR($P$) representation (Franses \& Paap, 2004), the periodic stationarity and causality condition for a PVAR($\boldsymbol{p}$) process is satisfied if and only if its higher-dimensional VAR($P$) form is stationary and causal. As $\det(\boldsymbol{A}_0)=1$, according to Lütkepohl (2005), model \eqref{VAR_model} is equivalent to 
\begin{align*}
Y_n = \boldsymbol{A}_0^{-1} \nu+ \boldsymbol{A}_0^{-1} \boldsymbol{A}_1 Y_{n-1} + \dots +\boldsymbol{A}_0^{-1} \boldsymbol{A}_P  Y_{n-P} + \boldsymbol{A}_0^{-1} \xi_n, \quad   n \in \mathbb{Z},  
\end{align*}
such that $\{Y_n\}_{n \in \mathbb{Z}}$ is stationary and causal if and only if
\begin{align}\label{PeriodStatCon}
\det(\boldsymbol{I}_{Sm} - \boldsymbol{A}_0^{-1}\boldsymbol{A}_1 z - \dots - \boldsymbol{A}_0^{-1}\boldsymbol{A}_P z^P) \neq 0 
\end{align}
for all $z\in\mathbb{C}$ with $|z|\leq 1$.

There is also a direct link between PVAR impulse responses and the impulse responses of its higher-dimensional VAR representation. If the stationarity condition \eqref{PeriodStatCon} holds, the VAR($P$) in \eqref{VAR_model} has an MA representation 
\begin{align*}
 Y_n = \boldsymbol{A}_0^{-1} \nu + \sum\limits_{h=0}^{\infty} \Pi_h \boldsymbol{A}_0^{-1} \xi_{n-h},	\quad	n\in\mathbb{Z},
 \end{align*}
 where the moving average coefficient matrices $\Pi_h$ of the VAR process $\{Y_n\}_{n\in\mathbb{Z}}$ defined by \eqref{VAR_model} can be represented as $ \Pi_h = J\boldsymbol{A}^{h} J^{\prime}$, $h \in \mathbb{N}_{0} $ with $(Sm \times SmP)$ matrix $J = [ \boldsymbol{I}_{Sm}, \boldsymbol{0},\dots,\boldsymbol{0}]$ and the VAR($P$) companion matrix $\boldsymbol{A}$ is of dimension $(SmP \times SmP)$. The VAR($P$) companion matrix $\boldsymbol{A}$ is defined as the parameter matrix of the higher-dimensional VAR(1) representation of the VAR($P$). The VAR impulse responses $\Pi_{h}^{IR}$ are then obtained by 
 \begin{align}\label{VAR_IR}
  \Pi_{h}^{IR} = \Pi_h \boldsymbol{A}_0^{-1} = \begin{pmatrix}
 \Phi_{Sh}^{IR}(1) &  \Phi_{Sh-1}^{IR}(2) & \dots &  \Phi_{Sh-S+1}^{IR}(S) \\  \Phi_{Sh+1}^{IR}(1) &  \Phi_{Sh}^{IR}(2) & \ddots & \vdots \\ \vdots & \ddots & \ddots &  \Phi_{Sh-1}^{IR}(S) \\   \Phi_{Sh+S-1}^{IR}(1) & \dots &  \Phi_{Sh+1}^{IR}(S-1) &  \Phi_{Sh}^{IR}(S) \end{pmatrix}, \;\;  h \in \mathbb{N}_{0}.
\end{align} 
In order to prove the CLT of the (structural) impulse responses stated in Theorem \ref*{CLT_IR_SIR}, we make use of Proposition 3.6 of Lütkepohl (2005) 
and Theorem \ref*{theo_joint_clt}. 
For $h\in\mathbb{N}_0$, from the relation $\Pi_{h}^{IR} = \Pi_{h} \boldsymbol{A}_0^{-1}$, we immediately get that
the first derivative of $\Pi_{h}^{IR} $ with respect to $\beta$ is given by 
\begin{align*}
G_h= \frac{\partial vec \{\Pi_{h}^{IR} \} }{\partial \beta^{\prime}} = \frac{\partial vec \{\Pi_{h} \boldsymbol{A}_0^{-1} \} }{\partial \beta^{\prime}} = (\boldsymbol{I}_{Sm} \otimes \Pi_{h}) \frac{\partial vec \{ \boldsymbol{A}_0^{-1} \} }{\partial \beta^{\prime}}  +
\left(\left(\boldsymbol{A}_0^{-1}\right)^\prime \otimes \boldsymbol{I}_{Sm} \right) \frac{\partial vec \{ \Pi_{h} \} }{\partial \beta^{\prime}},
\end{align*}
where 
\begin{align*}
\frac{\partial vec \{ \boldsymbol{A}_0^{-1} \} }{\partial \beta^{\prime}}  = -\left(\left(\boldsymbol{A}_0^{-1}\right)^\prime \otimes \boldsymbol{A}_0^{-1} \right)) \frac{\partial vec \{ \boldsymbol{A}_0 \} }{\partial \beta^{\prime}} = -\left(\left(\boldsymbol{A}_0^{-1}\right)^\prime \otimes \boldsymbol{A}_0^{-1} \right) J_{ \boldsymbol{A}_0}
\end{align*}  
and 
\begin{align*}
\frac{\partial vec\{\Pi_{h}\} }{\partial \beta^{\prime}} &= \frac{\partial vec\{J\boldsymbol{A}^h J^{\prime} \} }{\partial \beta^{\prime}}  = (J \otimes J) \frac{\partial vec\{\boldsymbol{A}^h\} }{\partial \beta^{\prime}}=  (J \otimes J) \Bigl[ \sum\limits_{i=0}^{h-1}  (\boldsymbol{A}^{\prime})^{h-1-i} \otimes \boldsymbol{A}^h \Bigr] \frac{\partial vec\{\boldsymbol{A}\} }{\partial \beta^{\prime}} \\
 &=  \Bigl[ \sum\limits_{i=0}^{h-1}  J(\boldsymbol{A}^{\prime})^{h-1-i} \otimes J\boldsymbol{A}^h \Bigr] J_{ \boldsymbol{A}}.
\end{align*}
Here, we define  $J_{ \boldsymbol{A}_0} = \frac{\partial vec \{ \boldsymbol{A}_0 \} }{\partial \beta^{\prime}}$  and $J_{ \boldsymbol{A}} = \frac{\partial vec \{ \boldsymbol{A} \} }{\partial \beta^{\prime}}$. Note that $\Pi_{h}$ can be represented as  $\Pi_{h}= J\boldsymbol{A}^h J^{\prime} $ with $J = [ \boldsymbol{I}_{Sm}, \boldsymbol{0}_{Sm},\dots,\boldsymbol{0}_{Sm}]$ and VAR($P$) companion matrix $\boldsymbol{A}$.
We can represent the companion matrix $\boldsymbol{A}$ of the time-invariant VAR($P$) as the product $\boldsymbol{A} = \boldsymbol{A}_0^{*} \boldsymbol{A}^{*} $ where
\begin{align*}
\boldsymbol{A}_0^{*} = \begin{pmatrix}
\boldsymbol{A}_0^{-1} & \boldsymbol{0} & \dots & \boldsymbol{0} \\  \boldsymbol{0}  & \boldsymbol{I}_{Sm} & \ddots & \vdots \\ \vdots & \ddots  & \ddots &  \boldsymbol{0} \\   \boldsymbol{0} & \dots &  \boldsymbol{0} &  \boldsymbol{I}_{Sm}
\end{pmatrix}   \quad \text{and}    \quad \boldsymbol{A}^{*} = \begin{pmatrix}
\boldsymbol{A}_1&\boldsymbol{A}_2 & \dots &  \boldsymbol{A}_P \\  \boldsymbol{I}_{Sm}  & \boldsymbol{0} & \dots & \boldsymbol{0}  \\   & \ddots  & \ddots & \vdots \\   \boldsymbol{0} &   &  \boldsymbol{I}_{Sm} &  \boldsymbol{0}
\end{pmatrix}
\end{align*}
and use this representation in order to calculate $J_{ \boldsymbol{A}}$.  Hence, $G_h$ can be represented as
\[  G_h = -(\boldsymbol{I}_{Sm} \otimes \Pi_{h})\left(\left(\boldsymbol{A}_0^{-1}\right)^\prime \otimes \boldsymbol{A}_0^{-1} \right) J_{ \boldsymbol{A}_0}  +
\left(\left(\boldsymbol{A}_0^{-1}\right)^\prime \otimes \boldsymbol{I}_{Sm} \right)  \Bigl[ \sum\limits_{i=0}^{h-1}  J(\boldsymbol{A}^{\prime})^{h-1-i} \otimes J\boldsymbol{A}^h \Bigr] J_{ \boldsymbol{A}}.
\]
From $\Psi_{h}^{SIR} = \Pi_h^{IR}\boldsymbol{H}_0$, we immediately get
\begin{align*}
F_h &= \frac{\partial vec \{\Psi_{h}^{SIR}\} }{\partial \beta^{\prime}} = \frac{\partial vec \{\Pi_h^{IR}\boldsymbol{H}_0\} }{\partial \beta^{\prime}} = (\boldsymbol{I}_{Sm} \otimes \Pi_{h}^{IR}) \frac{\partial vec \{ \boldsymbol{H}_0\} }{\partial \beta^{\prime}}  +
(\boldsymbol{H}_0^{\prime} \otimes \boldsymbol{I}_{Sm} ) G_h \\
&= (\boldsymbol{H}_0^{\prime} \otimes \boldsymbol{I}_{Sm} ) G_h.
\end{align*}
Further, define $\boldsymbol{H} :=  \frac{\partial vec \{\boldsymbol{H}_0 \} }{\partial \sigma^{\prime}} =  L_{mS}^{\prime} \frac{\partial vech \{\boldsymbol{H}_0\} }{\partial \sigma^{\prime}}$.   Then $D_h =  \frac{\partial vec \{\Psi_{h}^{SIR}\} }{\partial \sigma^{\prime}} $ can be represented as 
\begin{align*}
D_h &=  \frac{\partial vec \{\Psi_{h}^{SIR}\} }{\partial \sigma^{\prime}} = \frac{\partial vec \{\Pi_h^{IR}\boldsymbol{H}_0 \} }{\partial \sigma^{\prime}} = (\boldsymbol{I}_{Sm} \otimes \Pi_{h}^{IR}) \boldsymbol{H}  +
(\boldsymbol{H}_0^{\prime} \otimes \boldsymbol{I}_{Sm} ) \frac{\partial vec \{\Pi_{h}^{IR}\} }{\partial \sigma^{\prime}} \\ &= (\boldsymbol{I}_{Sm} \otimes \Pi_{h}^{IR}) \boldsymbol{H}.
\end{align*}

\section{Proof of Theorem \ref*{Mixing_bootstrap_con}}\label{Bootstrap_Con_Proof}

We proceed in two steps to prove consistency of the residual-based seasonal block bootstrap, which is discussed in Theorem \ref*{Mixing_bootstrap_con}. First, we make use of Lemma \ref{Equivalence_BE} to show that it is asymptotically equivalent, if we replace the estimator $\widehat{\beta}_{res}$ by the true PVAR parameter $\beta$ when conducting the residual-based seasonal block bootstrap method described in Section \ref*{Subsubsection_Bootstrapscheme_weak}. Second we transfer the proof of Theorem \ref*{theo_joint_clt} to the bootstrap world and use Lemma \ref{CLT_Boot_Innovations} to derive the asymptotic normality results corresponding to the results obtained in Theorem \ref*{theo_joint_clt}. 

First, we introduce some notation. Let $\widecheck{\beta}^*-\widecheck{\beta} := R\left[R^{\prime} \{\widecheck{X}^* \widecheck{X}^{* \prime} \otimes I_m \} R \right]^{-1} R^{\prime} \{ \widecheck{X}^* \otimes I_m \} \widecheck{e}^*$, \\ $\widecheck{\sigma}^* -  \widecheck{\sigma} = (vech \{\widecheck{\Sigma}_{\epsilon}^*(1) \}^{\prime}, \dots, vech \{\widecheck
{\Sigma}_{\epsilon}^*(S) \}^{\prime})^{\prime} -  E^*[(vech \{\widecheck
{\Sigma}_{\epsilon}^*(1) \}^{\prime}, \dots, vech \{\widecheck
{\Sigma}_{\epsilon}^*(S) \}^{\prime})^{\prime}]$ with $\widecheck
{\Sigma}_{\epsilon}^*(s) = \frac{1}{N} \sum_{n=0}^{N-1}\widehat{\widecheck
{\epsilon}}_{Sn+s}^*\widehat{\widecheck{\epsilon}}_{Sn+s}^{*\prime}, s=1,\dots, S$.

The pre-sample values $\widecheck{y}^*_{s^{\prime}},\dots,\widecheck{y}^*_0$ are set to $y_{s^{\prime}},\dots,y_0$ and $\widecheck{y}^*_{1},\dots,\widecheck
{y}^*_{SN}$ are generated using
$\widecheck{y}^*_{Sn+s} =  \nu(s) +  A_{1}(s) \widecheck{y}^*_{Sn+s-1} + \dots +  A_{p(s)}(s) \widecheck{y}^*_{Sn+s-p(s)} + \widecheck{\epsilon}^{*}_{Sn+s}$, where $\widecheck{\epsilon}^{*}_{1}, \dots, \widecheck{\epsilon}^{*}_{SN}$ is an analogously drawn version of $\epsilon^{*}_{1}, \dots, \epsilon^{*}_{SN}$ as described in step (2) of the residual-based seasonal block bootstrap in Section \ref*{Subsubsection_Bootstrapscheme_weak}, but from $\epsilon_{1}, \dots, \epsilon_{SN}$ instead of the estimated counterparts. Further, we define $\widecheck
{X}^*_{n}(s) = (1,\widecheck{y}_{Sn+s-1}^{*\prime}, \dots, \widecheck{y}_{Sn+s-p(s)}^{*\prime})^{\prime}$, $s = 1,\dots, S$, $n = 0, \dots, N-1$, and $\widecheck{X}^*$ as the corresponding counterpart of $X$ introduced in Section \ref*{Section_Estimation} and $\widecheck{e}^* = vec \{ \widecheck{\epsilon}^{*}_{1}, \dots, \widecheck{\epsilon}^{*}_{SN}\}$. Finally, we get the residuals by $\widehat{\widecheck{\epsilon}}^{*}_{Sn+s} = \widecheck{y}^*_{Sn+s} -  \widecheck{\nu}^*(s) +  \widecheck{A}^*_{1}(s) \widecheck{y}^*_{Sn+s-1} + \dots +  \widecheck{A}^*_{p(s)}(s) \widecheck{y}^*_{Sn+s-p(s)}$, where $\widecheck{\nu}^*(s),\widecheck{A}^*_{1}(s),\ldots, \widecheck{A}^*_{p(s)}(s)$, $s = 1, \dots, S $ are the bootstrap estimators based on the sample $\widecheck{y}^*_{s^{\prime}},\dots,\widecheck{y}^*_{SN}$.

By Polya's Theorem and Lemma \ref{Equivalence_BE},
it suffices to show that $\sqrt{T} \bigl( ( \widecheck{\beta}^* - \ \widecheck
{\beta})^{\prime}, ({\widecheck{\sigma}}^* - \widecheck{\sigma})^{\prime}\big)^{\prime} \overset{d}{\to}  \mathcal{N}\bigl(\boldsymbol{0},   V  \bigr),$ in probability, where $V$ is the limiting variance matrix derived in Theorem \ref*{theo_joint_clt}. If we transfer the proof of Theorem \ref*{theo_joint_clt} to the bootstrap world, this results in

\begin{align}\label{Deviationform_Boot}
\sqrt{N} \begin{pmatrix}
\widecheck{\beta}^* - \widecheck{\beta}  \\  \widecheck{\sigma}^* - \widecheck{\sigma} \end{pmatrix} &=   \begin{pmatrix} \widecheck{Q}_N^{\beta*} \widecheck{R}^{\beta} \frac{1}{\sqrt{N}} \sum\limits_{n=0}^{N-1} \begin{pmatrix} \widecheck{\epsilon}^*_{Sn+1} \\ vec\{\widecheck
{\epsilon}^*_{Sn+1} [\widecheck{\epsilon}_{Sn+1-k}^{*\prime} + \nu(1-k)^{\prime}] \} \\ \widecheck{\epsilon}_{Sn+2}^* \\ vec\{\widecheck{\epsilon}^*_{Sn+2}[ \widecheck{\epsilon}_{Sn+2-k}^{*\prime} +   \nu(2-k)^{\prime}] \} \\ \vdots \\ \widecheck{\epsilon}^*_{Sn+S}  \\ vec\{\widecheck{\epsilon}^*_{Sn+S} [ \widecheck{\epsilon}_{Sn+S-k}^{*\prime} +   \nu(S-k)^{\prime} ]   \} \end{pmatrix}  \\  \frac{1}{\sqrt{N}}\sum\limits_{n=0}^{N-1} L_{mS} \begin{pmatrix} vec\{ \widecheck
{\epsilon}^*_{Sn+1}  \widecheck{\epsilon}_{Sn+1}^{*\prime} \} - E^*[vec\{ \widecheck{\epsilon}^*_{Sn+1}  \widecheck{\epsilon}_{Sn+1}^{*\prime} \}]\\ \vdots \\ vec\{ \widecheck{\epsilon}^*_{Sn+S}  \widecheck
{\epsilon}_{Sn+S}^{*\prime} \} - E^*[vec\{ \widecheck{\epsilon}^*_{Sn+S}  \widecheck{\epsilon}_{Sn+S}^{*\prime} \}] \end{pmatrix}  \end{pmatrix}\\ \nonumber \\ &= A^*_q + (A^*-A^*_q) \nonumber  ,
\end{align}
with
\begin{align*}
    \widecheck{Q}_N^{\beta*} &= R \left[R^{\prime} \left\{\frac{1}{N} \widecheck{X}^* \widecheck{X}^{*\prime} \otimes \boldsymbol{I}_m \right\} R \right]^{-1} R^{\prime}, \\
    \widecheck{R}^{\beta} &= \begin{pmatrix}
\boldsymbol{I}_m & \boldsymbol{0}_{m\times m^2} & \dots & \dots &  \boldsymbol{0} \\ \boldsymbol{0}_{m^2p(1) \times m} & \sum_{k=1}^{SN-1}  (C_k(1) \otimes \boldsymbol{I}_m)  & \ddots  & & \vdots \\ \vdots & \ddots & \ddots & \ddots & \vdots \\  \vdots &  & \ddots & \boldsymbol{I}_m &  \boldsymbol{0}_{m \times m^2} \\ \boldsymbol{0} & \dots & \dots & \boldsymbol{0}_{m^2p(S) \times m} & \sum_{k=1}^{SN-1}  (C_k(S)  \otimes \boldsymbol{I}_m)
\end{pmatrix}, 
\end{align*}
where $A^*$ denotes the right-hand side in \eqref{Deviationform_Boot} and $A^*_q$ is defined as the same, but with  $\sum_{k=1}^{SN-1}$ replaced by $\sum_{k=1}^{q}$ for some $q \in \mathbb{N}$. Note again that, in this short-hand notation, $\widecheck{R}^{\beta}$ contains sums over $k$ that are also applied to the factor $\frac{1}{\sqrt{N}}\sum_{n=0}^{N-1}(\cdots)$. To prove the bootstrap central limit theorem of the joint deviation form, we make use of Proposition 6.3.9 of Brockwell \& Davies (1991) 
and it suffices to show
 
 \begin{itemize}
 \item[(a)] $A^*_q \overset{d}{\to}  \mathcal{N}(0, V_q)$ in probability as $N \to \infty$
 \item[(b)] $V_q \to V$  as $q \to \infty$
 \item[(c)] $\forall \delta > 0: \; \; \underset{q \to \infty}{\lim} \underset{N \to \infty}{\limsup}\ P^*( |A^* - A_q^*|_1 > \delta) = 0  $ in probability.
 \end{itemize}
In order to prove (a), with $\widetilde{m} = \frac{m(m+1)}{2}$, we bring $A^*_q$ to the following form
\begin{align*}
A^*_q = \widecheck{Q}^*_N  R_q \frac{1}{\sqrt{N}} \sum_{n=0}^{N-1} 
\widecheck{W}_{n,q}^{*},
\end{align*}
where $R_q$ as defined in the proof of Theorem \ref*{theo_joint_clt} and with
\begin{align*}
\widecheck{Q}^*_N = \begin{pmatrix}
R \left[R^{\prime} \left\{\frac{1}{N} \widecheck{X}^* \widecheck{X}^{*\prime} \otimes \boldsymbol{I}_m \right\} R \right]^{-1} R^{\prime} \; \;  \;  & \boldsymbol{0}_{m \sum\limits_{s=1}^S (mp(s)+1) \times S \widetilde{m}} \\ \boldsymbol{0}_{ S\widetilde{m} \times m \sum\limits_{s=1}^S (mp(s)+1)} & \boldsymbol{I}_{S\tilde{m}} 
\end{pmatrix},
\end{align*}
and 
\begin{align*}
  \widecheck{W}_{n,q}^{*} = \begin{pmatrix}
  \begin{pmatrix} \widecheck{\epsilon}^*_{Sn+1} \\ vec\{\widecheck{\epsilon}^*_{Sn+1} [\widecheck{\epsilon}_{Sn}^{*\prime} + \nu(0)^{\prime}] \}
 \\ \vdots \\ vec\{\widecheck{\epsilon}^*_{Sn+1} [\widecheck{\epsilon}_{Sn+1-q}^{*\prime} + \nu(1-q)^{\prime}]  \}  \\ 
 \widecheck{\epsilon}^*_{Sn+2} \\ vec\{\widecheck{\epsilon}^*_{Sn+2} [\widecheck{\epsilon}_{Sn+1}^{*\prime} + \nu(1)^{\prime}]  \}
 \\ \vdots \\ vec\{\widecheck{\epsilon}^*_{Sn+2}  [\widecheck{\epsilon}_{Sn+2-q}^{*\prime} + \nu(2-q)^{\prime}] \} \\ \vdots \\  \widecheck{\epsilon}^*_{Sn+S} \\ vec\{\widecheck{\epsilon}^*_{Sn+S} [\widecheck{\epsilon}_{Sn+S-1}^{*\prime} + \nu(S-1)^{\prime}] \}
 \\ \vdots \\ vec\{\widecheck{\epsilon}^*_{Sn+S} [\widecheck{\epsilon}_{Sn+S-q}^{*\prime} + \nu(S-q)^{\prime}] \}
\end{pmatrix}  \\  L_{mS} \begin{pmatrix} vec\{ \widecheck{\epsilon}^*_{Sn+1}  \widecheck{\epsilon}_{Sn+1}^{*\prime} \} - E^*[ vec\{ \widecheck{\epsilon}^*_{Sn+1}  
\widecheck{\epsilon}_{Sn+1}^{*\prime} \}] \\ \vdots \\ vec\{ \widecheck{\epsilon}^*_{Sn+S}  \widecheck{\epsilon}_{Sn+S}^{*\prime} \} -  E^*[vec\{ \widecheck{\epsilon}^*_{Sn+S}  \widecheck{\epsilon}_{Sn+S}^{*\prime} \}] \end{pmatrix} \end{pmatrix},
\end{align*}
where $\widecheck{Q}_N^*$ and $\widecheck{W}_{n,q}^{*}$ are of dimension $(m\sum_{s=1}^S (mp(s) +1) + S\widetilde{m}) \times (m\sum_{s=1}^S (mp(s) +1) + S\widetilde{m})$ and $S(m+qm^2)+S\widetilde{m}$. Using Lemma \ref{CLT_Boot_Innovations}, we have 
\[ A_q^* \overset{d}{\to}  \mathcal{N}(0, V_q) \; \text{ as } \; N \to \infty,\]
with
\[ V_q = QR_q \Omega_q R_q^{\prime} Q^{\prime}, \]
where $Q$ and $\Omega_q$ are defined as in the proof of Theorem \ref*{theo_joint_clt}. Note that it is possible to show that $\widecheck{Q}_N^* \to Q$ with respect to $P^*$ by using similar arguments as in the proof of Lemma A.2 in Brüggemann et al.~(2016).  
Part b) of the proof follows directly by applying the cumulant summability assumption imposed in Assumption \ref*{Mixing_bootstrap_con} and the exponential decay of $\{ C_k(s)\}_{k \in \mathbb{N}}$ for all $s = 1, \dots, S$. Part c) of the proof follows by using similar arguments as in Theorem 4.1 in Brüggemann et al.~(2016).
\hfill  $\square$

\begin{lemma}[Equivalence of Bootstrap Estimators]\label{Equivalence_BE}
Under the assumptions of Theorem \ref*{Mixing_bootstrap_con}, we have
\begin{align*}
    \sqrt{N}\bigl((\widehat{\beta}_{res}^* - \widehat{\beta}_{res}) -  (\widecheck
    {\beta}^* - \widecheck{\beta})\bigr) &= o_{P^*}(1), \\
    \sqrt{N}\bigl((\widehat{\sigma}^* - \widehat{\sigma}) -  (\widecheck{\sigma}^* - \widecheck{\sigma})\bigr) &= o_{P^*}(1).
\end{align*}
\end{lemma}
\begin{proof}
To prove the statement, we use similar arguments as in the proof of Lemma A.1 in Brüggemann et al.~(2016) 
and Lemma B.2 in Jentsch \& Lunsford (2022). 
However, note that
we have to allow for periodically correlated processes and replace the moving block bootstrap method by a seasonal variant of it. 
\end{proof}

\begin{lemma}[CLT for Bootstrap Innovations]\label{CLT_Boot_Innovations}
Let $q \in \mathbb{N}$ and define $\widecheck{W}_{n,q}^* = (\widecheck{W}^{(1)*\prime}_{n,q},  \widecheck{W}^{(2)*\prime}_{n})^{\prime}$ as
\begin{align*}
\widecheck{W}^{(1)*}_{n,q}  = \begin{pmatrix}
\widecheck{W}^{(1,1)*}_{n,q} \\ \vdots \\ \widecheck{W}^{(1,S)*}_{n,q}
\end{pmatrix}, 
\quad
\widecheck{W}^{(2)*\prime}_{n}  = 
 L_{mS} \begin{pmatrix} vec\{ \widecheck{\epsilon}^*_{Sn+1}  \widecheck{\epsilon}_{Sn+1}^{*\prime} \} - E^*[vec\{ \widecheck{\epsilon}_{Sn+1}^* \widecheck{\epsilon}_{Sn+1}^{*\prime}  \} ] \\ \vdots \\ vec\{ \widecheck{\epsilon}^*_{Sn+S}  \widecheck{\epsilon}_{Sn+S}^{*\prime} \} - E^*[vec\{ \widecheck{\epsilon}_{Sn+S}^* \widecheck{\epsilon}_{Sn+S}^{*\prime}  \} ]  \end{pmatrix},
\end{align*}
where, for $s=1,\dots, S$,
\begin{align*}
\widecheck{W}^{(1,s)*}_{n,q} = \begin{pmatrix}
 \widecheck{\epsilon}_{Sn+s}^* \\ vec\{\widecheck{\epsilon}_{Sn+s}^* [\widecheck{\epsilon}_{Sn+s-1}^{*\prime} + \nu(s-1)^{\prime}] \}
 \\ \vdots \\ vec\{\widecheck{\epsilon}_{Sn+s}^* [\widecheck{\epsilon}_{Sn+s-q}^{*\prime} + \nu(s-q)^{\prime}]  \} \end{pmatrix}.
\end{align*}
Then, under the assumptions of Theorem \ref*{Mixing_bootstrap_con}, we have
\[ \frac{1}{\sqrt{N}} \sum_{n=0}^{N-1} \widecheck{W}^*_{n,q} \overset{d}{\to}  \mathcal{N}\bigl(\boldsymbol{0},   \boldsymbol{\Omega}_q  \bigr) \] in probability, where $\boldsymbol{\Omega}_q$ as defined in Lemma \ref{CLT_Innovation}.
\end{lemma}

\begin{proof}
To prove the statement, we use similar arguments as in the proof of Lemma A.3 in Brüggemann et al.~(2016) 
and Lemma B.3 in Jentsch \& Lunsford (2022). 
In the proof, we make use of the assumptions made in Theorem \ref*{Mixing_bootstrap_con} to derive asymptotic normality of the properly centered quantity $\frac{1}{\sqrt{N}} \sum_{n=0}^{N-1} \widecheck{W}^*_{n,q} - E^*(\widecheck
{W}^*_{n,q})$. Note that $E^*(\widecheck{W}^{(1)*}_{n,q})$ does not necessarily have to be exactly zero. Finally, using $\sqrt{N}E^*(\widecheck{W}^*_{n,q}) = o_P(1)$, we obtain the claimed result.
\end{proof}

\section{Seasonal Test for Impulse Responses}
If the goal is to test whether the seasonal impulse responses differ significantly from their \emph{averaged} seasonal impulse responses across the seasons, we can formulate the null hypothesis as
\begin{align}\label{H_0_2}
\widetilde K_0: \Theta_{ij}^{SIR}(s)  = \widebar{\Theta}^{SIR}_{ij}, \;  \; \text{ for all } s = 1,\dots, S,
\end{align}
where $\widebar{\Theta}^{SIR}_{ij}$ is the impulse response function from the SVAR($p$) with seasonal means. Then, we define the test statistic as $\widetilde{Q}_N(i,j) = \sum_{s=1}^S | \widehat{\Theta}^{SIR}_{ij}(s) - \widehat{\widebar{\Theta}}^{SIR}_{ij}|_2^2$ and perform the following steps to test for $\widetilde K_0$. 
\begin{itemize} 
\item[$(1^{\prime})$] Fit the SPVAR($p$) to the seasonal unadjusted data $y^{SU*}$ and calculate the SPVAR estimates $\widehat{\nu}(s),\widehat{A}_1(s), \dots, \widehat{A}_p(s), \widehat{H}_0(s)$ for $s=1,\dots, S$ and the structural residuals $\widehat{w}_1, \dots, \widehat{w}_{SN}$. Further, calculate the pooled PVAR estimates $\widebar{\widehat{A}}_1, \dots, \widebar{\widehat{A}}_p$ and the pooled periodic impact matrix $\widebar{\widehat{H}}_0$.
\item[$(2^{\prime})$] Perform a seasonal block bootstrap to obtain bootstrap pseudo observations $\widehat{w}_1^*, \dots, \widehat{w}_{SN}^*$ and use $(\nu(1),\dots, \nu(S),\widebar{\widehat{A}}_1, \dots, \widebar{\widehat{A}}_p, \widebar{\widehat{H}}_0)$ to get a bootstrap sample $y^{*} = (y_1^{*},\dots,y_{SN}^*$) of seasonally unadjusted data.
\item[$(3^{\prime})$] Fit an SVAR($p$) with seasonal intercepts to $y^{*}$ and an SPVAR($p$) to $y^{*}$ and calculate the bootstrap versions $\widehat{\widebar{\Theta}}^{SIR*}_{ij}$ and $\widehat{\Theta}^{SIR*}_{ij}(s), s = 1,\dots,S$. Calculate $\widetilde{Q}_N^*(i,j)$ by inserting the bootstrap versions $\widehat{\widebar{\Theta}}^{SIR*}_{ij}$ and $\widehat{\Theta}^{SIR*}_{ij}(s)$. 
\item[$(4^{\prime})$] Repeat $(2^{\prime})- (3^{\prime})$ $L$ times, where $L$ is large, and reject $\widetilde K_0$ in \eqref{H_0_2}, if $\widetilde{Q}_N(i,j) > \widetilde{q}^*_{1-\alpha}$, where $\widetilde{q}^*_{ 1-\alpha}$ the $1-\alpha$ quantile of $\widetilde{Q}_N^{*(1)}(i,j), \dots,\widetilde{Q}_N^{*(L)}(i,j)$.
\end{itemize}

\section{Real Data Application with Approximate In-Average Identification}\label{Real_Data_InAverage}

In this section, seasonal structural impulse responses using the restricted SPVAR(12) on seasonally unadjusted data are analyzed, in which the aggregate demand $(ad)$, aggregate supply $(as)$ and monetary policy shocks $(mp)$ are identified. To identify the structural $(ad)$, $(as)$ and $(mp)$ shocks, we use approximate in-average identification and apply a mixture of short- and long-run restrictions as described in the last part of Section \ref{subsection_Identification}. In this Section, we compare SPVAR impulse responses identified by approximate in-average identification with the counterparts identified by the full identification approach. 

\begin{figure}[t]
\centering
\includegraphics[scale=0.8]{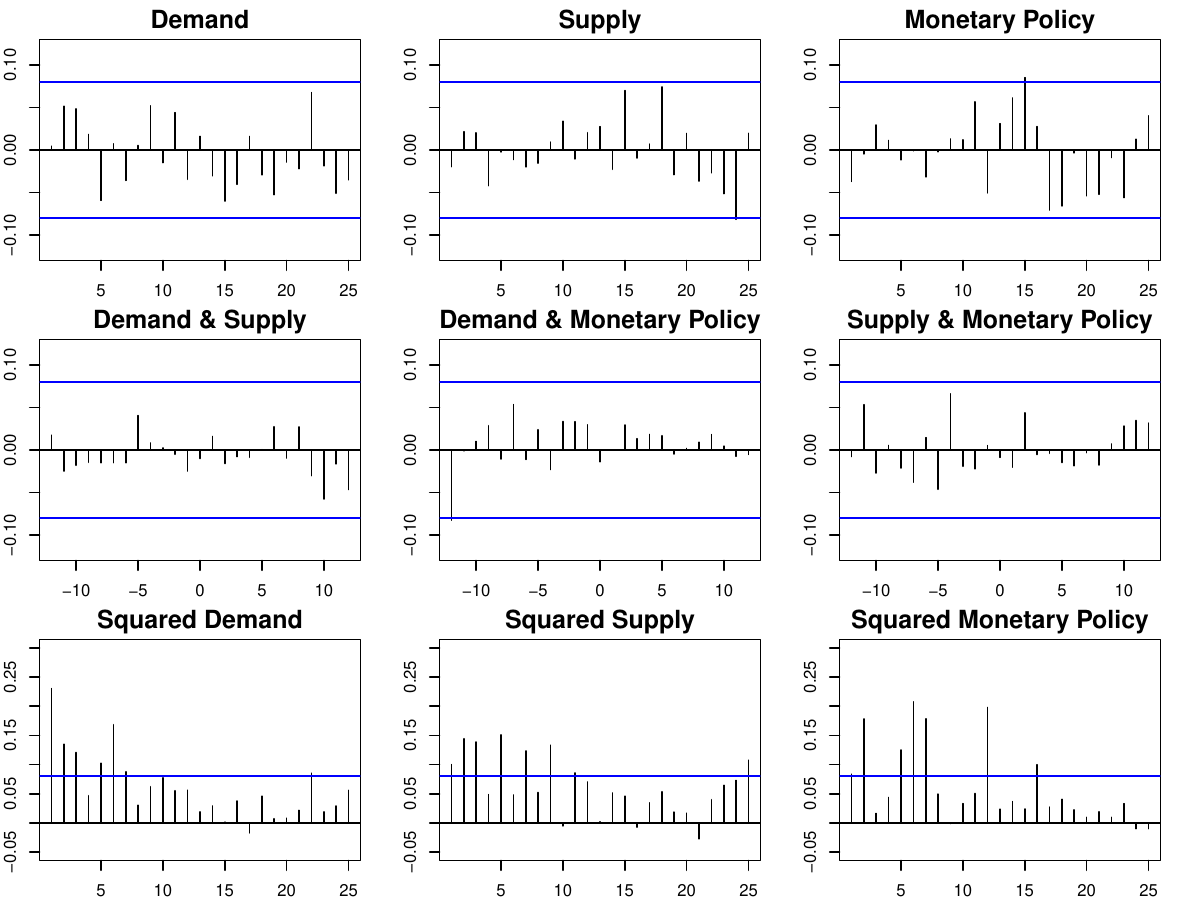}
\caption{ACFs (top panels) and CCFs (central panels) of structural PVAR shocks  $\widehat{w}_1, \dots, \widehat{w}_{SN}$ and ACFs of squared structural PVAR shocks $\widehat{w}^2_1, \dots, \widehat{w}^2_{SN}$ (bottom panels). Structural shocks are identified based on approximate in-average identification.}
\label{WWN_InAverage}
\end{figure}

Note that approximate in-average identification does not necessarily ensure that the structural shocks are fully orthogonal to each other. In the panels of the center row of Figure \ref{WWN_InAverage}, however, it can be clearly seen that all cross-correlations (CCFs) at lag 0 appear to be negligibly small. Accordingly, the non-orthogonality of the structural shocks has only a very minor effect on the point estimates of the impulse responses and the structural shocks can therefore be treated as (approximately) orthogonal to each other. In general, the non-orthogonality of the structural shocks should not be too pronounced, since in addition to the point estimates, the confidence bands of the impulse responses in particular can also be influenced. In our application, a subsequent Cholesky identification approach shows that the effects of the non-orthogonality of the structural shocks on the point estimate and confidence intervals of the impulse responses are negligibly small.

In Figures \ref{Comp_IR1} - \ref{Comp_IR9}, we compare the SPVAR impulse responses identified by the full identification approach (in black) with those identified by approximate in-average identification (in red). The dashed lines give (pointwise) 68\% confidence intervals of the structural impulse responses of the SPVARs. For the construction of confidence intervals of the seasonal impulse responses, the non-seasonal variant of the residual-based seasonal block bootstrap described in Section \ref{Subsubsection_Bootstrapscheme_weak} is used, because the bottom rows of Figures \ref{WWN1} and \ref{WWN_InAverage} show significant autoregressive structure in the squared structural shocks such that the strong periodic white noise assumption appears to not hold in both SPVARs. The block length of the residual-based seasonal block bootstrap is set to $b=5$, which is in line with the choice in \cite{jentsch2022asymptotically}, who considered a similar data generating process, and with the results in \cite{bertail2024optimal}.
We again use standard percentile intervals for the construction of the confidence intervals and $L = 1000$ bootstrap repetitions for each model.

Figures \ref{Comp_IR1}, \ref{Comp_IR2}, and \ref{Comp_IR3} present the SPVAR impulse responses generated by a monetary policy shock. In Figure \ref{Comp_IR1}, we clearly see the effect of long-run restrictions in the discussed identification schemes. While the impulse response is restricted to 0 in the long-run for all seasons in the case of full identification, it will be 0 only on average in the case of approximate in-average identification. However, this periodic flexibility of approximate in-average identification does not result in significant periodic impulse responses triggered by monetary policy shocks. In total, it appears that monetary policy shocks do not have significant periodic effects on the macro variables.

Figures \ref{Comp_IR4} - \ref{Comp_IR9} show some clear differences in the SPVAR impulse responses. These clear differences arise, on the one hand, from the fact that full identification takes into account periodic contemporaneous effects, while also imposing significantly more restrictions with regard to seasonality, compared to approximate in-average identification. 

 \begin{figure}[h!]
\centering
\includegraphics[scale=0.75]{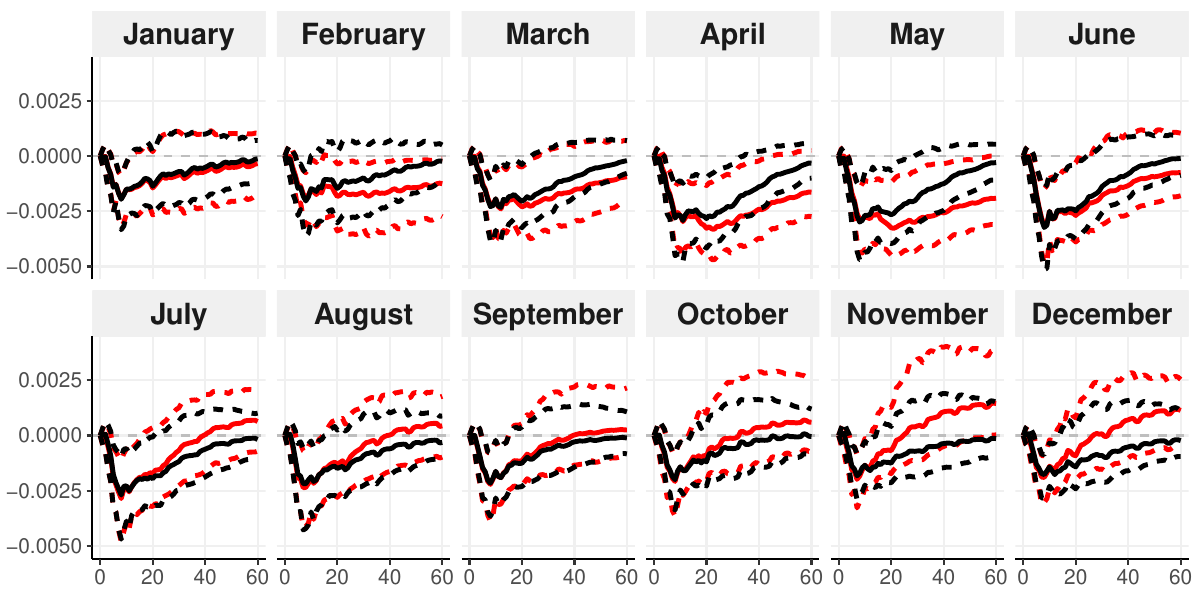}
\caption{Seasonal structural impulse responses of IP after a monetary policy shock. SPVAR impulse responses in red are identified by approximate in-average identification, while SPVAR impulse responses in black are identified by the full identification approach. The corresponding month indicates the time of occurrence of the shock.}
\label{Comp_IR1}
\end{figure}
\newpage
\begin{figure}[h!]
\centering
\includegraphics[scale=0.75]{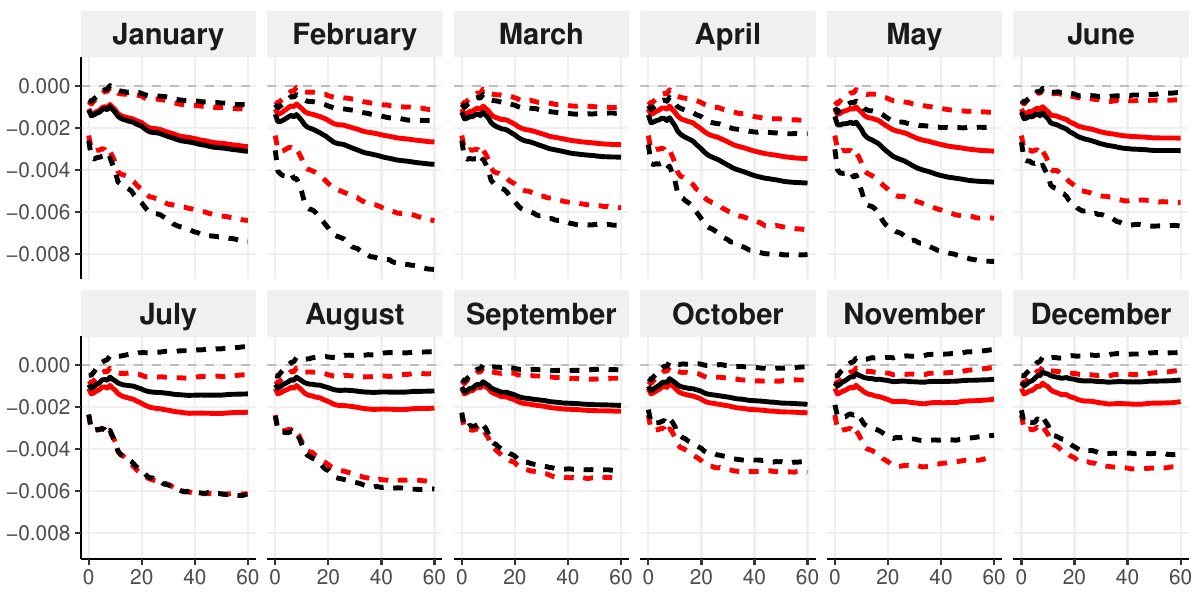}
\caption{Seasonal structural impulse responses of INF after a monetary policy shock. SPVAR impulse responses in red are identified by approximate in-average identification, while SPVAR impulse responses in black are identified by the full identification approach. The corresponding month indicates the time of occurrence of the shock.}
\label{Comp_IR2}
\end{figure}

\begin{figure}[h!]
\centering
\includegraphics[scale=0.75]{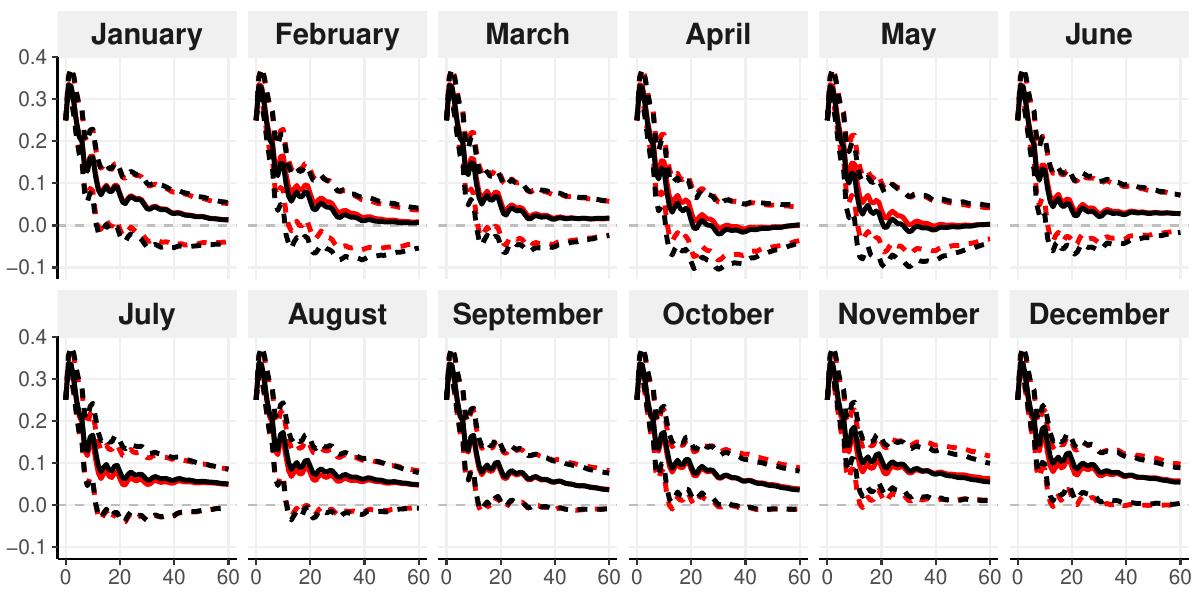}
\caption{Seasonal structural impulse responses of FFR after a monetary policy shock. SPVAR impulse responses in red are identified by approximate in-average identification, while SPVAR impulse responses in black are identified by the full identification approach. The corresponding month indicates the time of occurrence of the shock.}
\label{Comp_IR3}
\end{figure}

\newpage
\begin{figure}[h!]
\centering
\includegraphics[scale=0.75]{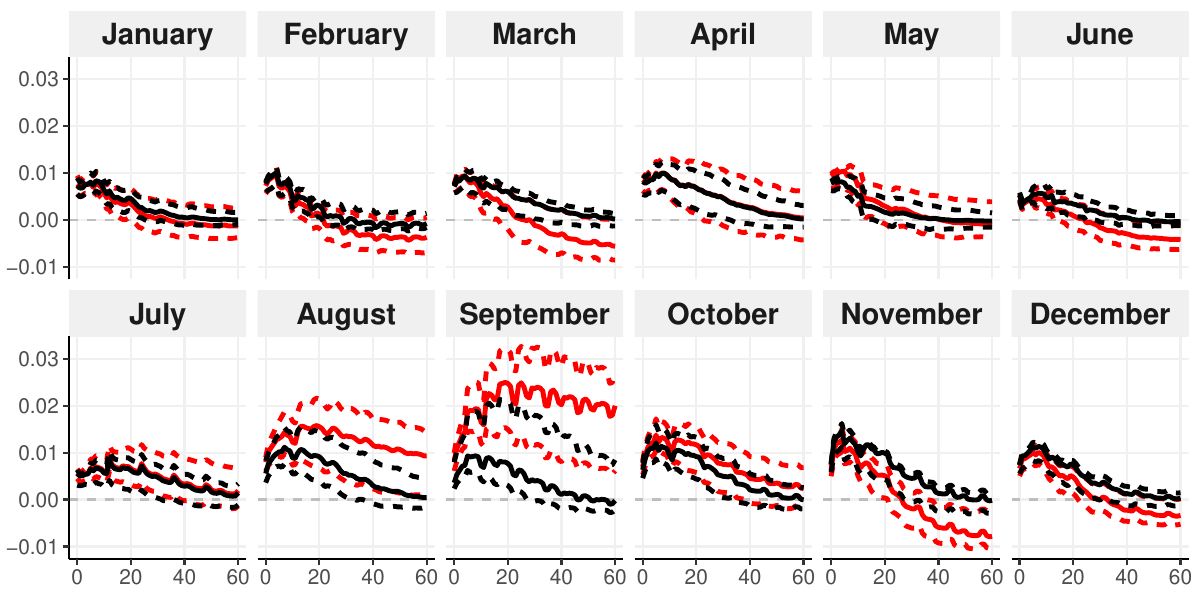}
\caption{Seasonal structural impulse responses of IP after a positive demand shock. SPVAR impulse responses in red are identified by approximate in-average identification, while SPVAR impulse responses in black are identified by the full identification approach. The corresponding month indicates the time of occurrence of the shock.}
\label{Comp_IR4}
\end{figure}

\begin{figure}[h!]
\centering
\includegraphics[scale=0.75]{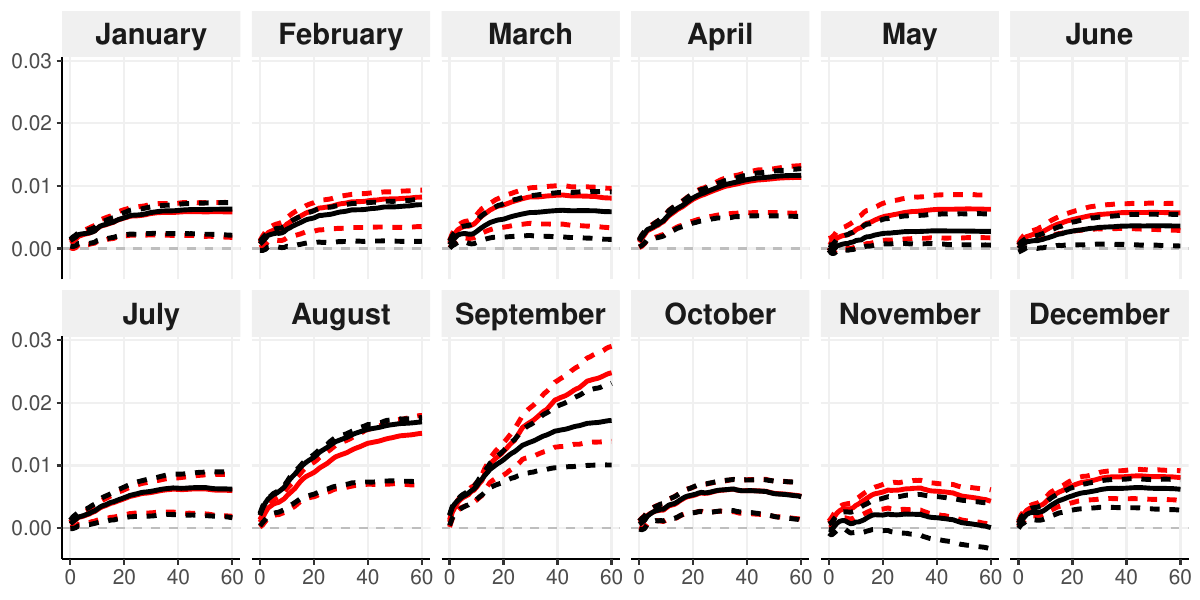}
\caption{Seasonal structural impulse responses of INF after a positive demand shock. SPVAR impulse responses in red are identified by approximate in-average identification, while SPVAR impulse responses in black are identified by the full identification approach. The corresponding month indicates the time of occurrence of the shock.}
\label{Comp_IR5}
\end{figure}
\newpage
\begin{figure}[h!]
\centering
\includegraphics[scale=0.75]{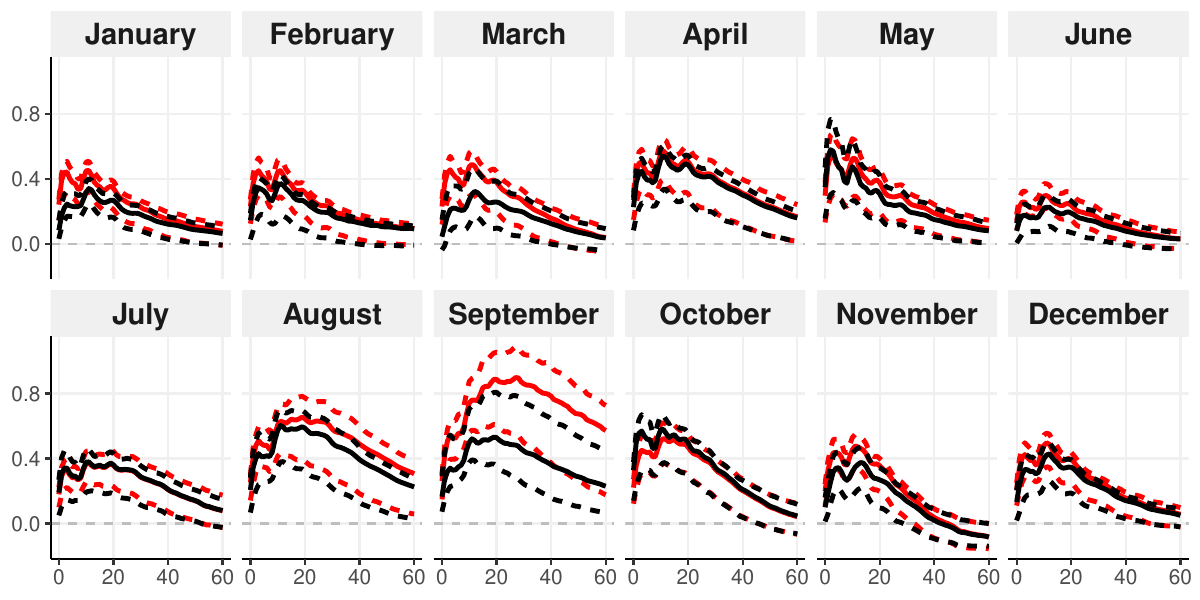}
\caption{Seasonal structural impulse responses of FFR after a positive demand shock. SPVAR impulse responses in red are identified by approximate in-average identification, while SPVAR impulse responses in black are identified by the full identification approach. The corresponding month indicates the time of occurrence of the shock.}
\label{Comp_IR6}
\end{figure}

\begin{figure}[h!]
\centering
\includegraphics[scale=0.75]{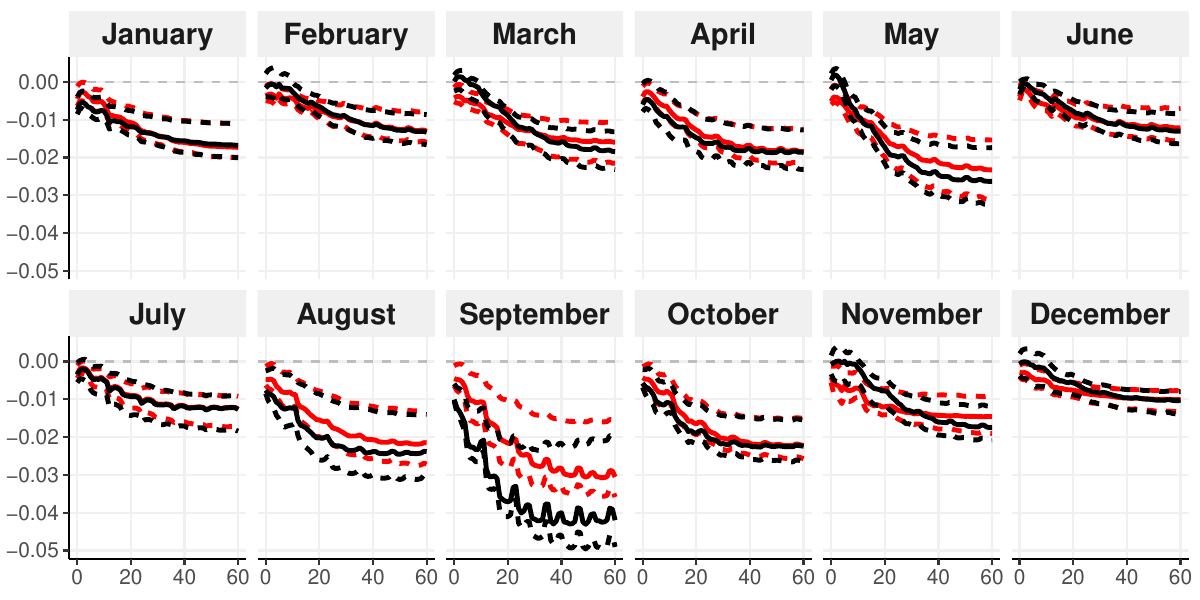}
\caption{Seasonal structural impulse responses of IP after a negative supply shock. SPVAR impulse responses in red are identified by approximate in-average identification, while SPVAR impulse responses in black are identified by the full identification approach. The corresponding month indicates the time of occurrence of the shock.}
\label{Comp_IR7}
\end{figure}
\newpage
\begin{figure}[h!]
\centering
\includegraphics[scale=0.75]{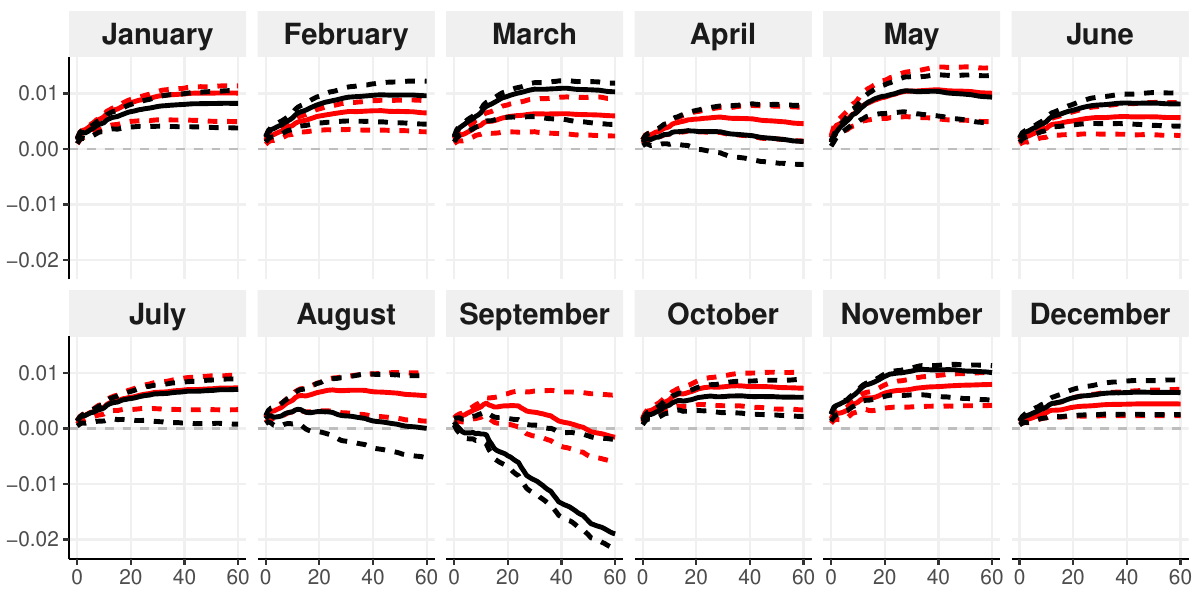}
\caption{Seasonal structural impulse responses of INF after a negative supply shock. SPVAR impulse responses in red are identified by approximate in-average identification, while SPVAR impulse responses in black are identified by the full identification approach. The corresponding month indicates the time of occurrence of the shock.}
\label{Comp_IR8}
\end{figure}

\begin{figure}[h!]
\centering
\includegraphics[scale=0.75]{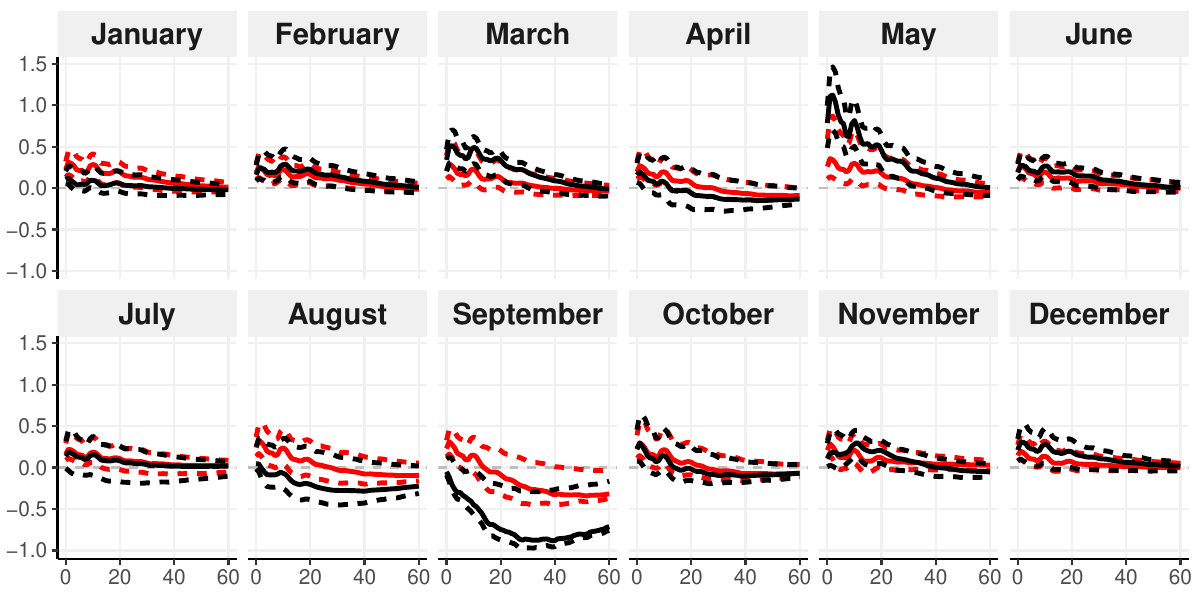}
\caption{Seasonal structural impulse responses of FFR after a negative supply shock. SPVAR impulse responses in red are identified by approximate in-average identification, while SPVAR impulse responses in black are identified by the full identification approach. The corresponding month indicates the time of occurrence of the shock.}
\label{Comp_IR9}
\end{figure}

\newpage

\section{Structural Impulse Responses Plots}\label{SIR_Plots}

In this section, we present the seasonal structural impulse responses generated by positive demand and negative supply shocks to IP, INF and FFR. The SPVAR impulse responses are identified by the full identification approach and are compared to SVAR impulse responses on seasonally adjusted data which are identified by imposing the same mixture of short- and long-run restrictions.

 \begin{figure}[h!]
\centering
\includegraphics[scale=0.75]{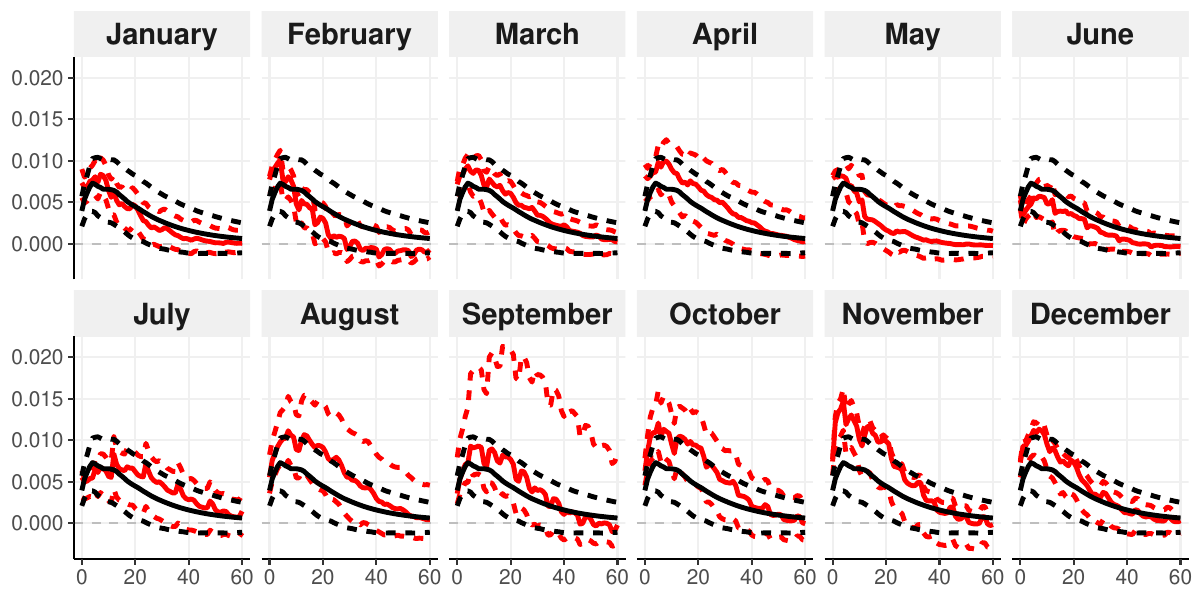}
\caption{Seasonal structural impulse responses of IP after a positive demand shock. SPVAR impulse responses in red are identified by the full identification approach, while SVAR impulse responses are in black and are constant across the seasons. The corresponding month indicates the time of occurrence of the shock.}
\label{NS_IR4}
\end{figure}
\newpage
 \begin{figure}[h!]
\centering
\includegraphics[scale=0.75]{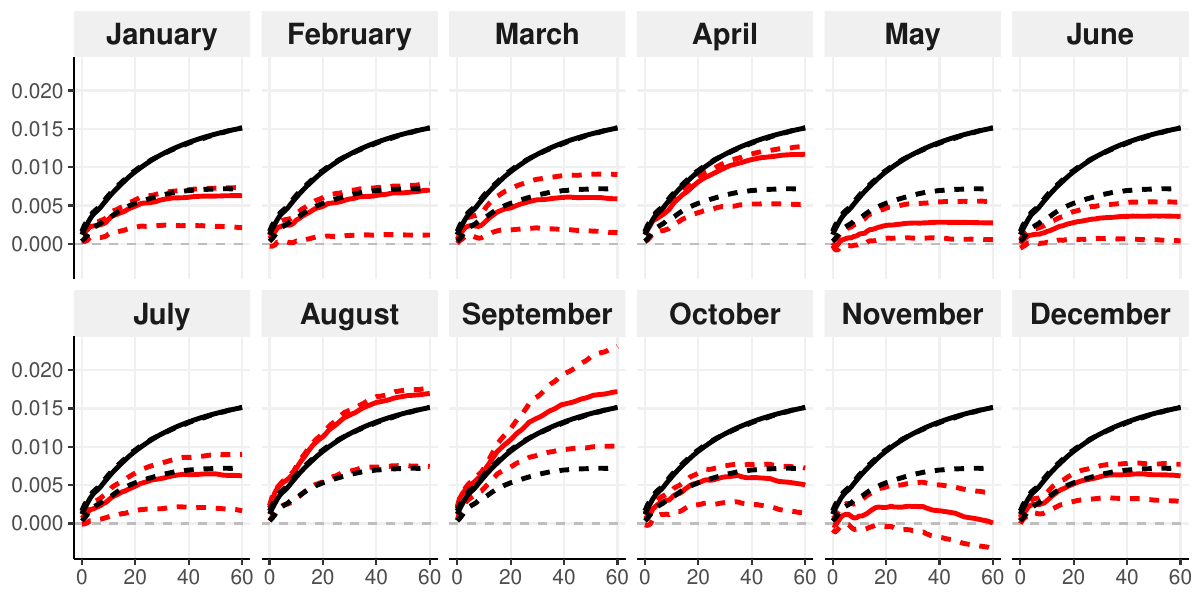}
\caption{Seasonal structural impulse responses of INF after a positive demand shock. SPVAR impulse responses in red are identified by the full identification approach, while SVAR impulse responses are in black and are constant across the seasons. The corresponding month indicates the time of occurrence of the shock.}
\label{NS_IR5}
\end{figure}
\newpage
 \begin{figure}[h!]
\centering
\includegraphics[scale=0.75]{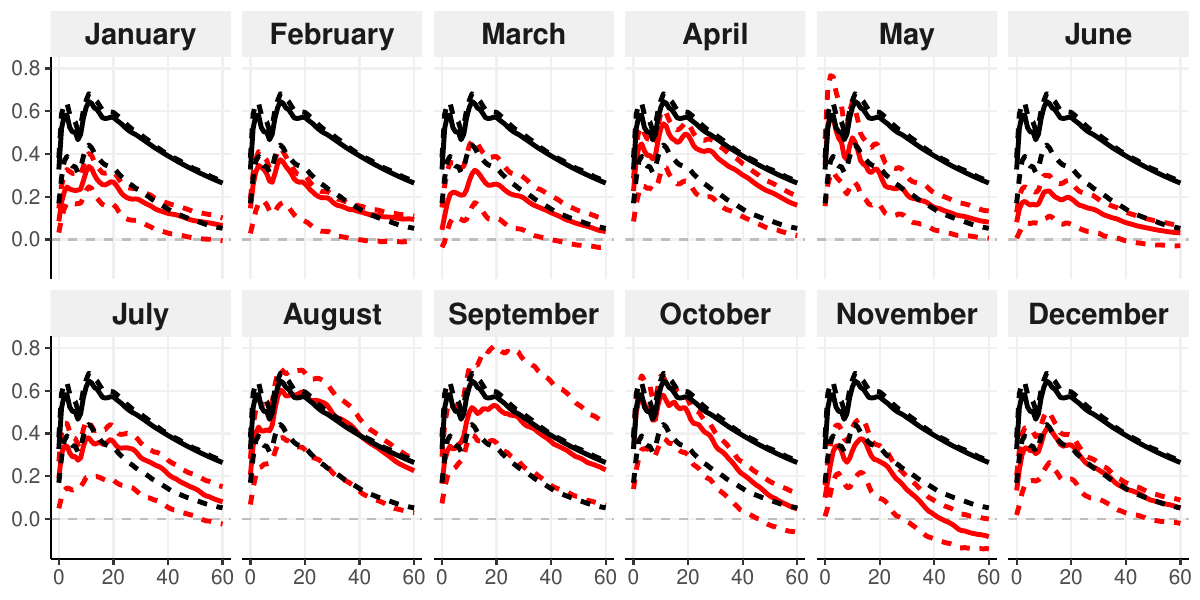}
\caption{Seasonal structural impulse responses of FFR after a positive demand shock. SPVAR impulse responses in red are identified by the full identification approach, while SVAR impulse responses are in black and are constant across the seasons. The corresponding month indicates the time of occurrence of the shock.}
\label{NS_IR6}
\end{figure}

 \begin{figure}[h!]
\centering
\includegraphics[scale=0.75]{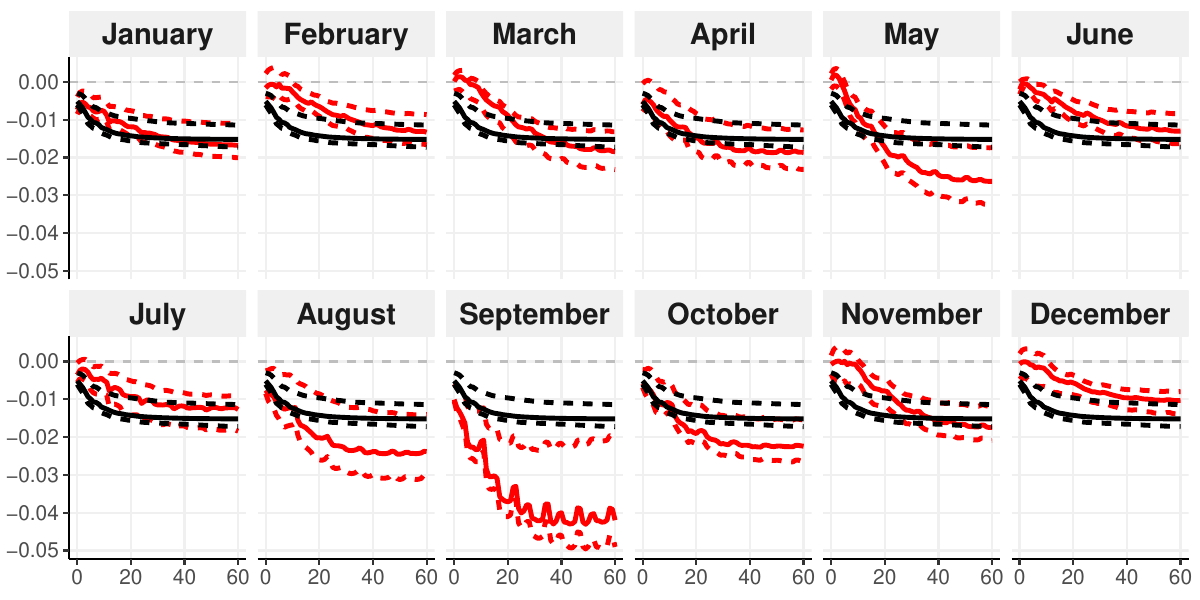}
\caption{Seasonal structural impulse responses of IP after a negative supply shock. SPVAR impulse responses in red are identified by the full identification approach, while SVAR impulse responses are in black and are constant across the seasons. The corresponding month indicates the time of occurrence of the shock.}
\label{NS_IR7}
\end{figure}
\newpage

\begin{figure}[h!]
\centering
\includegraphics[scale=0.75]{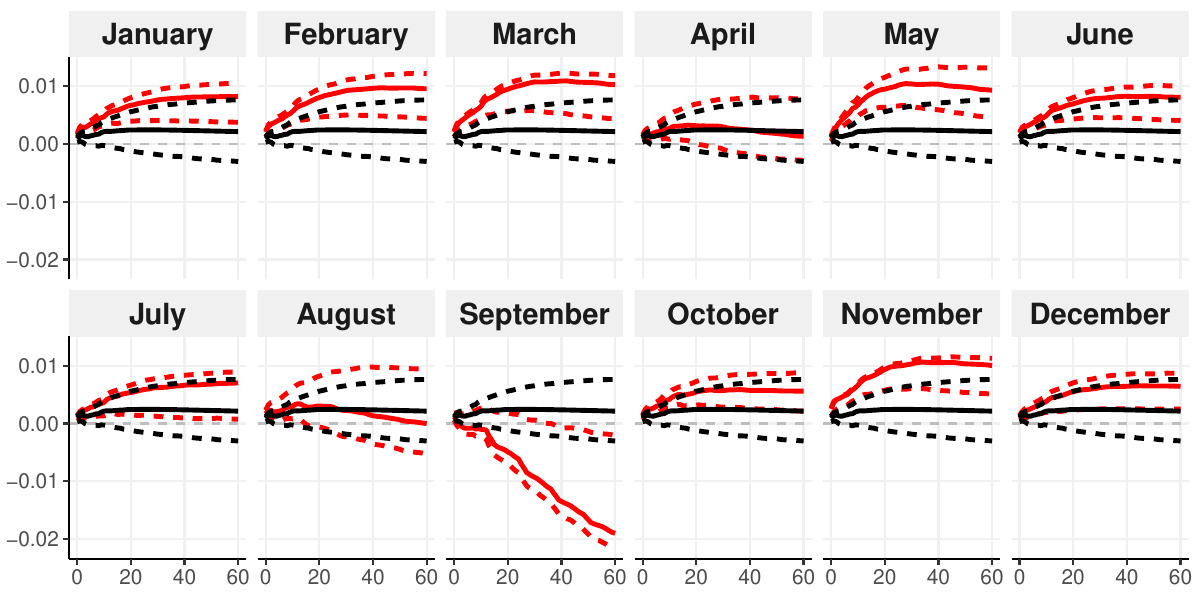}
\caption{Seasonal structural impulse responses of INF after a negative supply shock. SPVAR impulse responses in red are identified by the full identification approach, while SVAR impulse responses are in black and are constant across the seasons. The corresponding month indicates the time of occurrence of the shock.}
\label{NS_IR8}
\end{figure}

\begin{figure}[h!]
\centering
\includegraphics[scale=0.75]{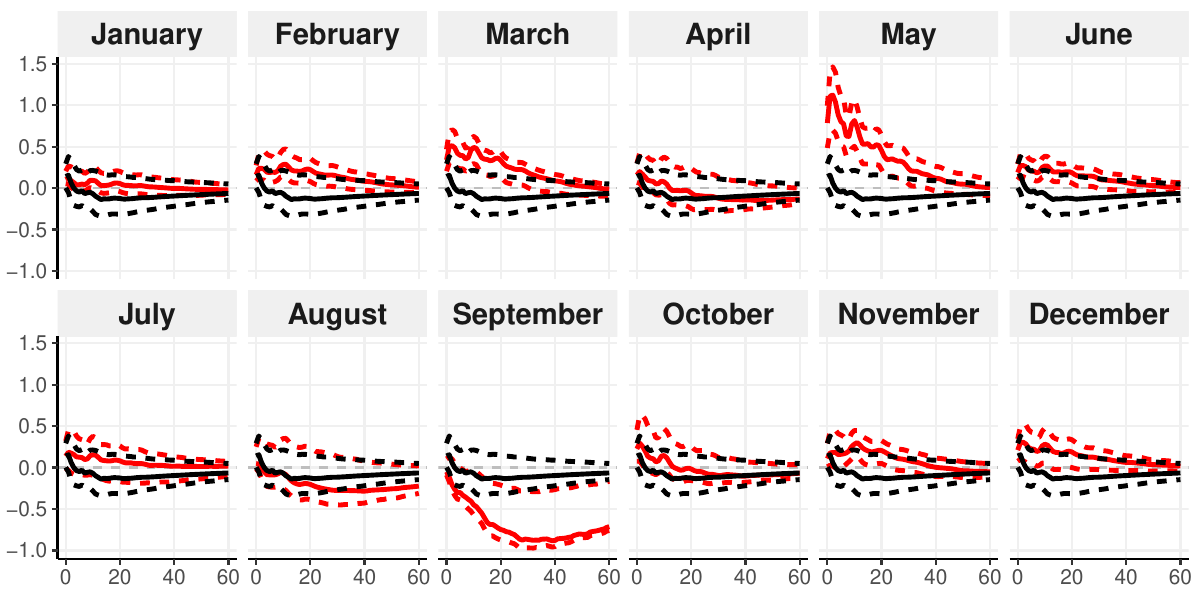}
\caption{Seasonal structural impulse responses of FFR after a negative supply shock. SPVAR impulse responses in red are identified by the full identification approach, while SVAR impulse responses are in black and are constant across the seasons. The corresponding month indicates the time of occurrence of the shock.
}
\label{NS_IR9}
\end{figure}

\newpage

\section{Coverage Rates}\label{Coverage_Tables}

In this section, simulated coverage rates are presented. The coverage rates for GARCH specifications G0 and G3 for lag $k=12$ and $N=20$ are shown in Table \ref{Cov_G0G3_N100}. Table \ref{Cov_b1swn_k0} - \ref{Cov_b7para1_k0} illustrate the coverage rates for lag $k=0$, while Table \ref{Cov_b1para1_k12} - \ref{Cov_b7para1_k12} present the coverage rates for GARCH specifications G1 and G2 for lag $k=12$. Due to a simplified presentation, we only provide coverage rates for lags $k \in\{0, 12\}$. For other lags, the coverage rate is not particularly different.

\begin{table}[h!]
\centering
\scriptsize
\begin{tabular}{l|l|l|c|ccccccccc}
 & $b$ &$ N$ & $s$ & $ ad \to y^{ip}$ & $ ad \to  \pi$ &  $ ad \to  i$ & $as \to y^{ip}$ & $ as\to \pi$ & $ as \to i$ & $mp \to y^{ip}$ & $ mp \to \pi$ &  $mp \to i$ \\   \hline
& & &  1 & 0.818 & 0.763 & 0.734 & 0.714 & 0.859 & 0.963 & 0.971 & 0.987 & 0.922 \\ 
&  & &    2 & 0.829 & 0.900 & 0.710 & 0.993 & 0.782 & 0.954 & 0.984 & 0.992 & 0.941 \\ 
  &  & &  3 & 0.734 & 0.962 & 0.888 & 0.980 & 0.681 & 0.810 & 0.944 & 0.991 & 0.895 \\ 
  &   & &   4 & 0.760 & 0.556 & 0.514 & 0.851 & 0.979 & 0.968 & 0.954 & 0.993 & 0.934 \\ 
   & & &    5 & 0.970 & 0.983 & 0.725 & 0.947 & 0.834 & 0.777 & 0.947 & 0.959 & 0.934 \\ 
 G0 & 1 & 20 &   6 & 0.864 & 0.993 & 0.938 & 0.970 & 0.759 & 0.949 & 0.960 & 0.987 & 0.926 \\ 
  &  & &   7 & 0.764 & 0.907 & 0.748 & 0.944 & 0.891 & 0.979 & 0.910 & 0.961 & 0.807 \\ 
  &  & &   8 & 0.802 & 0.385 & 0.613 & 0.660 & 0.986 & 0.931 & 0.947 & 0.968 & 0.844 \\ 
  &  & &   9 & 0.976 & 0.495 & 0.851 & 0.428 & 0.949 & 0.730 & 0.935 & 0.982 & 0.852 \\ 
  &  & &   10 & 0.843 & 0.867 & 0.574 & 0.748 & 0.928 & 0.988 & 0.939 & 0.985 & 0.829 \\ 
   & & &   11 & 0.563 & 0.996 & 0.858 & 0.989 & 0.587 & 0.961 & 0.943 & 0.965 & 0.760 \\ 
  &  & &   12 & 0.604 & 0.880 & 0.657 & 0.987 & 0.916 & 0.963 & 0.949 & 0.954 & 0.815 \\  
   \hline
  &  & &  1 & 0.826 & 0.762 & 0.739 & 0.719 & 0.856 & 0.965 & 0.974 & 0.988 & 0.922 \\ 
 &  & &    2 & 0.832 & 0.904 & 0.707 & 0.992 & 0.775 & 0.956 & 0.986 & 0.991 & 0.943 \\ 
  &  & &   3 & 0.727 & 0.962 & 0.892 & 0.980 & 0.684 & 0.814 & 0.943 & 0.992 & 0.897 \\ 
 &  & &    4 & 0.761 & 0.550 & 0.514 & 0.848 & 0.978 & 0.967 & 0.957 & 0.992 & 0.931 \\ 
   & & &   5 & 0.972 & 0.985 & 0.728 & 0.947 & 0.830 & 0.783 & 0.945 & 0.951 & 0.930 \\ 
  G0 & 3 & 20 &   6 & 0.852 & 0.990 & 0.941 & 0.972 & 0.753 & 0.949 & 0.955 & 0.989 & 0.929 \\ 
  &  & &   7 & 0.758 & 0.912 & 0.737 & 0.937 & 0.891 & 0.980 & 0.911 & 0.965 & 0.822 \\ 
  &  & &   8 & 0.793 & 0.375 & 0.606 & 0.646 & 0.985 & 0.925 & 0.943 & 0.965 & 0.851 \\ 
   & & &   9 & 0.975 & 0.505 & 0.853 & 0.425 & 0.949 & 0.734 & 0.938 & 0.979 & 0.854 \\ 
   &  & &  10 & 0.830 & 0.868 & 0.572 & 0.748 & 0.930 & 0.983 & 0.940 & 0.985 & 0.836 \\ 
   & & &   11 & 0.561 & 0.995 & 0.850 & 0.988 & 0.578 & 0.959 & 0.932 & 0.964 & 0.765 \\ 
  &  & &   12 & 0.610 & 0.883 & 0.655 & 0.985 & 0.903 & 0.965 & 0.952 & 0.957 & 0.817 \\ 
  \hline
  & & & 1 & 0.827 & 0.756 & 0.707 & 0.717 & 0.838 & 0.961 & 0.984 & 0.992 & 0.907 \\ 
  &  & &  2 & 0.839 & 0.895 & 0.669 & 0.988 & 0.725 & 0.942 & 0.981 & 0.992 & 0.937 \\ 
  &  & &  3 & 0.700 & 0.950 & 0.855 & 0.987 & 0.645 & 0.790 & 0.954 & 0.994 & 0.900 \\ 
   &  & & 4 & 0.751 & 0.531 & 0.486 & 0.830 & 0.974 & 0.973 & 0.954 & 0.990 & 0.936 \\ 
  &  & &  5 & 0.967 & 0.990 & 0.680 & 0.945 & 0.775 & 0.752 & 0.946 & 0.943 & 0.927 \\ 
 G3 & 1 & 20 &  6 & 0.844 & 0.991 & 0.928 & 0.965 & 0.711 & 0.940 & 0.955 & 0.992 & 0.930 \\ 
   & & &  7 & 0.734 & 0.891 & 0.706 & 0.921 & 0.875 & 0.973 & 0.921 & 0.969 & 0.830 \\ 
  &  & &  8 & 0.778 & 0.350 & 0.559 & 0.630 & 0.983 & 0.920 & 0.947 & 0.980 & 0.856 \\ 
  &  & &  9 & 0.976 & 0.469 & 0.821 & 0.377 & 0.952 & 0.723 & 0.935 & 0.986 & 0.849 \\ 
  &  & &  10 & 0.800 & 0.839 & 0.539 & 0.706 & 0.918 & 0.984 & 0.945 & 0.988 & 0.848 \\ 
  &  & &  11 & 0.514 & 0.992 & 0.812 & 0.992 & 0.561 & 0.945 & 0.942 & 0.972 & 0.759 \\ 
  &  & &  12 & 0.533 & 0.846 & 0.603 & 0.986 & 0.901 & 0.956 & 0.951 & 0.961 & 0.828 \\  
   \hline
  &  & & 1 & 0.824 & 0.758 & 0.696 & 0.712 & 0.846 & 0.960 & 0.985 & 0.990 & 0.914 \\ 
  &  & &  2 & 0.829 & 0.892 & 0.653 & 0.991 & 0.723 & 0.944 & 0.984 & 0.993 & 0.940 \\ 
  &  & &  3 & 0.700 & 0.952 & 0.847 & 0.986 & 0.646 & 0.794 & 0.953 & 0.992 & 0.898 \\ 
  &  & &  4 & 0.749 & 0.532 & 0.475 & 0.826 & 0.973 & 0.973 & 0.961 & 0.994 & 0.939 \\ 
  &   & & 5 & 0.955 & 0.988 & 0.685 & 0.940 & 0.776 & 0.750 & 0.945 & 0.942 & 0.931 \\ 
 G3 &  3 & 20 &  6 & 0.845 & 0.991 & 0.928 & 0.966 & 0.707 & 0.936 & 0.963 & 0.992 & 0.934 \\ 
  &  & &  7 & 0.733 & 0.890 & 0.699 & 0.921 & 0.875 & 0.970 & 0.925 & 0.966 & 0.845 \\ 
  &  & &  8 & 0.773 & 0.351 & 0.550 & 0.628 & 0.981 & 0.919 & 0.945 & 0.979 & 0.852 \\ 
   &  & & 9 & 0.974 & 0.469 & 0.820 & 0.369 & 0.954 & 0.719 & 0.946 & 0.987 & 0.853 \\ 
  &  & &  10 & 0.785 & 0.839 & 0.533 & 0.715 & 0.916 & 0.981 & 0.946 & 0.986 & 0.847 \\ 
  &  & &  11 & 0.498 & 0.992 & 0.815 & 0.993 & 0.566 & 0.942 & 0.934 & 0.972 & 0.759 \\ 
  &   & &  12 & 0.539 & 0.850 & 0.594 & 0.989 & 0.894 & 0.954 & 0.951 & 0.961 & 0.824 \\ 
\end{tabular}
\caption{Simulated coverage of 90\%-confidence intervals of seasonal structural impulse response of industrial production ($y^{ip}$), inflation $(\pi)$ and federal funds rate ($i$) $k=12$ months after a shock to aggregate demand $(ad)$, aggregate supply $(as)$ and monetary policy ($mp$). Structural schocks are generated using GARCH specification G0 and G3.}
\label{Cov_G0G3_N100}
\end{table}

\begin{table}[t!]
\centering
\scriptsize
\begin{tabular}{l|l|c|ccccccccc}
 $ b$ &$ N$ & $s$ & $ ad \to y^{ip}$ & $ ad \to  \pi$ &  $ ad \to  i$ & $as \to y^{ip}$ & $ as\to \pi$ & $ as \to i$ & $mp \to y^{ip}$ & $ mp \to \pi$ &  $mp \to i$ \\ 
  \hline
 &  &    1 & 0.847 & 0.746 & 0.996 & 0.623 & 0.914 & 0.982 & 1.000 & 0.992 & 0.962 \\ 
   &  &    2 & 0.393 & 0.976 & 0.980 & 0.993 & 0.692 & 0.986 & 1.000 & 0.993 & 0.986 \\ 
   &  &    3 & 0.477 & 0.990 & 0.998 & 0.983 & 0.710 & 0.973 & 1.000 & 0.984 & 0.915 \\ 
   &  &    4 & 0.641 & 0.821 & 0.925 & 0.835 & 0.958 & 0.977 & 1.000 & 0.991 & 0.962 \\ 
   &  &    5 & 0.446 & 0.995 & 0.919 & 0.986 & 0.952 & 0.909 & 1.000 & 0.800 & 0.952 \\ 
  1 & 20  &    6 & 0.849 & 0.999 & 1.000 & 0.996 & 0.661 & 0.994 & 1.000 & 0.990 & 0.967 \\ 
   &  &    7 & 0.785 & 0.979 & 0.992 & 0.974 & 0.824 & 0.989 & 1.000 & 0.994 & 0.835 \\ 
   &  &    8 & 0.970 & 0.584 & 0.980 & 0.567 & 0.917 & 0.970 & 1.000 & 0.985 & 0.841 \\ 
   &  &    9 & 0.986 & 0.588 & 0.998 & 0.303 & 0.961 & 0.943 & 1.000 & 0.988 & 0.843 \\ 
   &  &   10 & 0.852 & 0.961 & 0.876 & 0.794 & 0.894 & 0.987 & 1.000 & 0.978 & 0.852 \\ 
   &  &   11 & 0.391 & 0.993 & 0.996 & 0.988 & 0.478 & 0.977 & 1.000 & 0.985 & 0.708 \\ 
   &  &   12 & 0.400 & 0.995 & 0.986 & 0.990 & 0.890 & 0.986 & 1.000 & 0.976 & 0.919 \\  
   \hline
 &  &    1 & 0.847 & 0.759 & 0.997 & 0.620 & 0.914 & 0.977 & 1.000 & 0.994 & 0.961 \\ 
   &  &    2 & 0.386 & 0.974 & 0.978 & 0.991 & 0.695 & 0.987 & 1.000 & 0.987 & 0.982 \\ 
   &  &    3 & 0.477 & 0.989 & 0.999 & 0.984 & 0.709 & 0.971 & 1.000 & 0.988 & 0.911 \\ 
   &  &    4 & 0.638 & 0.812 & 0.929 & 0.832 & 0.964 & 0.975 & 1.000 & 0.991 & 0.961 \\ 
   &  &    5 & 0.443 & 0.996 & 0.912 & 0.988 & 0.952 & 0.913 & 1.000 & 0.791 & 0.954 \\ 
  3 & 20  &    6 & 0.853 & 0.999 & 1.000 & 0.994 & 0.654 & 0.995 & 1.000 & 0.990 & 0.971 \\ 
   &  &    7 & 0.788 & 0.977 & 0.992 & 0.976 & 0.824 & 0.987 & 1.000 & 0.992 & 0.834 \\ 
   &  &    8 & 0.969 & 0.583 & 0.982 & 0.551 & 0.914 & 0.969 & 1.000 & 0.989 & 0.849 \\ 
   &  &    9 & 0.984 & 0.585 & 0.999 & 0.305 & 0.961 & 0.943 & 1.000 & 0.988 & 0.836 \\ 
   &  &   10 & 0.861 & 0.966 & 0.883 & 0.793 & 0.887 & 0.988 & 1.000 & 0.982 & 0.854 \\ 
   &  &   11 & 0.376 & 0.992 & 0.996 & 0.987 & 0.484 & 0.978 & 1.000 & 0.984 & 0.713 \\ 
   &  &   12 & 0.394 & 0.994 & 0.984 & 0.989 & 0.889 & 0.987 & 1.000 & 0.977 & 0.916 \\ 
  \hline
 &  &    1 & 0.958 & 0.774 & 0.982 & 0.805 & 0.915 & 0.898 & 1.000 & 0.913 & 0.919 \\ 
   &  &    2 & 0.717 & 0.898 & 0.905 & 0.927 & 0.777 & 0.903 & 1.000 & 0.910 & 0.912 \\ 
   &  &    3 & 0.716 & 0.956 & 0.970 & 0.934 & 0.787 & 0.927 & 1.000 & 0.937 & 0.848 \\ 
   &  &    4 & 0.922 & 0.779 & 0.890 & 0.849 & 0.946 & 0.890 & 1.000 & 0.933 & 0.891 \\ 
   &  &    5 & 0.698 & 0.935 & 0.866 & 0.924 & 0.887 & 0.938 & 1.000 & 0.907 & 0.864 \\ 
  1  & 50 &    6 & 0.941 & 0.964 & 0.969 & 0.965 & 0.739 & 0.922 & 1.000 & 0.909 & 0.901 \\ 
   &  &    7 & 0.924 & 0.923 & 0.973 & 0.954 & 0.793 & 0.928 & 1.000 & 0.910 & 0.806 \\ 
   &  &    8 & 0.923 & 0.699 & 0.959 & 0.759 & 0.927 & 0.848 & 1.000 & 0.890 & 0.798 \\ 
   &  &    9 & 0.940 & 0.691 & 0.977 & 0.606 & 0.962 & 0.843 & 1.000 & 0.918 & 0.808 \\ 
   &  &   10 & 0.934 & 0.905 & 0.865 & 0.859 & 0.862 & 0.897 & 1.000 & 0.938 & 0.862 \\ 
   &  &   11 & 0.688 & 0.964 & 0.982 & 0.949 & 0.662 & 0.889 & 1.000 & 0.880 & 0.765 \\ 
   &  &   12 & 0.682 & 0.952 & 0.943 & 0.932 & 0.843 & 0.912 & 1.000 & 0.879 & 0.882 \\ 
   \hline
&  &    1 & 0.957 & 0.774 & 0.981 & 0.800 & 0.917 & 0.905 & 1.000 & 0.914 & 0.915 \\ 
   &  &    2 & 0.712 & 0.902 & 0.907 & 0.923 & 0.780 & 0.899 & 1.000 & 0.916 & 0.919 \\ 
   &  &    3 & 0.712 & 0.962 & 0.967 & 0.929 & 0.779 & 0.928 & 1.000 & 0.940 & 0.848 \\ 
   &  &    4 & 0.920 & 0.777 & 0.885 & 0.843 & 0.942 & 0.891 & 1.000 & 0.938 & 0.883 \\ 
   &  &    5 & 0.693 & 0.937 & 0.869 & 0.931 & 0.890 & 0.941 & 1.000 & 0.902 & 0.860 \\ 
 5  & 50  &    6 & 0.945 & 0.967 & 0.970 & 0.962 & 0.737 & 0.924 & 1.000 & 0.906 & 0.901 \\ 
   &  &    7 & 0.926 & 0.921 & 0.971 & 0.949 & 0.792 & 0.932 & 1.000 & 0.909 & 0.805 \\ 
   &  &    8 & 0.928 & 0.701 & 0.957 & 0.759 & 0.927 & 0.852 & 1.000 & 0.888 & 0.798 \\ 
   &  &    9 & 0.936 & 0.688 & 0.978 & 0.601 & 0.964 & 0.850 & 1.000 & 0.916 & 0.806 \\ 
   &  &   10 & 0.929 & 0.911 & 0.862 & 0.850 & 0.861 & 0.900 & 1.000 & 0.939 & 0.858 \\ 
   &  &   11 & 0.690 & 0.961 & 0.983 & 0.952 & 0.653 & 0.891 & 1.000 & 0.885 & 0.761 \\ 
   &  &   12 & 0.686 & 0.954 & 0.943 & 0.931 & 0.852 & 0.922 & 1.000 & 0.883 & 0.883 \\ 
\end{tabular}
\caption{Simulated coverage of 90\%-confidence intervals of contemporaneous ($k = 0$) seasonal structural impulse responses generated by aggregate demand $(ad)$, aggregate supply ($as$) and monetary policy $(mp)$ shocks on industrial production ($y^{ip}$), inflation $(\pi)$ and federal funds rate ($i$). Structural schocks are generated using GARCH specification G0.} 
\label{Cov_b1swn_k0}
\end{table}

\begin{table}[t!]
\centering
\scriptsize
\begin{tabular}{l|l|c|ccccccccc}
 $b$ &$ N$ & $s$ & $ ad \to y^{ip}$ & $ ad \to  \pi$ &  $ ad \to  i$ & $as \to y^{ip}$ & $ as\to \pi$ & $ as \to i$ & $mp \to y^{ip}$ & $ mp \to \pi$ &  $mp \to i$ \\ 
  \hline
  &  &    1 & 0.789 & 0.728 & 1.000 & 0.579 & 0.886 & 0.980 & 1.000 & 0.992 & 0.952 \\ 
   &  &    2 & 0.333 & 0.977 & 0.981 & 0.996 & 0.639 & 0.991 & 1.000 & 0.991 & 0.975 \\ 
   &  &    3 & 0.385 & 0.993 & 0.998 & 0.981 & 0.670 & 0.961 & 1.000 & 0.974 & 0.887 \\ 
   &  &    4 & 0.539 & 0.813 & 0.921 & 0.798 & 0.949 & 0.988 & 1.000 & 0.986 & 0.944 \\ 
   &  &    5 & 0.336 & 0.997 & 0.898 & 0.986 & 0.946 & 0.873 & 1.000 & 0.714 & 0.944 \\ 
  1 & 20 &    6 & 0.753 & 0.997 & 0.999 & 0.995 & 0.618 & 0.990 & 1.000 & 0.987 & 0.965 \\ 
   &  &    7 & 0.694 & 0.972 & 0.988 & 0.974 & 0.811 & 0.985 & 1.000 & 0.991 & 0.821 \\ 
   &  &    8 & 0.944 & 0.542 & 0.979 & 0.515 & 0.893 & 0.977 & 1.000 & 0.984 & 0.820 \\ 
   &  &    9 & 0.982 & 0.540 & 0.997 & 0.267 & 0.951 & 0.954 & 1.000 & 0.990 & 0.792 \\ 
   &  &   10 & 0.786 & 0.967 & 0.878 & 0.754 & 0.890 & 0.986 & 1.000 & 0.966 & 0.854 \\ 
   &  &   11 & 0.316 & 0.995 & 0.999 & 0.990 & 0.448 & 0.981 & 1.000 & 0.984 & 0.634 \\ 
   &  &   12 & 0.300 & 0.991 & 0.982 & 0.992 & 0.865 & 0.989 & 1.000 & 0.983 & 0.891 \\ 
   \hline
  &  &    1 & 0.788 & 0.737 & 1.000 & 0.577 & 0.893 & 0.984 & 1.000 & 0.992 & 0.953 \\ 
   &  &    2 & 0.322 & 0.972 & 0.980 & 0.994 & 0.631 & 0.991 & 1.000 & 0.993 & 0.975 \\ 
   &  &    3 & 0.381 & 0.995 & 0.998 & 0.982 & 0.657 & 0.960 & 1.000 & 0.974 & 0.887 \\ 
   &  &    4 & 0.539 & 0.816 & 0.917 & 0.802 & 0.952 & 0.988 & 1.000 & 0.987 & 0.949 \\ 
   &  &    5 & 0.334 & 0.996 & 0.895 & 0.986 & 0.949 & 0.873 & 1.000 & 0.711 & 0.950 \\ 
  3 & 20 &    6 & 0.750 & 0.998 & 0.999 & 0.996 & 0.607 & 0.992 & 1.000 & 0.991 & 0.963 \\ 
   &  &    7 & 0.695 & 0.970 & 0.992 & 0.978 & 0.811 & 0.989 & 1.000 & 0.988 & 0.827 \\ 
   &  &    8 & 0.944 & 0.530 & 0.979 & 0.510 & 0.889 & 0.976 & 1.000 & 0.987 & 0.809 \\ 
   &  &    9 & 0.985 & 0.539 & 0.998 & 0.262 & 0.956 & 0.952 & 1.000 & 0.992 & 0.791 \\ 
   &  &   10 & 0.791 & 0.965 & 0.883 & 0.749 & 0.889 & 0.989 & 1.000 & 0.968 & 0.851 \\ 
   &  &   11 & 0.320 & 0.996 & 0.998 & 0.991 & 0.444 & 0.980 & 1.000 & 0.986 & 0.639 \\ 
   &  &   12 & 0.294 & 0.993 & 0.980 & 0.989 & 0.853 & 0.990 & 1.000 & 0.979 & 0.894 \\ 
    \hline
 &  &    1 & 0.932 & 0.739 & 0.977 & 0.768 & 0.883 & 0.906 & 1.000 & 0.910 & 0.877 \\ 
   &  &    2 & 0.646 & 0.896 & 0.900 & 0.938 & 0.752 & 0.922 & 1.000 & 0.916 & 0.897 \\ 
   &  &    3 & 0.627 & 0.960 & 0.970 & 0.936 & 0.729 & 0.908 & 1.000 & 0.918 & 0.820 \\ 
   &  &    4 & 0.866 & 0.751 & 0.871 & 0.833 & 0.927 & 0.876 & 1.000 & 0.930 & 0.863 \\ 
   &  &    5 & 0.595 & 0.948 & 0.863 & 0.934 & 0.907 & 0.896 & 1.000 & 0.831 & 0.841 \\ 
  1 & 50 &    6 & 0.866 & 0.966 & 0.966 & 0.958 & 0.676 & 0.929 & 1.000 & 0.917 & 0.861 \\ 
   &  &    7 & 0.875 & 0.919 & 0.970 & 0.942 & 0.771 & 0.934 & 1.000 & 0.923 & 0.778 \\ 
   &  &    8 & 0.926 & 0.679 & 0.960 & 0.724 & 0.917 & 0.865 & 1.000 & 0.906 & 0.755 \\ 
   &  &    9 & 0.933 & 0.659 & 0.985 & 0.563 & 0.954 & 0.860 & 1.000 & 0.919 & 0.756 \\ 
   &  &   10 & 0.907 & 0.910 & 0.827 & 0.835 & 0.842 & 0.893 & 1.000 & 0.936 & 0.824 \\ 
   &  &   11 & 0.570 & 0.962 & 0.974 & 0.945 & 0.594 & 0.889 & 1.000 & 0.894 & 0.694 \\ 
   &  &   12 & 0.574 & 0.937 & 0.947 & 0.935 & 0.811 & 0.918 & 1.000 & 0.890 & 0.837 \\  
   \hline
   &  &    1 & 0.933 & 0.734 & 0.977 & 0.774 & 0.889 & 0.909 & 1.000 & 0.911 & 0.883 \\ 
   &  &    2 & 0.630 & 0.911 & 0.903 & 0.939 & 0.751 & 0.928 & 1.000 & 0.915 & 0.897 \\ 
   &  &    3 & 0.625 & 0.962 & 0.979 & 0.943 & 0.732 & 0.912 & 1.000 & 0.923 & 0.819 \\ 
   &  &    4 & 0.863 & 0.761 & 0.876 & 0.838 & 0.927 & 0.880 & 1.000 & 0.930 & 0.864 \\ 
   &  &    5 & 0.588 & 0.953 & 0.865 & 0.939 & 0.906 & 0.895 & 1.000 & 0.830 & 0.839 \\ 
  5 & 50 &    6 & 0.865 & 0.974 & 0.979 & 0.963 & 0.665 & 0.930 & 1.000 & 0.913 & 0.862 \\ 
   &  &    7 & 0.869 & 0.919 & 0.968 & 0.949 & 0.776 & 0.942 & 1.000 & 0.927 & 0.782 \\ 
   &  &    8 & 0.922 & 0.685 & 0.961 & 0.738 & 0.919 & 0.878 & 1.000 & 0.916 & 0.753 \\ 
   &  &    9 & 0.939 & 0.653 & 0.983 & 0.558 & 0.959 & 0.863 & 1.000 & 0.928 & 0.746 \\ 
   &  &   10 & 0.916 & 0.918 & 0.836 & 0.839 & 0.841 & 0.900 & 1.000 & 0.940 & 0.826 \\ 
   &  &   11 & 0.563 & 0.962 & 0.973 & 0.954 & 0.577 & 0.890 & 1.000 & 0.891 & 0.683 \\ 
   &  &   12 & 0.569 & 0.942 & 0.947 & 0.934 & 0.814 & 0.928 & 1.000 & 0.889 & 0.833 \\ 
\end{tabular}
\caption{Simulated coverage of 90\%-confidence intervals of contemporaneous ($k = 0$) seasonal structural impulse responses generated by aggregate demand $(ad)$, aggregate supply ($as$) and monetary policy $(mp)$ shocks on industrial production ($y^{ip}$), inflation $(\pi)$ and federal funds rate ($i$). Structural schocks are generated using GARCH specification G3.} 
\label{Cov_b1para3_k0}
\end{table}

\begin{table}[h!]
\centering
\scriptsize
\begin{tabular}{l|l|l|c|ccccccccc}
 & $b$ &$ N$ & $s$ & $ ad \to y^{ip}$ & $ ad \to  \pi$ &  $ ad \to  i$ & $as \to y^{ip}$ & $ as\to \pi$ & $ as \to i$ & $mp \to y^{ip}$ & $ mp \to \pi$ &  $mp \to i$ \\ 
    \hline
 &  &  &    1 & 0.927 & 0.823 & 0.943 & 0.854 & 0.913 & 0.880 & 1.000 & 0.895 & 0.896 \\ 
   &  &  &    2 & 0.786 & 0.867 & 0.903 & 0.877 & 0.850 & 0.887 & 1.000 & 0.883 & 0.897 \\ 
   &  &  &    3 & 0.800 & 0.929 & 0.944 & 0.931 & 0.838 & 0.923 & 1.000 & 0.900 & 0.851 \\ 
   &  &  &    4 & 0.937 & 0.821 & 0.890 & 0.876 & 0.937 & 0.883 & 1.000 & 0.903 & 0.869 \\ 
   &  &  &    5 & 0.782 & 0.913 & 0.878 & 0.916 & 0.903 & 0.901 & 1.000 & 0.912 & 0.861 \\ 
 G0  & 1  & 100 &    6 & 0.944 & 0.933 & 0.925 & 0.918 & 0.823 & 0.901 & 1.000 & 0.890 & 0.883 \\ 
   &  &  &    7 & 0.936 & 0.913 & 0.956 & 0.940 & 0.816 & 0.914 & 1.000 & 0.900 & 0.846 \\ 
   &  &  &    8 & 0.903 & 0.796 & 0.958 & 0.833 & 0.938 & 0.866 & 1.000 & 0.881 & 0.823 \\ 
   &  &  &    9 & 0.887 & 0.760 & 0.945 & 0.750 & 0.931 & 0.851 & 1.000 & 0.890 & 0.834 \\ 
   &  &  &   10 & 0.928 & 0.896 & 0.893 & 0.877 & 0.870 & 0.891 & 1.000 & 0.908 & 0.873 \\ 
   &  &  &   11 & 0.765 & 0.931 & 0.949 & 0.926 & 0.749 & 0.884 & 1.000 & 0.877 & 0.804 \\ 
   &  &  &   12 & 0.773 & 0.910 & 0.927 & 0.910 & 0.871 & 0.912 & 1.000 & 0.881 & 0.873 \\ 
   \hline
  &  &  &    1 & 0.912 & 0.824 & 0.943 & 0.841 & 0.916 & 0.883 & 1.000 & 0.904 & 0.897 \\ 
   &  &  &    2 & 0.788 & 0.869 & 0.905 & 0.871 & 0.843 & 0.888 & 1.000 & 0.887 & 0.896 \\ 
   &  &  &    3 & 0.802 & 0.928 & 0.945 & 0.931 & 0.840 & 0.921 & 1.000 & 0.906 & 0.862 \\ 
   &  &  &    4 & 0.943 & 0.819 & 0.892 & 0.879 & 0.938 & 0.886 & 1.000 & 0.905 & 0.862 \\ 
   &  &  &    5 & 0.782 & 0.915 & 0.883 & 0.911 & 0.907 & 0.907 & 1.000 & 0.914 & 0.859 \\ 
   G0 & 7 & 100  &    6 & 0.942 & 0.934 & 0.927 & 0.919 & 0.814 & 0.899 & 1.000 & 0.900 & 0.885 \\ 
   &  &  &    7 & 0.934 & 0.908 & 0.958 & 0.936 & 0.812 & 0.915 & 1.000 & 0.905 & 0.845 \\ 
   &  &  &    8 & 0.909 & 0.797 & 0.955 & 0.837 & 0.941 & 0.865 & 1.000 & 0.889 & 0.819 \\ 
   &  &  &    9 & 0.885 & 0.769 & 0.947 & 0.740 & 0.924 & 0.851 & 1.000 & 0.897 & 0.849 \\ 
   &  &  &   10 & 0.923 & 0.889 & 0.896 & 0.883 & 0.873 & 0.891 & 1.000 & 0.912 & 0.877 \\ 
   &  &  &   11 & 0.767 & 0.934 & 0.948 & 0.923 & 0.747 & 0.889 & 1.000 & 0.880 & 0.804 \\ 
   &  &  &   12 & 0.773 & 0.903 & 0.931 & 0.915 & 0.866 & 0.911 & 1.000 & 0.886 & 0.884 \\ 
    \hline
 &  &  &    1 & 0.898 & 0.796 & 0.938 & 0.834 & 0.885 & 0.888 & 1.000 & 0.898 & 0.879 \\ 
   &  &  &    2 & 0.731 & 0.866 & 0.879 & 0.886 & 0.825 & 0.889 & 1.000 & 0.909 & 0.875 \\ 
   &  &  &    3 & 0.719 & 0.923 & 0.926 & 0.937 & 0.812 & 0.902 & 1.000 & 0.895 & 0.848 \\ 
   &  &  &    4 & 0.927 & 0.805 & 0.873 & 0.871 & 0.928 & 0.889 & 1.000 & 0.903 & 0.853 \\ 
   &  &  &    5 & 0.716 & 0.928 & 0.864 & 0.926 & 0.923 & 0.887 & 1.000 & 0.833 & 0.851 \\ 
  G3 & 1 & 100 &    6 & 0.884 & 0.925 & 0.928 & 0.923 & 0.780 & 0.908 & 1.000 & 0.888 & 0.863 \\ 
   &  &  &    7 & 0.906 & 0.913 & 0.943 & 0.935 & 0.780 & 0.933 & 1.000 & 0.900 & 0.820 \\ 
   &  &  &    8 & 0.888 & 0.772 & 0.943 & 0.806 & 0.926 & 0.872 & 1.000 & 0.889 & 0.811 \\ 
   &  &  &    9 & 0.880 & 0.721 & 0.946 & 0.694 & 0.930 & 0.871 & 1.000 & 0.903 & 0.823 \\ 
   &  &  &   10 & 0.886 & 0.897 & 0.853 & 0.874 & 0.864 & 0.888 & 1.000 & 0.903 & 0.864 \\ 
   &  &  &   11 & 0.695 & 0.913 & 0.946 & 0.922 & 0.708 & 0.904 & 1.000 & 0.893 & 0.747 \\ 
   &  &  &   12 & 0.696 & 0.909 & 0.932 & 0.911 & 0.858 & 0.895 & 1.000 & 0.882 & 0.851 \\ 
   \hline
  &  &  &    1 & 0.902 & 0.795 & 0.944 & 0.843 & 0.889 & 0.892 & 1.000 & 0.905 & 0.881 \\ 
   &  &  &    2 & 0.726 & 0.874 & 0.878 & 0.894 & 0.820 & 0.895 & 1.000 & 0.913 & 0.889 \\ 
   &  &  &    3 & 0.711 & 0.930 & 0.938 & 0.944 & 0.811 & 0.910 & 1.000 & 0.906 & 0.846 \\ 
   &  &  &    4 & 0.936 & 0.809 & 0.878 & 0.875 & 0.933 & 0.889 & 1.000 & 0.916 & 0.853 \\ 
   &  &  &    5 & 0.701 & 0.930 & 0.867 & 0.926 & 0.920 & 0.888 & 1.000 & 0.836 & 0.853 \\ 
  G3  & 7 & 100 &    6 & 0.887 & 0.937 & 0.936 & 0.933 & 0.775 & 0.915 & 1.000 & 0.897 & 0.871 \\ 
   &  &  &    7 & 0.908 & 0.927 & 0.945 & 0.940 & 0.771 & 0.933 & 1.000 & 0.904 & 0.834 \\ 
   &  &  &    8 & 0.902 & 0.779 & 0.939 & 0.804 & 0.935 & 0.875 & 1.000 & 0.896 & 0.814 \\ 
   &  &  &    9 & 0.885 & 0.723 & 0.948 & 0.684 & 0.926 & 0.872 & 1.000 & 0.909 & 0.811 \\ 
   &  &  &   10 & 0.899 & 0.894 & 0.854 & 0.866 & 0.861 & 0.890 & 1.000 & 0.907 & 0.867 \\ 
   &  &  &   11 & 0.675 & 0.923 & 0.955 & 0.926 & 0.706 & 0.908 & 1.000 & 0.893 & 0.751 \\ 
   &  &  &   12 & 0.683 & 0.917 & 0.938 & 0.914 & 0.860 & 0.902 & 1.000 & 0.885 & 0.866 \\ 
\end{tabular}
\caption{Simulated coverage of 90\%-confidence intervals of contemporaneous ($k = 0$) seasonal structural impulse responses generated by aggregate demand $(ad)$, aggregate supply ($as$) and monetary policy $(mp)$ shocks on industrial production ($y^{ip}$), inflation $(\pi)$ and federal funds rate ($i$). Structural schocks are generated using GARCH specification G0 and G3.}
\label{Cov_G0G3_N100_k0}
\end{table}

\begin{table}[t!]
\centering
\scriptsize
\begin{tabular}{l|l|c|ccccccccc}
$b$&$ N$ & $s$ & $ ad \to y^{ip}$ & $ ad \to  \pi$ &  $ ad \to  i$ & $as \to y^{ip}$ & $ as\to \pi$ & $ as \to i$ & $mp \to y^{ip}$ & $ mp \to \pi$ &  $mp \to i$ \\ 
  \hline
&  &    1 & 0.826 & 0.746 & 0.997 & 0.607 & 0.901 & 0.978 & 1.000 & 0.988 & 0.947 \\ 
   &  &    2 & 0.374 & 0.975 & 0.985 & 0.994 & 0.677 & 0.992 & 1.000 & 0.992 & 0.972 \\ 
   &  &    3 & 0.442 & 0.986 & 0.999 & 0.986 & 0.701 & 0.971 & 1.000 & 0.984 & 0.904 \\ 
   &  &    4 & 0.598 & 0.820 & 0.919 & 0.821 & 0.953 & 0.977 & 1.000 & 0.987 & 0.950 \\ 
   &  &    5 & 0.402 & 0.995 & 0.909 & 0.987 & 0.948 & 0.885 & 1.000 & 0.766 & 0.947 \\ 
 1  & 20 &    6 & 0.817 & 0.999 & 0.999 & 0.996 & 0.653 & 0.993 & 1.000 & 0.986 & 0.963 \\ 
   &  &    7 & 0.741 & 0.972 & 0.994 & 0.973 & 0.812 & 0.985 & 1.000 & 0.992 & 0.833 \\ 
   &  &    8 & 0.958 & 0.567 & 0.982 & 0.542 & 0.904 & 0.967 & 1.000 & 0.987 & 0.828 \\ 
   &  &    9 & 0.986 & 0.592 & 0.998 & 0.292 & 0.962 & 0.949 & 1.000 & 0.986 & 0.823 \\ 
   &  &   10 & 0.835 & 0.960 & 0.877 & 0.780 & 0.880 & 0.990 & 1.000 & 0.975 & 0.850 \\ 
   &  &   11 & 0.375 & 0.993 & 0.997 & 0.991 & 0.479 & 0.976 & 1.000 & 0.983 & 0.690 \\ 
   &  &   12 & 0.375 & 0.991 & 0.984 & 0.989 & 0.880 & 0.985 & 1.000 & 0.979 & 0.911 \\  
   \hline
  &  &    1 & 0.826 & 0.749 & 0.995 & 0.605 & 0.901 & 0.978 & 1.000 & 0.990 & 0.953 \\ 
   &  &    2 & 0.365 & 0.974 & 0.981 & 0.994 & 0.681 & 0.989 & 1.000 & 0.992 & 0.970 \\ 
   &  &    3 & 0.443 & 0.990 & 0.999 & 0.984 & 0.703 & 0.971 & 1.000 & 0.986 & 0.899 \\ 
   &  &    4 & 0.607 & 0.820 & 0.926 & 0.822 & 0.953 & 0.974 & 1.000 & 0.986 & 0.955 \\ 
   &  &    5 & 0.398 & 0.995 & 0.909 & 0.988 & 0.948 & 0.888 & 1.000 & 0.766 & 0.953 \\ 
  3 & 20 &    6 & 0.819 & 1.000 & 1.000 & 0.996 & 0.647 & 0.993 & 1.000 & 0.986 & 0.957 \\ 
   &  &    7 & 0.747 & 0.972 & 0.994 & 0.973 & 0.816 & 0.987 & 1.000 & 0.995 & 0.832 \\ 
   &  &    8 & 0.961 & 0.565 & 0.985 & 0.544 & 0.904 & 0.970 & 1.000 & 0.989 & 0.829 \\ 
   &  &    9 & 0.986 & 0.587 & 0.998 & 0.288 & 0.961 & 0.954 & 1.000 & 0.988 & 0.822 \\ 
   &  &   10 & 0.840 & 0.964 & 0.878 & 0.785 & 0.887 & 0.988 & 1.000 & 0.974 & 0.847 \\ 
   &  &   11 & 0.366 & 0.994 & 0.997 & 0.988 & 0.473 & 0.978 & 1.000 & 0.985 & 0.675 \\ 
   &  &   12 & 0.368 & 0.994 & 0.981 & 0.990 & 0.873 & 0.984 & 1.000 & 0.980 & 0.912 \\  
    \hline
  &  &    1 & 0.941 & 0.762 & 0.983 & 0.791 & 0.912 & 0.907 & 1.000 & 0.917 & 0.898 \\ 
   &  &    2 & 0.675 & 0.900 & 0.909 & 0.922 & 0.771 & 0.900 & 1.000 & 0.911 & 0.908 \\ 
   &  &    3 & 0.681 & 0.957 & 0.970 & 0.938 & 0.772 & 0.923 & 1.000 & 0.941 & 0.834 \\ 
   &  &    4 & 0.907 & 0.762 & 0.880 & 0.836 & 0.933 & 0.888 & 1.000 & 0.931 & 0.877 \\ 
   &  &    5 & 0.669 & 0.935 & 0.853 & 0.928 & 0.885 & 0.935 & 1.000 & 0.888 & 0.864 \\ 
  1 & 50 &    6 & 0.918 & 0.962 & 0.974 & 0.963 & 0.722 & 0.926 & 1.000 & 0.907 & 0.883 \\ 
   &  &    7 & 0.915 & 0.915 & 0.969 & 0.947 & 0.780 & 0.928 & 1.000 & 0.910 & 0.800 \\ 
   &  &    8 & 0.915 & 0.680 & 0.955 & 0.745 & 0.922 & 0.849 & 1.000 & 0.893 & 0.781 \\ 
   &  &    9 & 0.936 & 0.686 & 0.979 & 0.590 & 0.957 & 0.846 & 1.000 & 0.916 & 0.789 \\ 
   &  &   10 & 0.920 & 0.906 & 0.851 & 0.846 & 0.859 & 0.887 & 1.000 & 0.936 & 0.854 \\ 
   &  &   11 & 0.660 & 0.959 & 0.980 & 0.943 & 0.627 & 0.882 & 1.000 & 0.881 & 0.752 \\ 
   &  &   12 & 0.647 & 0.950 & 0.932 & 0.928 & 0.830 & 0.911 & 1.000 & 0.884 & 0.868 \\   
   \hline
 &  &    1 & 0.943 & 0.753 & 0.978 & 0.788 & 0.916 & 0.902 & 1.000 & 0.921 & 0.906 \\ 
   &  &    2 & 0.667 & 0.906 & 0.905 & 0.932 & 0.771 & 0.900 & 1.000 & 0.910 & 0.907 \\ 
   &  &    3 & 0.678 & 0.960 & 0.972 & 0.932 & 0.774 & 0.925 & 1.000 & 0.934 & 0.832 \\ 
   &  &    4 & 0.903 & 0.768 & 0.872 & 0.842 & 0.930 & 0.886 & 1.000 & 0.929 & 0.881 \\ 
   &  &    5 & 0.665 & 0.937 & 0.852 & 0.932 & 0.884 & 0.932 & 1.000 & 0.886 & 0.858 \\ 
 5  & 50 &    6 & 0.917 & 0.964 & 0.971 & 0.960 & 0.719 & 0.922 & 1.000 & 0.907 & 0.890 \\ 
   &  &    7 & 0.917 & 0.916 & 0.968 & 0.946 & 0.789 & 0.931 & 1.000 & 0.916 & 0.795 \\ 
   &  &    8 & 0.919 & 0.688 & 0.953 & 0.741 & 0.920 & 0.843 & 1.000 & 0.890 & 0.783 \\ 
   &  &    9 & 0.937 & 0.687 & 0.981 & 0.584 & 0.960 & 0.843 & 1.000 & 0.910 & 0.785 \\ 
   &  &   10 & 0.928 & 0.908 & 0.854 & 0.845 & 0.862 & 0.897 & 1.000 & 0.933 & 0.856 \\ 
   &  &   11 & 0.662 & 0.957 & 0.981 & 0.948 & 0.622 & 0.889 & 1.000 & 0.880 & 0.737 \\ 
   &  &   12 & 0.643 & 0.950 & 0.939 & 0.926 & 0.828 & 0.913 & 1.000 & 0.885 & 0.862 \\ 
\end{tabular}
\caption{Simulated coverage of 90\%-confidence intervals of contemporaneous ($k = 0$) seasonal structural impulse responses generated by aggregate demand $(ad)$, aggregate supply ($as$) and monetary policy $(mp)$ shocks on industrial production ($y^{ip}$), inflation $(\pi)$ and federal funds rate ($i$). Structural schocks are generated using GARCH specification G1.} 
\label{Cov_b1para1_k0}
\end{table}

\begin{table}[t!]
\centering
\scriptsize
\begin{tabular}{l|l|c|ccccccccc}
$b$&$ N$ & $s$ & $ ad \to y^{ip}$ & $ ad \to  \pi$ &  $ ad \to  i$ & $as \to y^{ip}$ & $ as\to \pi$ & $ as \to i$ & $mp \to y^{ip}$ & $ mp \to \pi$ &  $mp \to i$ \\ 
  \hline
   &  &    1 & 0.699 & 0.687 & 0.996 & 0.506 & 0.834 & 0.982 & 1.000 & 0.997 & 0.924 \\ 
   &  &    2 & 0.283 & 0.973 & 0.978 & 0.995 & 0.582 & 0.995 & 1.000 & 0.993 & 0.951 \\ 
   &  &    3 & 0.314 & 0.992 & 1.000 & 0.981 & 0.600 & 0.934 & 1.000 & 0.950 & 0.858 \\ 
   &  &    4 & 0.456 & 0.777 & 0.900 & 0.746 & 0.932 & 0.979 & 1.000 & 0.968 & 0.930 \\ 
   &  &    5 & 0.291 & 0.995 & 0.886 & 0.982 & 0.942 & 0.830 & 1.000 & 0.633 & 0.941 \\ 
  1 & 20 &    6 & 0.646 & 1.000 & 0.999 & 0.990 & 0.553 & 0.989 & 1.000 & 0.988 & 0.941 \\ 
   &  &    7 & 0.561 & 0.966 & 0.990 & 0.950 & 0.766 & 0.983 & 1.000 & 0.989 & 0.774 \\ 
   &  &    8 & 0.903 & 0.494 & 0.978 & 0.447 & 0.864 & 0.981 & 1.000 & 0.987 & 0.760 \\ 
   &  &    9 & 0.979 & 0.491 & 0.999 & 0.221 & 0.944 & 0.951 & 1.000 & 0.981 & 0.727 \\ 
   &  &   10 & 0.682 & 0.956 & 0.821 & 0.673 & 0.850 & 0.988 & 1.000 & 0.930 & 0.807 \\ 
   &  &   11 & 0.275 & 0.994 & 0.996 & 0.988 & 0.382 & 0.983 & 1.000 & 0.984 & 0.564 \\ 
   &  &   12 & 0.261 & 0.988 & 0.973 & 0.986 & 0.823 & 0.984 & 1.000 & 0.990 & 0.859 \\ 
   \hline
  &  &    1 & 0.698 & 0.685 & 0.997 & 0.515 & 0.838 & 0.980 & 1.000 & 0.997 & 0.916 \\ 
   &  &    2 & 0.284 & 0.973 & 0.983 & 0.992 & 0.579 & 0.994 & 1.000 & 0.994 & 0.954 \\ 
   &  &    3 & 0.317 & 0.993 & 1.000 & 0.983 & 0.601 & 0.935 & 1.000 & 0.949 & 0.855 \\ 
   &  &    4 & 0.448 & 0.778 & 0.899 & 0.744 & 0.929 & 0.983 & 1.000 & 0.969 & 0.934 \\ 
   &  &    5 & 0.284 & 0.995 & 0.881 & 0.986 & 0.939 & 0.826 & 1.000 & 0.629 & 0.939 \\ 
  3 & 20 &    6 & 0.645 & 1.000 & 1.000 & 0.991 & 0.551 & 0.990 & 1.000 & 0.989 & 0.942 \\ 
   &  &    7 & 0.565 & 0.969 & 0.989 & 0.955 & 0.778 & 0.985 & 1.000 & 0.987 & 0.777 \\ 
   &  &    8 & 0.895 & 0.486 & 0.977 & 0.448 & 0.863 & 0.978 & 1.000 & 0.991 & 0.756 \\ 
   &  &    9 & 0.981 & 0.501 & 0.998 & 0.214 & 0.941 & 0.946 & 1.000 & 0.987 & 0.728 \\ 
   &  &   10 & 0.696 & 0.960 & 0.823 & 0.681 & 0.852 & 0.984 & 1.000 & 0.927 & 0.805 \\ 
   &  &   11 & 0.268 & 0.993 & 0.997 & 0.989 & 0.382 & 0.985 & 1.000 & 0.986 & 0.568 \\ 
   &  &   12 & 0.258 & 0.990 & 0.971 & 0.979 & 0.830 & 0.987 & 1.000 & 0.989 & 0.853 \\  
   \hline
 &  &    1 & 0.883 & 0.694 & 0.971 & 0.721 & 0.864 & 0.896 & 1.000 & 0.912 & 0.812 \\ 
   &  &    2 & 0.520 & 0.885 & 0.887 & 0.929 & 0.682 & 0.918 & 1.000 & 0.912 & 0.852 \\ 
   &  &    3 & 0.542 & 0.956 & 0.962 & 0.939 & 0.688 & 0.896 & 1.000 & 0.920 & 0.770 \\ 
   &  &    4 & 0.799 & 0.713 & 0.830 & 0.807 & 0.888 & 0.881 & 1.000 & 0.907 & 0.819 \\ 
   &  &    5 & 0.510 & 0.946 & 0.825 & 0.925 & 0.869 & 0.846 & 1.000 & 0.747 & 0.811 \\ 
  1 & 50 &    6 & 0.773 & 0.964 & 0.966 & 0.944 & 0.622 & 0.936 & 1.000 & 0.911 & 0.828 \\ 
   &  &    7 & 0.798 & 0.904 & 0.963 & 0.932 & 0.740 & 0.921 & 1.000 & 0.922 & 0.723 \\ 
   &  &    8 & 0.925 & 0.624 & 0.936 & 0.660 & 0.896 & 0.874 & 1.000 & 0.896 & 0.673 \\ 
   &  &    9 & 0.934 & 0.605 & 0.975 & 0.481 & 0.941 & 0.856 & 1.000 & 0.905 & 0.680 \\ 
   &  &   10 & 0.865 & 0.889 & 0.798 & 0.803 & 0.826 & 0.895 & 1.000 & 0.902 & 0.790 \\ 
   &  &   11 & 0.499 & 0.952 & 0.976 & 0.937 & 0.510 & 0.888 & 1.000 & 0.889 & 0.592 \\ 
   &  &   12 & 0.481 & 0.925 & 0.929 & 0.910 & 0.785 & 0.914 & 1.000 & 0.902 & 0.801 \\   
   \hline
  &  &    1 & 0.886 & 0.695 & 0.972 & 0.717 & 0.865 & 0.903 & 1.000 & 0.918 & 0.819 \\ 
   &  &    2 & 0.503 & 0.882 & 0.882 & 0.928 & 0.685 & 0.923 & 1.000 & 0.918 & 0.858 \\ 
   &  &    3 & 0.538 & 0.956 & 0.967 & 0.938 & 0.686 & 0.897 & 1.000 & 0.920 & 0.777 \\ 
   &  &    4 & 0.795 & 0.724 & 0.833 & 0.806 & 0.890 & 0.894 & 1.000 & 0.905 & 0.822 \\ 
   &  &    5 & 0.499 & 0.945 & 0.826 & 0.921 & 0.869 & 0.852 & 1.000 & 0.740 & 0.816 \\ 
  5 & 50 &    6 & 0.775 & 0.967 & 0.974 & 0.950 & 0.614 & 0.944 & 1.000 & 0.911 & 0.832 \\ 
   &  &    7 & 0.805 & 0.898 & 0.960 & 0.935 & 0.742 & 0.932 & 1.000 & 0.930 & 0.725 \\ 
   &  &    8 & 0.935 & 0.625 & 0.940 & 0.653 & 0.900 & 0.886 & 1.000 & 0.898 & 0.672 \\ 
   &  &    9 & 0.941 & 0.595 & 0.977 & 0.476 & 0.942 & 0.854 & 1.000 & 0.906 & 0.680 \\ 
   &  &   10 & 0.871 & 0.885 & 0.805 & 0.811 & 0.823 & 0.900 & 1.000 & 0.901 & 0.792 \\ 
   &  &   11 & 0.492 & 0.953 & 0.976 & 0.943 & 0.513 & 0.896 & 1.000 & 0.891 & 0.597 \\ 
   &  &   12 & 0.466 & 0.927 & 0.929 & 0.909 & 0.785 & 0.915 & 1.000 & 0.902 & 0.805 \\ 
\end{tabular}
\caption{Simulated coverage of 90\%-confidence intervals of contemporaneous ($k = 0$) seasonal structural impulse responses generated by aggregate demand $(ad)$, aggregate supply ($as$) and monetary policy $(mp)$ shocks on industrial production ($y^{ip}$), inflation $(\pi)$ and federal funds rate ($i$). Structural schocks are generated using GARCH specification G2.} 
\label{Cov_b1para2_k0}
\end{table}

\begin{table}[t!]
\centering
\scriptsize
\begin{tabular}{l|l|l|c|ccccccccc}
  & $b$&$N$ & $s$ & $ ad \to y^{ip}$ & $ ad \to  \pi$ &  $ ad \to  i$ & $as \to y^{ip}$ & $ as\to \pi$ & $ as \to i$ & $mp \to y^{ip}$ & $ mp \to \pi$ &  $mp \to i$ \\ 
 \hline
 &  &  &    1 & 0.917 & 0.815 & 0.946 & 0.844 & 0.914 & 0.884 & 1.000 & 0.894 & 0.872 \\ 
   &  &  &    2 & 0.764 & 0.868 & 0.896 & 0.890 & 0.827 & 0.893 & 1.000 & 0.885 & 0.896 \\ 
   &  &  &    3 & 0.769 & 0.924 & 0.942 & 0.929 & 0.824 & 0.908 & 1.000 & 0.894 & 0.850 \\ 
   &  &  &    4 & 0.935 & 0.815 & 0.890 & 0.876 & 0.935 & 0.885 & 1.000 & 0.893 & 0.867 \\ 
   &  &  &    5 & 0.756 & 0.916 & 0.883 & 0.912 & 0.899 & 0.900 & 1.000 & 0.883 & 0.860 \\ 
  G1 & 1 & 100 &    6 & 0.924 & 0.930 & 0.924 & 0.915 & 0.788 & 0.909 & 1.000 & 0.889 & 0.876 \\ 
   &  &  &    7 & 0.925 & 0.916 & 0.958 & 0.937 & 0.805 & 0.917 & 1.000 & 0.893 & 0.835 \\ 
   &  &  &    8 & 0.906 & 0.801 & 0.961 & 0.819 & 0.931 & 0.860 & 1.000 & 0.881 & 0.808 \\ 
   &  &  &    9 & 0.885 & 0.759 & 0.940 & 0.722 & 0.931 & 0.850 & 1.000 & 0.895 & 0.831 \\ 
   &  &  &   10 & 0.923 & 0.894 & 0.890 & 0.867 & 0.866 & 0.896 & 1.000 & 0.903 & 0.863 \\ 
   &  &  &   11 & 0.747 & 0.931 & 0.949 & 0.921 & 0.731 & 0.883 & 1.000 & 0.876 & 0.786 \\ 
   &  &  &   12 & 0.747 & 0.912 & 0.931 & 0.915 & 0.866 & 0.905 & 1.000 & 0.875 & 0.863 \\ 
   \hline
  &  &  &    1 & 0.918 & 0.813 & 0.941 & 0.841 & 0.915 & 0.882 & 1.000 & 0.902 & 0.882 \\ 
   &  &  &    2 & 0.761 & 0.864 & 0.893 & 0.885 & 0.833 & 0.891 & 1.000 & 0.887 & 0.895 \\ 
   &  &  &    3 & 0.763 & 0.925 & 0.945 & 0.932 & 0.826 & 0.915 & 1.000 & 0.900 & 0.846 \\ 
   &  &  &    4 & 0.939 & 0.809 & 0.888 & 0.876 & 0.933 & 0.885 & 1.000 & 0.906 & 0.873 \\ 
   &  &  &    5 & 0.757 & 0.916 & 0.883 & 0.913 & 0.900 & 0.902 & 1.000 & 0.880 & 0.866 \\ 
  G1 & 7 & 100 &    6 & 0.925 & 0.933 & 0.931 & 0.917 & 0.784 & 0.906 & 1.000 & 0.889 & 0.871 \\ 
   &  &  &    7 & 0.924 & 0.917 & 0.957 & 0.939 & 0.803 & 0.923 & 1.000 & 0.900 & 0.837 \\ 
   &  &  &    8 & 0.902 & 0.792 & 0.956 & 0.820 & 0.936 & 0.866 & 1.000 & 0.885 & 0.807 \\ 
   &  &  &    9 & 0.884 & 0.766 & 0.946 & 0.717 & 0.930 & 0.855 & 1.000 & 0.897 & 0.829 \\ 
   &  &  &   10 & 0.923 & 0.893 & 0.891 & 0.865 & 0.866 & 0.899 & 1.000 & 0.902 & 0.869 \\ 
   &  &  &   11 & 0.746 & 0.934 & 0.946 & 0.922 & 0.732 & 0.890 & 1.000 & 0.885 & 0.786 \\ 
   &  &  &   12 & 0.746 & 0.910 & 0.930 & 0.918 & 0.866 & 0.899 & 1.000 & 0.890 & 0.871 \\ 
   \hline
  &  &  &    1 & 0.881 & 0.763 & 0.945 & 0.796 & 0.860 & 0.894 & 1.000 & 0.877 & 0.804 \\ 
   &  &  &    2 & 0.637 & 0.859 & 0.868 & 0.879 & 0.767 & 0.883 & 1.000 & 0.890 & 0.815 \\ 
   &  &  &    3 & 0.615 & 0.913 & 0.907 & 0.913 & 0.758 & 0.887 & 1.000 & 0.885 & 0.780 \\ 
   &  &  &    4 & 0.884 & 0.762 & 0.832 & 0.842 & 0.905 & 0.875 & 1.000 & 0.870 & 0.791 \\ 
   &  &  &    5 & 0.622 & 0.927 & 0.840 & 0.921 & 0.895 & 0.834 & 1.000 & 0.748 & 0.813 \\ 
  G2 & 1 & 100 &    6 & 0.801 & 0.924 & 0.920 & 0.898 & 0.709 & 0.912 & 1.000 & 0.874 & 0.790 \\ 
   &  &  &    7 & 0.860 & 0.911 & 0.938 & 0.920 & 0.771 & 0.922 & 1.000 & 0.895 & 0.762 \\ 
   &  &  &    8 & 0.876 & 0.734 & 0.929 & 0.746 & 0.904 & 0.878 & 1.000 & 0.887 & 0.729 \\ 
   &  &  &    9 & 0.889 & 0.661 & 0.928 & 0.611 & 0.937 & 0.857 & 1.000 & 0.888 & 0.733 \\ 
   &  &  &   10 & 0.846 & 0.878 & 0.808 & 0.842 & 0.824 & 0.892 & 1.000 & 0.881 & 0.805 \\ 
   &  &  &   11 & 0.581 & 0.900 & 0.946 & 0.905 & 0.632 & 0.901 & 1.000 & 0.900 & 0.641 \\ 
   &  &  &   12 & 0.594 & 0.902 & 0.910 & 0.910 & 0.805 & 0.891 & 1.000 & 0.874 & 0.809 \\   
  \hline
  &  &  &    1 & 0.889 & 0.768 & 0.952 & 0.802 & 0.868 & 0.901 & 1.000 & 0.893 & 0.811 \\ 
   &  &  &    2 & 0.628 & 0.865 & 0.875 & 0.883 & 0.750 & 0.886 & 1.000 & 0.896 & 0.825 \\ 
   &  &  &    3 & 0.608 & 0.924 & 0.918 & 0.925 & 0.756 & 0.884 & 1.000 & 0.894 & 0.784 \\ 
   &  &  &    4 & 0.893 & 0.768 & 0.842 & 0.846 & 0.909 & 0.872 & 1.000 & 0.877 & 0.809 \\ 
   &  &  &    5 & 0.613 & 0.930 & 0.844 & 0.921 & 0.901 & 0.828 & 1.000 & 0.753 & 0.827 \\ 
  G2 & 7 & 100 &    6 & 0.803 & 0.938 & 0.934 & 0.909 & 0.703 & 0.915 & 1.000 & 0.875 & 0.804 \\ 
   &  &  &    7 & 0.859 & 0.923 & 0.945 & 0.936 & 0.777 & 0.933 & 1.000 & 0.908 & 0.778 \\ 
   &  &  &    8 & 0.880 & 0.743 & 0.929 & 0.747 & 0.909 & 0.879 & 1.000 & 0.896 & 0.738 \\ 
   &  &  &    9 & 0.889 & 0.657 & 0.938 & 0.595 & 0.940 & 0.862 & 1.000 & 0.893 & 0.735 \\ 
   &  &  &   10 & 0.848 & 0.875 & 0.813 & 0.843 & 0.827 & 0.903 & 1.000 & 0.882 & 0.809 \\ 
   &  &  &   11 & 0.571 & 0.912 & 0.943 & 0.920 & 0.630 & 0.908 & 1.000 & 0.902 & 0.644 \\ 
   &  &  &   12 & 0.584 & 0.916 & 0.913 & 0.923 & 0.801 & 0.902 & 1.000 & 0.884 & 0.804 \\  
\end{tabular}
\caption{Simulated coverage of 90\%-confidence intervals of contemporaneous ($k = 0$) seasonal structural impulse responses generated by aggregate demand $(ad)$, aggregate supply ($as$) and monetary policy $(mp)$ shocks on industrial production ($y^{ip}$), inflation $(\pi)$ and federal funds rate ($i$). Structural schocks are generated using GARCH specification G1 and G2.} 
\label{Cov_b7para1_k0}
\end{table}

\begin{table}[t!]
\centering
\scriptsize
\begin{tabular}{l|l|c|ccccccccc}
$b$&$ N$ & $s$ & $ ad \to y^{ip}$ & $ ad \to  \pi$ &  $ ad \to  i$ & $as \to y^{ip}$ & $ as\to \pi$ & $ as \to i$ & $mp \to y^{ip}$ & $ mp \to \pi$ &  $mp \to i$ \\ 
  \hline
 &  &    1 & 0.829 & 0.756 & 0.727 & 0.715 & 0.846 & 0.965 & 0.978 & 0.992 & 0.921 \\ 
   &  &    2 & 0.830 & 0.897 & 0.707 & 0.988 & 0.761 & 0.952 & 0.985 & 0.991 & 0.935 \\ 
   &  &    3 & 0.725 & 0.956 & 0.889 & 0.979 & 0.666 & 0.805 & 0.940 & 0.988 & 0.901 \\ 
   &  &    4 & 0.766 & 0.547 & 0.502 & 0.847 & 0.972 & 0.970 & 0.955 & 0.992 & 0.931 \\ 
   &  &    5 & 0.971 & 0.988 & 0.707 & 0.944 & 0.816 & 0.771 & 0.957 & 0.954 & 0.933 \\ 
  1 & 20  &    6 & 0.857 & 0.990 & 0.931 & 0.966 & 0.743 & 0.946 & 0.957 & 0.989 & 0.925 \\ 
   &  &    7 & 0.739 & 0.902 & 0.726 & 0.940 & 0.887 & 0.977 & 0.910 & 0.966 & 0.810 \\ 
   &  &    8 & 0.783 & 0.383 & 0.597 & 0.647 & 0.987 & 0.931 & 0.944 & 0.974 & 0.850 \\ 
   &  &    9 & 0.973 & 0.486 & 0.846 & 0.414 & 0.948 & 0.732 & 0.940 & 0.983 & 0.849 \\ 
   &  &   10 & 0.819 & 0.847 & 0.582 & 0.746 & 0.927 & 0.980 & 0.936 & 0.986 & 0.828 \\ 
   &  &   11 & 0.543 & 0.993 & 0.843 & 0.990 & 0.561 & 0.960 & 0.935 & 0.958 & 0.758 \\ 
   &  &   12 & 0.583 & 0.879 & 0.637 & 0.986 & 0.904 & 0.966 & 0.948 & 0.953 & 0.825 \\ 
   \hline
 &  &    1 & 0.822 & 0.749 & 0.729 & 0.710 & 0.851 & 0.961 & 0.976 & 0.990 & 0.914 \\ 
   &  &    2 & 0.827 & 0.896 & 0.708 & 0.990 & 0.764 & 0.959 & 0.982 & 0.991 & 0.944 \\ 
   &  &    3 & 0.716 & 0.954 & 0.888 & 0.982 & 0.662 & 0.814 & 0.938 & 0.991 & 0.903 \\ 
   &  &    4 & 0.760 & 0.533 & 0.499 & 0.842 & 0.972 & 0.966 & 0.956 & 0.994 & 0.932 \\ 
   &  &    5 & 0.974 & 0.989 & 0.705 & 0.942 & 0.808 & 0.780 & 0.948 & 0.946 & 0.931 \\ 
 3  & 20 &    6 & 0.855 & 0.989 & 0.931 & 0.966 & 0.733 & 0.947 & 0.951 & 0.993 & 0.925 \\ 
   &  &    7 & 0.737 & 0.909 & 0.723 & 0.933 & 0.891 & 0.978 & 0.910 & 0.966 & 0.821 \\ 
   &  &    8 & 0.783 & 0.368 & 0.599 & 0.633 & 0.986 & 0.926 & 0.942 & 0.975 & 0.849 \\ 
   &  &    9 & 0.974 & 0.492 & 0.850 & 0.409 & 0.950 & 0.729 & 0.939 & 0.983 & 0.848 \\ 
   &  &   10 & 0.812 & 0.851 & 0.567 & 0.728 & 0.926 & 0.981 & 0.942 & 0.989 & 0.839 \\ 
   &  &   11 & 0.541 & 0.994 & 0.844 & 0.989 & 0.563 & 0.959 & 0.939 & 0.963 & 0.762 \\ 
   &  &   12 & 0.580 & 0.878 & 0.628 & 0.987 & 0.905 & 0.965 & 0.948 & 0.958 & 0.823 \\ 
   \hline
  &  &    1 & 0.894 & 0.780 & 0.840 & 0.774 & 0.884 & 0.895 & 0.928 & 0.923 & 0.849 \\ 
   &  &    2 & 0.884 & 0.858 & 0.783 & 0.945 & 0.874 & 0.936 & 0.923 & 0.924 & 0.852 \\ 
   &  &    3 & 0.788 & 0.950 & 0.918 & 0.934 & 0.768 & 0.863 & 0.906 & 0.920 & 0.811 \\ 
   &  &    4 & 0.857 & 0.686 & 0.717 & 0.824 & 0.926 & 0.866 & 0.894 & 0.931 & 0.837 \\ 
   &  &    5 & 0.919 & 0.908 & 0.789 & 0.911 & 0.828 & 0.860 & 0.903 & 0.950 & 0.820 \\ 
  1 & 50 &    6 & 0.889 & 0.950 & 0.926 & 0.932 & 0.804 & 0.937 & 0.917 & 0.902 & 0.820 \\ 
   &  &    7 & 0.859 & 0.914 & 0.799 & 0.925 & 0.896 & 0.948 & 0.885 & 0.861 & 0.779 \\ 
   &  &    8 & 0.917 & 0.575 & 0.784 & 0.758 & 0.923 & 0.861 & 0.909 & 0.856 & 0.787 \\ 
   &  &    9 & 0.944 & 0.731 & 0.958 & 0.635 & 0.883 & 0.805 & 0.903 & 0.903 & 0.790 \\ 
   &  &   10 & 0.906 & 0.894 & 0.763 & 0.797 & 0.895 & 0.890 & 0.887 & 0.897 & 0.805 \\ 
   &  &   11 & 0.692 & 0.944 & 0.920 & 0.942 & 0.680 & 0.956 & 0.876 & 0.863 & 0.754 \\ 
   &  &   12 & 0.704 & 0.908 & 0.812 & 0.946 & 0.896 & 0.934 & 0.877 & 0.863 & 0.803 \\  
  \hline
 &  &    1 & 0.892 & 0.780 & 0.841 & 0.773 & 0.885 & 0.897 & 0.929 & 0.926 & 0.853 \\ 
   &  &    2 & 0.879 & 0.864 & 0.779 & 0.943 & 0.872 & 0.941 & 0.915 & 0.924 & 0.861 \\ 
   &  &    3 & 0.774 & 0.950 & 0.906 & 0.934 & 0.769 & 0.863 & 0.905 & 0.924 & 0.815 \\ 
   &  &    4 & 0.857 & 0.687 & 0.715 & 0.830 & 0.930 & 0.865 & 0.898 & 0.931 & 0.842 \\ 
   &  &    5 & 0.913 & 0.911 & 0.789 & 0.920 & 0.830 & 0.860 & 0.906 & 0.947 & 0.817 \\ 
  5 & 50 &    6 & 0.890 & 0.957 & 0.926 & 0.932 & 0.805 & 0.941 & 0.923 & 0.911 & 0.825 \\ 
   &  &    7 & 0.861 & 0.914 & 0.804 & 0.920 & 0.891 & 0.955 & 0.890 & 0.866 & 0.780 \\ 
   &  &    8 & 0.922 & 0.557 & 0.780 & 0.753 & 0.925 & 0.860 & 0.920 & 0.853 & 0.780 \\ 
   &  &    9 & 0.947 & 0.718 & 0.956 & 0.622 & 0.878 & 0.810 & 0.910 & 0.901 & 0.798 \\ 
   &  &   10 & 0.909 & 0.885 & 0.758 & 0.793 & 0.900 & 0.898 & 0.881 & 0.897 & 0.809 \\ 
   &  &   11 & 0.686 & 0.938 & 0.919 & 0.950 & 0.666 & 0.956 & 0.882 & 0.862 & 0.745 \\ 
   &  &   12 & 0.691 & 0.904 & 0.813 & 0.951 & 0.893 & 0.934 & 0.879 & 0.861 & 0.803 \\ 
\end{tabular}
\caption{Simulated coverage of 90\%-confidence intervals of seasonal structural impulse response of industrial production ($y^{ip}$), inflation $(\pi)$ and federal funds rate ($i$) $k=12$ months after a shock to aggregate demand $(ad)$, aggregate supply $(as)$ and monetary policy ($mp$). Structural schocks are generated using GARCH specification G1.}
\label{Cov_b1para1_k12}
\end{table}

\begin{table}[t!]
\centering
\scriptsize
\begin{tabular}{l|l|c|ccccccccc}
$b$&$ N$ & $s$ & $ ad \to y^{ip}$ & $ ad \to  \pi$ &  $ ad \to  i$ & $as \to y^{ip}$ & $ as\to \pi$ & $ as \to i$ & $mp \to y^{ip}$ & $ mp \to \pi$ &  $mp \to i$ \\ 
  \hline
  &  &    1 & 0.798 & 0.704 & 0.653 & 0.668 & 0.807 & 0.966 & 0.974 & 0.991 & 0.929 \\ 
   &  &    2 & 0.798 & 0.871 & 0.646 & 0.976 & 0.683 & 0.937 & 0.973 & 0.996 & 0.943 \\ 
   &  &    3 & 0.664 & 0.927 & 0.845 & 0.986 & 0.587 & 0.748 & 0.926 & 0.989 & 0.906 \\ 
   &  &    4 & 0.717 & 0.484 & 0.434 & 0.805 & 0.966 & 0.971 & 0.941 & 0.989 & 0.948 \\ 
   &  &    5 & 0.947 & 0.985 & 0.638 & 0.919 & 0.729 & 0.731 & 0.938 & 0.931 & 0.945 \\ 
 1  & 20 &    6 & 0.828 & 0.986 & 0.898 & 0.952 & 0.668 & 0.927 & 0.942 & 0.995 & 0.934 \\ 
   &  &    7 & 0.672 & 0.855 & 0.665 & 0.901 & 0.846 & 0.975 & 0.886 & 0.976 & 0.829 \\ 
   &  &    8 & 0.738 & 0.319 & 0.513 & 0.592 & 0.981 & 0.930 & 0.935 & 0.979 & 0.852 \\ 
   &  &    9 & 0.966 & 0.419 & 0.802 & 0.334 & 0.932 & 0.684 & 0.924 & 0.983 & 0.849 \\ 
   &  &   10 & 0.743 & 0.804 & 0.505 & 0.653 & 0.893 & 0.982 & 0.926 & 0.987 & 0.844 \\ 
   &  &   11 & 0.481 & 0.993 & 0.779 & 0.984 & 0.499 & 0.928 & 0.924 & 0.975 & 0.738 \\ 
   &  &   12 & 0.501 & 0.821 & 0.535 & 0.979 & 0.869 & 0.948 & 0.934 & 0.972 & 0.813 \\ 
   \hline
 &  &    1 & 0.794 & 0.705 & 0.648 & 0.666 & 0.809 & 0.971 & 0.975 & 0.991 & 0.922 \\ 
   &  &    2 & 0.790 & 0.874 & 0.639 & 0.977 & 0.680 & 0.932 & 0.972 & 0.995 & 0.946 \\ 
   &  &    3 & 0.662 & 0.938 & 0.840 & 0.987 & 0.579 & 0.749 & 0.920 & 0.992 & 0.906 \\ 
   &  &    4 & 0.711 & 0.473 & 0.433 & 0.797 & 0.967 & 0.973 & 0.942 & 0.993 & 0.949 \\ 
   &  &    5 & 0.954 & 0.982 & 0.625 & 0.920 & 0.731 & 0.727 & 0.932 & 0.932 & 0.948 \\ 
  3 & 20 &    6 & 0.825 & 0.988 & 0.902 & 0.950 & 0.666 & 0.924 & 0.943 & 0.995 & 0.937 \\ 
   &  &    7 & 0.667 & 0.854 & 0.656 & 0.903 & 0.842 & 0.975 & 0.891 & 0.979 & 0.833 \\ 
   &  &    8 & 0.743 & 0.315 & 0.514 & 0.581 & 0.975 & 0.926 & 0.934 & 0.980 & 0.860 \\ 
   &  &    9 & 0.966 & 0.415 & 0.801 & 0.330 & 0.939 & 0.683 & 0.921 & 0.983 & 0.850 \\ 
   &  &   10 & 0.743 & 0.812 & 0.504 & 0.646 & 0.896 & 0.981 & 0.923 & 0.987 & 0.847 \\ 
   &  &   11 & 0.477 & 0.995 & 0.776 & 0.984 & 0.499 & 0.926 & 0.926 & 0.975 & 0.750 \\ 
   &  &   12 & 0.495 & 0.814 & 0.537 & 0.978 & 0.866 & 0.945 & 0.940 & 0.975 & 0.821 \\ 
   \hline
  &  &    1 & 0.849 & 0.699 & 0.734 & 0.698 & 0.817 & 0.883 & 0.910 & 0.927 & 0.833 \\ 
   &  &    2 & 0.814 & 0.829 & 0.710 & 0.918 & 0.805 & 0.928 & 0.903 & 0.930 & 0.835 \\ 
   &  &    3 & 0.705 & 0.917 & 0.883 & 0.926 & 0.672 & 0.813 & 0.863 & 0.914 & 0.784 \\ 
   &  &    4 & 0.820 & 0.606 & 0.602 & 0.785 & 0.893 & 0.843 & 0.874 & 0.922 & 0.819 \\ 
   &  &    5 & 0.872 & 0.901 & 0.739 & 0.861 & 0.751 & 0.798 & 0.862 & 0.913 & 0.806 \\ 
  1 & 50  &    6 & 0.825 & 0.935 & 0.869 & 0.895 & 0.717 & 0.906 & 0.896 & 0.914 & 0.813 \\ 
   &  &    7 & 0.779 & 0.875 & 0.732 & 0.891 & 0.849 & 0.932 & 0.842 & 0.884 & 0.762 \\ 
   &  &    8 & 0.894 & 0.455 & 0.680 & 0.696 & 0.911 & 0.842 & 0.897 & 0.871 & 0.758 \\ 
   &  &    9 & 0.942 & 0.596 & 0.922 & 0.533 & 0.884 & 0.757 & 0.878 & 0.885 & 0.746 \\ 
   &  &   10 & 0.867 & 0.827 & 0.661 & 0.737 & 0.858 & 0.886 & 0.859 & 0.898 & 0.773 \\ 
   &  &   11 & 0.574 & 0.933 & 0.851 & 0.932 & 0.574 & 0.912 & 0.851 & 0.866 & 0.710 \\ 
   &  &   12 & 0.606 & 0.843 & 0.700 & 0.926 & 0.823 & 0.900 & 0.859 & 0.864 & 0.763 \\ 
   \hline
   &  &    1 & 0.850 & 0.709 & 0.736 & 0.705 & 0.820 & 0.889 & 0.910 & 0.931 & 0.841 \\ 
   &  &    2 & 0.817 & 0.823 & 0.704 & 0.919 & 0.798 & 0.937 & 0.906 & 0.928 & 0.847 \\ 
   &  &    3 & 0.710 & 0.923 & 0.887 & 0.931 & 0.673 & 0.818 & 0.865 & 0.921 & 0.798 \\ 
   &  &    4 & 0.810 & 0.608 & 0.586 & 0.784 & 0.901 & 0.850 & 0.880 & 0.924 & 0.823 \\ 
   &  &    5 & 0.883 & 0.897 & 0.734 & 0.867 & 0.754 & 0.807 & 0.870 & 0.919 & 0.816 \\ 
 5  & 50 &    6 & 0.823 & 0.939 & 0.873 & 0.896 & 0.719 & 0.907 & 0.906 & 0.914 & 0.820 \\ 
   &  &    7 & 0.785 & 0.871 & 0.735 & 0.900 & 0.862 & 0.941 & 0.846 & 0.888 & 0.766 \\ 
   &  &    8 & 0.890 & 0.455 & 0.683 & 0.705 & 0.925 & 0.854 & 0.900 & 0.877 & 0.767 \\ 
   &  &    9 & 0.947 & 0.590 & 0.926 & 0.533 & 0.878 & 0.754 & 0.885 & 0.890 & 0.761 \\ 
   &  &   10 & 0.872 & 0.828 & 0.667 & 0.743 & 0.861 & 0.887 & 0.864 & 0.901 & 0.770 \\ 
   &  &   11 & 0.561 & 0.943 & 0.851 & 0.932 & 0.572 & 0.923 & 0.857 & 0.864 & 0.722 \\ 
   &  &   12 & 0.605 & 0.833 & 0.701 & 0.929 & 0.829 & 0.911 & 0.866 & 0.867 & 0.778 \\ 
\end{tabular}
\caption{Simulated coverage of 90\%-confidence intervals of seasonal structural impulse response of industrial production ($y^{ip}$), inflation $(\pi)$ and federal funds rate ($i$) $k=12$ months after a shock to aggregate demand $(ad)$, aggregate supply $(as)$ and monetary policy ($mp$). Structural schocks are generated using GARCH specification G2}
\label{Cov_b1para2_k12}
\end{table}

\begin{table}[t!]
\centering
\scriptsize
\begin{tabular}{l|l|l|c|ccccccccc}
& $b$&$ N$ & $s$ & $ ad \to y^{ip}$ & $ ad \to  \pi$ &  $ ad \to  i$ & $as \to y^{ip}$ & $ as\to \pi$ & $ as \to i$ & $mp \to y^{ip}$ & $ mp \to \pi$ &  $mp \to i$ \\ 
   \hline
 &  &  &    1 & 0.924 & 0.842 & 0.883 & 0.815 & 0.903 & 0.888 & 0.917 & 0.895 & 0.858 \\ 
   &  &  &    2 & 0.891 & 0.857 & 0.806 & 0.926 & 0.892 & 0.900 & 0.914 & 0.893 & 0.852 \\ 
   &  &  &    3 & 0.837 & 0.911 & 0.907 & 0.929 & 0.821 & 0.887 & 0.901 & 0.908 & 0.829 \\ 
   &  &  &    4 & 0.913 & 0.783 & 0.797 & 0.865 & 0.912 & 0.886 & 0.905 & 0.900 & 0.833 \\ 
   &  &  &    5 & 0.925 & 0.904 & 0.835 & 0.895 & 0.882 & 0.896 & 0.893 & 0.926 & 0.832 \\ 
  G1 & 1 & 100  &    6 & 0.891 & 0.922 & 0.907 & 0.915 & 0.867 & 0.931 & 0.913 & 0.889 & 0.850 \\ 
   &  &  &    7 & 0.909 & 0.924 & 0.844 & 0.914 & 0.911 & 0.938 & 0.896 & 0.868 & 0.821 \\ 
   &  &  &    8 & 0.924 & 0.708 & 0.860 & 0.833 & 0.910 & 0.864 & 0.914 & 0.865 & 0.815 \\ 
   &  &  &    9 & 0.897 & 0.810 & 0.959 & 0.750 & 0.871 & 0.826 & 0.900 & 0.889 & 0.823 \\ 
   &  &  &   10 & 0.919 & 0.908 & 0.840 & 0.848 & 0.911 & 0.900 & 0.909 & 0.895 & 0.842 \\ 
   &  &  &   11 & 0.788 & 0.927 & 0.925 & 0.913 & 0.740 & 0.928 & 0.882 & 0.865 & 0.802 \\ 
   &  &  &   12 & 0.794 & 0.916 & 0.858 & 0.933 & 0.932 & 0.921 & 0.896 & 0.867 & 0.822 \\ 
   \hline
  &  &  &    1 & 0.919 & 0.836 & 0.889 & 0.811 & 0.912 & 0.885 & 0.913 & 0.892 & 0.856 \\ 
   &  &  &    2 & 0.888 & 0.860 & 0.815 & 0.927 & 0.886 & 0.900 & 0.919 & 0.898 & 0.849 \\ 
   &  &  &    3 & 0.833 & 0.910 & 0.901 & 0.930 & 0.825 & 0.890 & 0.899 & 0.911 & 0.832 \\ 
   &  &  &    4 & 0.912 & 0.776 & 0.800 & 0.866 & 0.916 & 0.892 & 0.903 & 0.898 & 0.837 \\ 
   &  &  &    5 & 0.921 & 0.904 & 0.822 & 0.907 & 0.883 & 0.889 & 0.895 & 0.930 & 0.827 \\ 
  G1 & 7 & 100 &    6 & 0.884 & 0.922 & 0.912 & 0.917 & 0.866 & 0.931 & 0.919 & 0.894 & 0.845 \\ 
   &  &  &    7 & 0.910 & 0.920 & 0.839 & 0.915 & 0.914 & 0.938 & 0.893 & 0.873 & 0.824 \\ 
   &  &  &    8 & 0.920 & 0.712 & 0.850 & 0.830 & 0.906 & 0.863 & 0.914 & 0.865 & 0.820 \\ 
   &  &  &    9 & 0.892 & 0.812 & 0.965 & 0.754 & 0.874 & 0.833 & 0.902 & 0.890 & 0.824 \\ 
   &  &  &   10 & 0.930 & 0.899 & 0.828 & 0.853 & 0.912 & 0.899 & 0.903 & 0.899 & 0.835 \\ 
   &  &  &   11 & 0.785 & 0.920 & 0.924 & 0.913 & 0.734 & 0.932 & 0.886 & 0.867 & 0.802 \\ 
   &  &  &   12 & 0.798 & 0.911 & 0.853 & 0.930 & 0.932 & 0.917 & 0.893 & 0.866 & 0.812 \\ 
   \hline
  &  &  &    1 & 0.869 & 0.780 & 0.792 & 0.753 & 0.836 & 0.865 & 0.884 & 0.898 & 0.815 \\ 
   &  &  &    2 & 0.827 & 0.830 & 0.740 & 0.887 & 0.832 & 0.892 & 0.890 & 0.900 & 0.838 \\ 
   &  &  &    3 & 0.736 & 0.857 & 0.813 & 0.887 & 0.740 & 0.831 & 0.869 & 0.904 & 0.793 \\ 
   &  &  &    4 & 0.848 & 0.702 & 0.696 & 0.805 & 0.890 & 0.870 & 0.854 & 0.892 & 0.803 \\ 
   &  &  &    5 & 0.867 & 0.875 & 0.765 & 0.817 & 0.795 & 0.819 & 0.866 & 0.879 & 0.792 \\ 
  G2 & 1 & 100  &    6 & 0.805 & 0.888 & 0.805 & 0.860 & 0.772 & 0.904 & 0.878 & 0.878 & 0.793 \\ 
   &  &  &    7 & 0.839 & 0.891 & 0.769 & 0.868 & 0.859 & 0.910 & 0.850 & 0.866 & 0.793 \\ 
   &  &  &    8 & 0.905 & 0.589 & 0.744 & 0.731 & 0.881 & 0.845 & 0.882 & 0.865 & 0.786 \\ 
   &  &  &    9 & 0.884 & 0.687 & 0.914 & 0.647 & 0.861 & 0.775 & 0.875 & 0.900 & 0.792 \\ 
   &  &  &   10 & 0.863 & 0.862 & 0.730 & 0.783 & 0.855 & 0.881 & 0.864 & 0.906 & 0.802 \\ 
   &  &  &   11 & 0.660 & 0.883 & 0.838 & 0.890 & 0.628 & 0.886 & 0.858 & 0.859 & 0.760 \\ 
   &  &  &   12 & 0.684 & 0.869 & 0.773 & 0.907 & 0.868 & 0.919 & 0.877 & 0.854 & 0.782 \\   
   \hline
  &  &  &    1 & 0.883 & 0.781 & 0.804 & 0.766 & 0.857 & 0.884 & 0.890 & 0.904 & 0.826 \\ 
   &  &  &    2 & 0.847 & 0.838 & 0.747 & 0.902 & 0.836 & 0.894 & 0.899 & 0.912 & 0.834 \\ 
   &  &  &    3 & 0.756 & 0.873 & 0.845 & 0.902 & 0.753 & 0.837 & 0.874 & 0.911 & 0.800 \\ 
   &  &  &    4 & 0.853 & 0.718 & 0.698 & 0.822 & 0.896 & 0.877 & 0.860 & 0.905 & 0.815 \\ 
   &  &  &    5 & 0.884 & 0.881 & 0.759 & 0.843 & 0.811 & 0.824 & 0.873 & 0.898 & 0.806 \\ 
  G2 & 7 & 100 &    6 & 0.823 & 0.902 & 0.819 & 0.882 & 0.779 & 0.920 & 0.892 & 0.887 & 0.820 \\ 
   &  &  &    7 & 0.837 & 0.897 & 0.772 & 0.887 & 0.869 & 0.928 & 0.859 & 0.873 & 0.792 \\ 
   &  &  &    8 & 0.916 & 0.589 & 0.753 & 0.755 & 0.901 & 0.853 & 0.887 & 0.870 & 0.790 \\ 
   &  &  &    9 & 0.887 & 0.691 & 0.921 & 0.636 & 0.875 & 0.791 & 0.882 & 0.901 & 0.796 \\ 
   &  &  &   10 & 0.867 & 0.863 & 0.721 & 0.790 & 0.862 & 0.889 & 0.871 & 0.901 & 0.808 \\ 
   &  &  &   11 & 0.667 & 0.894 & 0.851 & 0.902 & 0.636 & 0.902 & 0.866 & 0.871 & 0.761 \\ 
   &  &  &   12 & 0.683 & 0.871 & 0.777 & 0.919 & 0.870 & 0.917 & 0.875 & 0.864 & 0.789 \\ 
\end{tabular}
\caption{Simulated coverage of 90\%-confidence intervals of seasonal structural impulse response of industrial production ($y^{ip}$), inflation $(\pi)$ and federal funds rate ($i$) $k=12$ months after a shock to aggregate demand $(ad)$, aggregate supply $(as)$ and monetary policy ($mp$). Structural schocks are generated using GARCH specification G1 and G2.}
\label{Cov_b7para1_k12}
\end{table}


\end{document}